\documentclass[onecolumn,11pt]{IEEEtran}

\usepackage[margin=1in]{geometry}

\usepackage{epsfig}

\usepackage[update,prepend]{epstopdf}

\usepackage{amsmath,amssymb}
\usepackage[usenames]{color}
\usepackage{umoline}
\usepackage{setspace}
\usepackage[ruled,boxed,linesnumbered]{algorithm2e}

\usepackage{url}
\usepackage{enumerate}
\usepackage{cite}
\usepackage{mathtools}
\usepackage{mathrsfs}
\usepackage{multirow}
\usepackage{datetime}
\usepackage{amsthm}

\usepackage{graphicx}
\usepackage{caption}
\usepackage{subcaption}
\usepackage{array}

\newcommand{\insrt}[1]{{\color{red}#1}}

\newcommand{\norm}[1]{\Vert#1\Vert}

\newcommand{\supp}[1]{\mathrm{supp}(#1)}

\newcommand{\rank}{\mathrm{rank}}

\newcommand{\id}{\text{\rm id}}
\newcommand{\idm}{\id}
\newcommand{\idvr}{I_{n_1}}
\newcommand{\idvc}{I_{n_2}}

\newcommand{\R}{\mathcal{R}}
\newcommand{\A}{\mathcal{A}}
\newcommand{\tnorm}[1]{{\left\vert\kern-0.25ex\left\vert\kern-0.25ex\left\vert #1
    \right\vert\kern-0.25ex\right\vert\kern-0.25ex\right\vert}}

\newtheorem{theorem}{Theorem}[section]
\newtheorem{lemma}[theorem]{Lemma}
\newtheorem{proposition}[theorem]{Proposition}
\newtheorem{corollary}[theorem]{Corollary}

\theoremstyle{definition}
\newtheorem{definition}[theorem]{Definition}
\newtheorem{remark}[theorem]{Remark}

\usepackage{prettyref,soul}
\newrefformat{eq}{(\ref{#1})}
\newrefformat{sec}{Section~\ref{#1}}
\newrefformat{subsec}{Section~\ref{#1}}
\newrefformat{alg}{Algorithm~\ref{#1}}
\newrefformat{fig}{Fig.~\ref{#1}}
\newrefformat{tab}{Table~\ref{#1}}
\newrefformat{rmk}{Remark~\ref{#1}}
\newrefformat{def}{Definition~\ref{#1}}
\newrefformat{cor}{Corollary~\ref{#1}}
\newrefformat{lmm}{Lemma~\ref{#1}}
\newrefformat{prop}{Proposition~\ref{#1}}
\newrefformat{app}{Appendix~\ref{#1}}

\usepackage[
            CJKbookmarks=true,
            bookmarksnumbered=true,
            bookmarksopen=true,
            colorlinks=true,
            citecolor=red,
            linkcolor=blue,
            anchorcolor=red,
            urlcolor=blue,
            pdfauthor={Kiryung Lee, Yihong Wu, and Yoram Bresler},
            pdfstartview=FitH,
            ]{hyperref}

\newcommand{\quant}[1]{{\left\langle #1\right\rangle}}
\newcommand{\var}{\mathsf{var}}
\newcommand{\floor}[1]{{\left\lfloor {#1} \right \rfloor}}
\newcommand{\ceil}[1]{{\left\lceil {#1} \right \rceil}}
\newcommand{\Real}{\mathrm{Re}}

\newcommand{\fnorm}[1]{\|#1\|_{\rm F}}

\newcommand{\iprod}[2]{\left \langle #1, #2 \right\rangle}
\newcommand{\Iprod}[2]{\langle #1, #2 \rangle}

\newcommand{\eg}{e.g.\xspace}
\newcommand{\ie}{i.e.\xspace}
\newcommand{\iid}{i.i.d.\xspace}
\newcommand{\argmin}{\mathop{\rm argmin}}
\newcommand{\argmax}{\mathop{\rm argmax}}
\newcommand{\Expect}{\mathbb{E}}
\newcommand{\expect}[1]{\mathbb{E}\left[ #1 \right]}

\newcommand{\complex}{\mathbb{C}}
\newcommand{\cz}{\mathbb{C}}

\newcommand{\naturals}{{\mathbb{N}}}

\newcommand{\transpose}{\top}
\newcommand{\vect}{\mathrm{vec}}

\newcommand{\eexp}{{\rm e}}

\newcommand{\pth}[1]{\left( #1 \right)}

\newcommand{\tu}{{\tilde{u}}}

\newcommand{\tR}{{\tilde{R}}}

\newcommand{\tT}{{\tilde{T}}}
\newcommand{\tU}{{\tilde{U}}}
\newcommand{\tV}{{\tilde{V}}}
\newcommand{\tW}{{\tilde{W}}}

\newcommand{\calA}{{\mathcal{A}}}

\newcommand{\calC}{{\mathcal{C}}}

\newcommand{\calF}{{\mathcal{F}}}
\newcommand{\calG}{{\mathcal{G}}}
\newcommand{\calH}{{\mathcal{H}}}

\newcommand{\calJ}{{\mathcal{J}}}

\newcommand{\calN}{{\mathcal{N}}}

\newcommand{\calP}{{\mathcal{P}}}

\newcommand{\calS}{{\mathcal{S}}}

\newcommand{\calCN}{{\mathcal{CN}}}

\newcommand{\hR}{{\hat{R}}}

\newcommand{\hT}{{\hat{T}}}
\newcommand{\hU}{{\hat{U}}}
\newcommand{\hV}{{\hat{V}}}
\newcommand{\hW}{{\hat{W}}}
\newcommand{\hX}{{\hat{X}}}
\newcommand{\hY}{{\hat{Y}}}

\newcommand{\bbC}{{\mathbb{C}}}

\newcommand{\bbN}{{\mathbb{N}}}

\newcommand{\bbS}{{\mathbb{S}}}

\newcommand{\sfV}{{\mathsf{V}}}

\newcommand{\vzstar}{v_0^\text{\rm opt}}
\newcommand{\vzth}{v_0^\text{\rm th}}
\usepackage{tikz}
\usetikzlibrary{arrows}
\tikzstyle{int}=[draw, fill=blue!20, minimum size=2em]
\tikzstyle{dot}=[circle, draw, fill=blue!20, minimum size=2em]
\tikzstyle{init} = [pin edge={to-,thin,black}]

\usepackage{tikz-cd}

\title{Near Optimal Compressed Sensing \\ of a Class of Sparse Low-Rank Matrices \\ via Sparse Power Factorization}
\author{Kiryung Lee, Yihong Wu, and Yoram Bresler
\thanks{K. Lee is with the School of Electrical and Computer Engineering, Georgia Institute of Technology, Atlanta, GA 30332. Email: kiryung@ece.gatech.edu.
Y. Wu and Y. Bresler are with Coordinated Science Laboratory and Department of Electrical and Computer Engineering, University of Illinois at Urbana-Champaign, Urbana, IL 61801. Emails: \{yihongwu,ybresler\}@illinois.edu.
This work was supported in part by the National Science Foundation (NSF) under Grant IIS 14-47879.
This paper was presented in part at ITW 2014 \cite{LeeWB2014}.}
}

\begin{document}

\maketitle

\doublespacing

\begin{abstract}
  Compressed sensing of simultaneously sparse and low-rank matrices enables recovery of sparse signals from a few linear measurements of their bilinear form. One important question is how many measurements are needed for a stable reconstruction in the presence of measurement noise. Unlike conventional compressed sensing for sparse vectors, where convex relaxation via the $\ell_1$-norm achieves near optimal performance, for compressed sensing of sparse low-rank matrices, it has been shown recently \cite{oymak2015simultaneously} that convex programmings using the nuclear norm and the mixed norm are highly suboptimal even in the noise-free scenario.

We propose an alternating minimization algorithm called sparse power factorization (SPF) for compressed sensing of sparse rank-one matrices. For a class of signals whose sparse representation coefficients are fast-decaying, SPF achieves stable recovery of the rank-1 matrix formed by their outer product and requires number of measurements within a logarithmic factor of the information-theoretic fundamental limit. For the recovery of general sparse low-rank matrices, we propose subspace-concatenated SPF (SCSPF), which has analogous near optimal performance guarantees to SPF in the rank-1 case. Numerical results show that SPF and SCSPF empirically outperform convex programmings using the best known combinations of mixed norm and nuclear norm.

\end{abstract}


\section{Introduction}
\subsection{Problem statement}

Let $X \in \bbC^{n_1 \times n_2}$ be an unknown rank-$r$ matrix
whose singular value decomposition is written as $X = U \Lambda V^*$,
where $\Lambda \in \mathbb{R}^{r \times r}$ is a strictly positive diagonal matrix, $U \in \bbC^{n_1 \times r}$, and $V \in \bbC^{n_2 \times r}$ satisfy $U^* U = V^* V = I_r$.
We further assume that $X$ is sparse in the following senses:
i) either $U$ or $V$ is row-$s$-sparse, that is, has at most $s$ nonzero rows;\footnote{Without loss of generality, we assume that $U$ is row-$s$-sparse.} or ii) $U$ and $V$ are row-$s_1$-sparse and row-$s_2$-sparse, respectively.

Suppose that the measurement vector $b \in \mathbb{C}^m$ of $X$ is obtained
using a known linear operator $\A: \bbC^{n_1 \times n_2} \to \mathbb{C}^m$ as
\begin{equation}
b = \A(X) + z,
    \label{eq:model}
\end{equation}
where $z \in \bbC^m$ denotes additive noise.

We study the problem of stable reconstruction of the unknown \emph{simultaneously sparse and low-rank} $X$ from the noisy linear measurements $b$.
Our goal is to find a good estimate $\hX$ of $X$ from a minimal number of measurements using a computationally efficient algorithm, which satisfies the following stability criterion:
\begin{equation}
\frac{\fnorm{\hat{X} - X}}{\fnorm{X}} \leq C \cdot \frac{\norm{z}_2}{\norm{\A(X)}_2}
	\label{eq:stability}
\end{equation}
for all $z \in \complex^m$ and some absolute (dimension-independent) constant $C$.
The condition in (\ref{eq:stability}) implies that the normalized reconstruction error is at most a constant factor of the noise-to-signal ratio in the measurements, which, in the absence of noise, automatically implies the perfect reconstruction of $X$.

\subsection{Motivating applications: sparse bilinear inverse problems}

Bilinear inverse problems arise ubiquitously in a variety of areas.
For example, \emph{blind deconvolution} (cf. \cite{ahmed2014blind} and references therein) factors
two input signals given their convolution, which is a bilinear function.
In general, bilinear inverse problems involve various ambiguities and do not admit a unique solution. For example, any bilinear inverse problem suffers from scaling ambiguity and the best result one can get is to identify the solution up to a scalar factor. Besides the fundamental ambiguities that cannot be overcome by any method, it is still challenging to identify the solution up to an appropriate equivalence class. To overcome this difficulty, various sparsity models were introduced and the resulting \emph{sparse} bilinear inverse problem has been shown empirically to admit good solutions in various real-world applications.

In blind deconvolution, signals of interest admit sparse representations \cite{mallat2008wavelet} and these sparse signal models have been exploited in denoising, compression, compressed sensing, etc. The impulse responses of convolution systems in applications have sparse representations too. For example, high definition television (HDTV) channels, hilly terrain delay profiles, and underwater acoustic or reverberant room channels, all have sparse channel coefficients (see \cite{AisAbe2008blind} and references therein). These sparsity models were employed to solve blind deconvolution problems, e.g., in the context of blind echo cancellation \cite{AisAbe2008blind,SudAG2010double}.

In these applications, sparsity models narrow down the solution set, which, in a bilinear inverse problem, is a product of subspaces, by replacing the subspaces by unions of low-dimensional subspaces.
This makes robust reconstruction possible even when subsampling is present, which is often desired in applications such as calibration-free parallel imaging.

Recently, it has been proposed to reformulate bilinear inverse problems as the recovery of a low-rank matrix from its linear measurements, through the so-called ``lifting'' procedure \cite{ahmed2014blind,ChoMit2013}. Ahmed et al. \cite{ahmed2014blind} first introduced this idea to solve the blind deconvolution problem as a matrix-valued linear inverse problem, and further showed that nuclear-norm minimization is nearly optimal for certain random linear operators. On the other hand, Choudhary et al. \cite{ChoMit2013} showed negative results in the setup where both the input and linear operator are deterministic. In a nutshell, in the lifted formulation, one obtains a solution to a bilinear system $f_i(x,y) = b_i$, $i = 1, \cdots, m$ from a low-rank solution to a linear system $\A(X) = b$ in the matrix-valued unknown $X$. In the lifted formulation of blind deconvolution, the unknown matrix $X$ is rank-one. On the other hand, in MIMO channel identification \cite{ding2001blind}, the measurements are given as superpositions of convolutions and therefore the solution to the lifted formulation is low-rank, where the rank is determined as the number of the input channels.

In this paper, we consider the lifted linear inverse problem, although we adopt a \emph{nonconvex} approach. Then, the scaling ambiguity is absorbed into the factorization of the matrix-valued solution to the lifted formulation. In the lifted formulation of the sparse bilinear inverse problem, the unknown $X$ has a sparsity model corresponding to those imposed on the unknowns of the original bilinear problem.

The product of two compatible matrices is also bilinear in the individual matrices;
hence, \emph{matrix factorization} is another bilinear inverse problem (cf. \cite{AhaEB2006ksvd}).
Sparsity models also arise in certain matrix factorization problems.
For example, dictionary learning aims to find a good sparse representation for a given data set, which can be formulated as a matrix factorization problem with a sparsity prior \cite{AhaEB2006ksvd}.
Learning a sparsifying dictionary or transform from compressive measurements \cite{RavBre2011dlmri,GleEld2011blind,ravishankar2015efficient} has a similar flavor, but is a more difficult problem, since fewer equations are available to determine the unknowns. Compressive blind source separation \cite{WuCC2010cbss} is yet another matrix factorization problem that exploits a sparsity prior.
These applications can be naturally formulated as the recovery of a sparse and low-rank matrix from its linear measurements.

\subsection{Related work}

The recovery of a \emph{sparse and low-rank} matrix from its minimal incoherent linear measurements
is a special case of compressed sensing of general low-rank matrices.
Compressed sensing of low-rank matrices without sparsity constraints has been well studied as an extension of compressed sensing of sparse vectors.
Recht et al.~\cite{RecFP2010} presented the analogy between the two problems
and showed that the minimum nuclear norm solution to the linear system given by the measurements
is guaranteed to recover the unknown low-rank matrix under the rank-restricted isometry property.
Greedy recovery algorithms for compressed sensing of sparse vectors and their performance guarantees under the restricted isometry property (RIP) have been also extended to analogous algorithms (e.g., ADMiRA \cite{LeeBre2010admira}, SVP \cite{JaiMekDhi2010SVP}) with corresponding guarantees for compressed sensing of low-rank matrices under the rank-restricted isometry property.
An alternating minimization algorithm called \textit{power factorization} (PF) has been proposed
as a computationally efficient heuristic \cite{HalHer2009pf} for the recovery of general low-rank matrices, and its performance guarantee in terms of the rank-restricted isometry property was presented recently \cite{JaiNS2012}. In particular, for a certain class of sensing systems, it has shown \cite{RecFP2010,CanPla2011oracle,liu2011universal} that $O(n r)$ or slightly more, by a logarithmic factor, measurements suffice for stable recovery of an $n \times n$ matrix of rank-$r$.

When the unknown matrix is not only low-rank but also sparse, the number of compressed sensing measurements required for its recovery is further reduced. Suppose that the unknown $n \times n$ matrix is of rank $r$ and has up to $s$ nonzero rows and up to $s$ nonzero columns, and its noise-free measurements are obtained as the inner products with i.i.d. Gaussian matrices. Oymak et al. \cite{oymak2015simultaneously} showed that by solving a combinatorial optimization problem, exact recovery is guaranteed with $O(\max\{rs,s\log (en/s)\})$ measurements. However, they also showed the following negative result: Combining convex surrogates for multiple nonconvex priors does not improve the recovery performance compared to the case of using just one of the priors via a convex surrogate. More precisely, exact recovery using combinations of the nuclear norm and the $\ell_{1,2}$ norm requires $\Omega(\min\{rn,sn\})$ measurements. This is significantly worse than the sample complexity that can be obtained by solving a combinatorial optimization problem.

One might attempt to modify ADMiRA or SVP to exploit the low-rank and sparsity priors simultaneously.
Unfortunately, the key procedure in these algorithms is to project a given matrix onto a set of low-rank and sparse matrices, which is another challenging open problem, in the sense that there is no algorithm with a performance guarantee for this problem.

A closely-related statistical problem involving both sparsity and low-rankness is \emph{sparse principal component analysis} (SPCA) \cite{JohnstoneLu09,Ma11,Birnbaum12,Vu12,CMW12,CMW13}, which deals with estimating the principal subspaces of a covariance matrix when the singular vectors are sparse. For instance, in the rank-one case, one observes independent samples from $\calN(0,I_p+\lambda vv^*)$ and estimates the leading singular vector $v$ which is known to be sparse a priori. The minimax estimation error of SPCA has been characterized within constant factors in \cite{Birnbaum12,CMW12,Vu12,CMW13}.
However, recently it has been shown \cite{Berthet13} that attaining the minimax rate of SPCA can be reduced to planted clique, an open problem that is believed to be hard. Another related problem is \emph{submatrix detection} or \emph{biclustering} \cite{Shabalin09,BI12}, where $X=\lambda uv^*+ Z$ is observed with $Z$ i.i.d. Gaussian and $u,v$ sparse binary vectors. The goal is to decide whether or not $\lambda=0$. It has been shown that attaining the optimal rate for this problem is computationally hard in a similar sense \cite{MW13b}. The main distinction between these problems and our sparse inverse problem \prettyref{eq:model} is that in the former the low-rank and sparse signals (either the covariance or the mean matrix) are observed directly from noisy samples, whereas in the latter the signal is only observed indirectly, through linear measurements that mix the components. Therefore, our sparse inverse problem \prettyref{eq:model} is harder than the other two problems. However, the relative difficulty arising from indirect access to the measurement through noisy linear measurements can be overcome if the linear operator satisfies the restrict isometry property. All the aforementioned problems are difficult in general. In our problem, under an additional ``peakiness'' assumption, we manage to solve the problem with a provably near optimal performance guarantee.

\subsection{Main contributions}

As discussed earlier, in the existing theory for compressed sensing of simultaneously sparse and low-rank matrices, the best known performance guarantee on the sample complexity for polynomial-time algorithms is significantly worse than the theoretical optimum.\footnote{When we posted the first draft of this work to arXiv \cite{LeeWB2013}, there were no algorithms that achieve a near optimal sample complexity for recovering simultaneously sparse and low-rank matrices. Subsequently, it has been shown that for a special measurement scheme with nested structures, a two-step approach achieves near optimal sample complexity \cite{bahmani2015near,iwen2015robust}. However, designing such nested structured sensing mechanisms are impossible in many applications such as blind deconvolution. Very recently, we have learned that the linear operator arising in blind deconvolution satisfies the RIP \cite{ahmed2016rip}. Therefore, the performance guarantee for SPF in this paper applies not only to compressed sensing with i.i.d. Gaussian measurements but also to blind deconvolution.}
Toward closing this gap, in this paper, we propose a set of recovery algorithms that provide near optimal performance guarantees at low computational cost. Motivated by the lifted formulation of various sparse bilinear inverse problems, we first focus on the rank-one case and will later extend both the algorithms and the analysis to the rank-$r$ case.

We propose an alternating minimization algorithm called \textit{sparse power factorization} (SPF),
which reconstructs the unknown sparse rank-one matrix from its linear measurements. SPF is obtained by modifying the updates of estimates of $u$ and $v$ in PF (see \prettyref{sec:PF-review} for a summary) to exploit their sparsity priors. In principle, any algorithm for recovery of sparse vectors from linear measurements can be employed for these steps. In this paper, we focus on a specific procedure called hard thresholding pursuit (HTP) \cite{Fou2011htp}. For recovering sparse vectors and under RIP assumptions, HTP provides performance guarantees for both estimation error and convergence rate, which can be further generalized in the presence of noisy measurements.
Exploiting these guarantees for HTP, we show that the iterative updates in SPF converge linearly under RIP assumption.

Like most alternating minimization methods, the empirical performance of PF and SPF depends crucially on the initial values. Furthermore, to obtain a provable guarantee, it is important to design the initialization procedure carefully. Let $\A^*$ denote the adjoint operator of $\A$. Jain et al. \cite{JaiNS2012} showed that PF initialized by the leading right singular vector of the proxy matrix $\A^*(b)$ \cite{NeeTro09} provides stable recovery of a rank-$r$ matrix under the rank-$2r$ RIP. In particular, if the unknown $n_1 \times n_2$ matrix is rank-one, then their guarantee holds with $O(\max\{n_1,n_2\})$ i.i.d. Gaussian measurements.

When the unknown $n_1 \times n_2$ matrix is row-$s$-sparse, the initialization needs to be modified accordingly. We propose to initialize the SPF algorithm by the leading right singular vector of a submatrix of $\A^*(b)$ whose rows are restricted to an estimated row-support of the left singular vector $u$, SPF initialized with a good approximation on either the left or the right singular vector provides near-optimal performance guarantee whenever the linear operator satisfies RIP for rank-2 and row-$3s$-sparse matrices, which holds with $O((s+n_2) \log (e n_1/s))$ i.i.d. Gaussian measurements.
In particular, when the entries of $u$ are fast-decaying, which is often satisfied by signal models in practice, we show that a simple thresholding algorithm provides such a good initialization.

In the case when the unknown $n \times n$ matrix is doubly-$s$-sparse (both row and column sparse), similarly to the previous case, SPF initialized with a good approximation on either $u$ or $v$ has a performance guarantee under the rank-2 and $(3s,3s)$-sparse RIP, which holds with $O(s \log (e n/s))$ i.i.d. Gaussian measurements, and significantly improves on the guarantee for PF. In particular, under extra decay conditions on the nonzero entries of $u$ and $v$, we show that a simple thresholding algorithm produces a desired good initialization for SPF.

Next, for the sparse and rank-$r$ matrices with $r>1$, we extend the SPF algorithm and its performance guarantees accordingly. The generalization is non-trivial in the sense that it is unclear whether the straightforward rank-$r$ extension of the SPF algorithm can be guaranteed to recover the unknown rank-$r$ and doubly-$s$-sparse matrix from $O(rs \log n)$ measurements. More specifically, the number of measurements for a performance guarantee depends on a power of the rank $r$ rather than linearly on $r$. To fix this, we considered a variation of SPF called \emph{subspace-concatenated SPF} (SCSPF) which provably achieves a sample complexity that scales linearly in the rank $r$. Again, similarly to the rank-one cases, the success of cheap initialization requires extra technical conditions that are analogous to the fast decay properties.

Even in the absence of sparsity where simple initialization works, our results improve the state of the art for recovering low-rank matrices using alternating minimization. Specifically, for rank-$r$ matrices with conditioning number at most $\kappa$, we show that SCSPF succeeds with
$m = O(\kappa^2 r n)$ measurements, which significantly improves on the previous result of $m = O(\kappa^4 r^3 n)$ \cite{JaiNS2012}.

How close is the performance of the SPF algorithm to optimality? We show that stable recovery of sparse rank-one matrices in the sense of \prettyref{eq:stability} requires at least $(s_1+s_2-3r/2)r$ measurements, where $s_1$ and $s_2$ are the row-sparsity levels of the left and right factor $U$ and $V$, respectively. Note that this lower bound coincides with the number of degrees of freedom in the singular vectors. While the parameter-counting argument is heuristic, our converse is obtained via information-theoretic arguments, which provide necessary conditions for stable recovery by any reconstruction method from any measurement mechanism -- linear or not. It follows that our performance guarantees for SPF are \emph{near-optimal} in the sense that SPF achieves robust reconstruction with a number of measurements that is within at most a logarithmic factor of the fundamental limit. Similar near-optimal guarantees for the structured (rather than just random Gaussian) measurements that arise in practical applications are presented in a companion paper \cite{LeeLJB2015}.

In addition to its near-optimal theoretic guarantees,
SPF also outperforms competing convex approaches in the following practical aspects:
\begin{itemize}
  \item SPF requires vastly less memory, since it solves the bilinear formulation with an explicit rank-$r$ factorization of the unknown with $r(n_1+n_2)$ variables. In contrast, the linear formulation solves an optimization problem with $n_1 n_2$ variables.
  \item SPF has lower computational cost. The SPF algorithm converges \emph{superlinearly} fast and each of the sparse recovery steps (inner iteration using HTP) converges in $O(s)$ iterations. Furthermore, each iteration is fast because it only updates $r n_1$ or $r n_2$ variables instead of $n_1 n_2$ variables. Also note that these guarantees are derived for the initialization method that only involves simple thresholding on the row and column norms of the $n_1 \times n_2$ matrix $\A^*(b)$ and the {truncated} singular value decomposition of a reduced $s_1 \times s_2$ matrix up to the first $r$ factors.
  \item As demonstrated in Section~\ref{sec:numres}, the empirical performance of SPF is significantly better than that of convex approaches. In fact, extensive numerical experiments suggest that the performance guarantee of SPF/SCSPF continues to hold even in the absence of the technical assumptions (e.g., sufficiently high SNR and fast decaying magnitudes). Therefore, we suspect that these technical conditions are just artifacts in the proofs.
\end{itemize}

\subsection{Organization}
	\label{sec:org}

The sparse power factorization algorithms are described in detail in Section~\ref{sec:alg},
followed by their performance guarantees in Section~\ref{sec:ub}.
The extension of both algorithms and performance guarantees to the general rank-$r$ case is presented in Section~\ref{sec:rankr}.
An information-theoretic lower bound on the number of measurements for stable recovery of sparse rank-one matrices is given in \prettyref{sec:lb}.
After reporting on the empirical performance of sparse power factorization algorithms in Section~\ref{sec:numres},
we conclude the paper in Section~\ref{sec:conclusion}.
Proofs of the main results are given in \prettyref{sec:pf}, with proofs of several technical lemmas deferred to the appendix.

\subsection{Notations}
	\label{sec:notation}

Let $\bbN=\{1,2,\cdots\}$ denote the set of natural numbers and $[n] \triangleq \{1, \ldots, n\}$ for $n \in \bbN$.
For a complex vector $x \in \bbC^n$, its $k$th element is denoted by $[x]_k$
and the element-wise complex conjugate of $x$ is denoted by $\overline{x}$.
The identity operator on $\bbC^{n_1 \times n_2}$ is denoted as ``$\idm$''.
The Frobenius norm, the spectral norm, and the Hermitian transpose of $X \in \bbC^{n_1 \times n_2}$ are denoted by $\fnorm{X}$, $\norm{X}$, and $X^*$, respectively.
The matrix inner product is defined by $\langle A, B\rangle = \text{trace}(A^*B)$.
For a linear operator $\A$ between two vector spaces, the range space is denoted by $\R(\A)$ and the adjoint operator of $\A$ is denoted by $\A^*$ such that $\iprod{\calA x}{y}=\iprod{x}{\calA^*y}$ for all $x$ and $y$.

For a subspace $\mathcal{S}$ of $\bbC^n$, let $P_{\mathcal{S}} \in \bbC^{n \times n}$ denote the orthogonal projection onto $\mathcal{S}$.
The coordinate projection $\Pi_J \in \bbC^{n \times n}$ is defined by
\begin{equation}
[\Pi_J x]_k =
\begin{cases}
[x]_k & \text{if $k \in J$} \\
0 & \text{else}
\end{cases}
\label{eq:def:coorproj}
\end{equation}
for $J \subset [n]$.
Then, $\Pi_J^\perp \in \bbC^{n \times n}$ is defined as $I_n - \Pi_J$ where $I_n$ is the $n \times n$ identity matrix.

\section{Sparse Power Factorization Algorithms}
\label{sec:alg}
In this section, we present alternating minimization algorithms for compressed sensing of sparse rank-one matrices. To describe these algorithms, we first introduce linear operators that
describe the restrictions of the linear operator $\A : \bbC^{n_1 \times n_2} \to \bbC^m$
acting on rank-one matrix $xy^* \in \bbC^{n_1 \times n_2}$ when either $x$ or $y$ are fixed.

In linear sensing schemes each measurement amounts to a matrix inner product. Indeed, there exist matrices $(M_\ell)_{\ell=1}^m \subset \bbC^{n_1 \times n_2}$ that describe the action of $\A$ on $Z \in \bbC^{n_1 \times n_2}$ and that of its adjoint $\A^*$ on $z = [z_1,\dots,z_m]^\transpose \in \bbC^m$ by
\begin{equation}
\label{eq:defcalA}
\A(Z) = [\langle M_1, Z \rangle, \ldots, \langle M_m, Z \rangle]^\transpose
\end{equation}
and
\begin{equation}
\A^*(z) = \sum_{\ell=1}^m z_\ell M_\ell,	
	\label{eq:Aadjoint}
\end{equation}
respectively. Using $(M_\ell)_{\ell=1}^m$, we define linear operators $F: \bbC^{n_2} \to \bbC^{m \times n_1}$ and $G: \bbC^{n_1} \to \bbC^{m \times n_2}$ by
\begin{equation}
F(y) \triangleq \begin{bmatrix} y^* M_1^* \\ y^* M_2^* \\ \vdots \\ y^* M_m^* \end{bmatrix}
\quad \text{and} \quad
G(x) \triangleq \begin{bmatrix} x^* M_1 \\ x^* M_2 \\ \vdots \\ x^* M_m \end{bmatrix},
\label{eq:defFnG}
\end{equation}
respectively, for $y \in \bbC^{n_2}$ and $x \in \bbC^{n_1}$.
Then, since $\A(xy^*)$ is sesqui-linear in $(x,y)$, $F$ and $G$ satisfy
\[
\A(x y^*) = [F(y)] x = \overline{[G(x)] y}.
\]

\subsection{Review of power factorization}
\label{sec:PF-review}

Power factorization (PF) \cite{HalHer2009pf} is an alternating minimization algorithm that estimates a rank-$r$ matrix $X \in \bbC^{n_1 \times n_2}$ from its linear measurements $b = \A(X) + z$.
In this section, we specialize PF to the rank-one case. Let $t \geq 0$ denote the iteration index. With a certain initialization $v_0$, PF iteratively updates estimates $X_t = u_t v_t^*$ by alternating between the following procedures:
\begin{itemize}
  \item For fixed $v_{t-1}$, update $u_t$ by
  \begin{equation}
  u_t = \argmin_{\tilde{u}} \norm{b - \A(\tilde{u} v_{t-1}^*)}_2^2.
  \label{eq:lsupdate_u}
  \end{equation}
  \item For fixed $u_t$, update $v_t$ by
  \begin{equation}
  v_t = \argmin_{\tilde{v}} \norm{b - \A(u_t \tilde{v}^*)}_2^2.
  \label{eq:lsupdate_v}
  \end{equation}
\end{itemize}

Using $F$ and $G$ defined in (\ref{eq:defFnG}),
the update rules in (\ref{eq:lsupdate_u}) and (\ref{eq:lsupdate_v}) can be rewritten respectively as
\begin{align}
u_t {} & = \argmin_{\tilde{u}} \norm{b - [F(v_{t-1})] \tilde{u}}_2^2 \label{eq:lsupdate2_u}
\intertext{and}
v_t {} & = \argmin_{\tilde{v}} \norm{\overline{b} - [G(u_t)] \tilde{v}}_2^2. \label{eq:lsupdate2_v}
\end{align}

\subsection{Sparse power factorization (SPF)}
\label{sec:alg-spf}

We propose an alternating minimization algorithm, called \textit{sparse power factorization} (SPF),
which recovers a row-sparse rank-one matrix $X = \lambda u v^* \in \bbC^{n_1 \times n_2}$
with $s_1$-sparse left singular vector $u$ and $s_2$-sparse right singular vector $v$.
SPF is obtained by modifying the updates of $u_t$ and $v_t$ in PF as follows.

Note that the measurement vector $b$ of $X$ can be expressed as
\[
b = \A(\lambda u v^*) + z = [F(v)] (\lambda u) + z.
\]
For fixed $v$, alternatively, $b$ can be understood as the measurement vector of the $s$-sparse vector $\lambda u$ using the sensing matrix $F(v)$. When $v_{t-1}$, normalized in the $\ell_2$ norm, corresponds to an estimate of the right singular vector $v$, the matrix $F(v_{t-1})$ can be interpreted as an estimate of the unknown sensing matrix $F(v)$. In the PF algorithm, the update of $u_t$ in (\ref{eq:lsupdate2_u}) corresponds to the least squares solution to the linear system consisting of the perturbed sensing matrix $F(v_{t-1})$ and the measurement vector $b$. In contrast, SPF exploits the sparsity of $u$ (when $s_1 < n_1$) and updates the left factor $u_t$ by an $s_1$-sparse estimate of $\lambda u$ from $b$ using the perturbed sensing matrix $F(v_{t-1})$. Existing sparse recovery algorithms such as CoSaMP \cite{NeeTro09}, subspace pursuit \cite{DaiMil09}, and hard thresholding pursuit (HTP) \cite{Fou2011htp} provide good estimates of $\lambda u$ at low computational cost.
Under certain conditions on the original and perturbed sensing matrices, these algorithms are guaranteed to have small estimation error. In this paper, we focus on a particular instance of SPF that updates $u_t$ using HTP, which is summarized in Alg.~\ref{alg:spf}. (For completeness, the HTP algorithm is detailed in Alg.~\ref{alg:htp}.) However, the results in this paper readily extend to instances of SPF employing other sparse recovery algorithms with a similar performance guarantee to that of HTP.
The step size $\gamma > 0$ in HTP depends on the scaling of the sensing matrix $F(v_t)$; hence, to fix the step size as $\gamma = 1$,\footnote{This step size leads to a guarantee using the sparsity-restricted isometry property of $F(v_{t-1})$ \cite{Fou2011htp}.} we normalize $v_{t-1}$ in the $\ell_2$ norm before the HTP step. Likewise, in the presence of column sparsity ($s_2 < n_2$),
the update of $v_t$ from $u_t$ is modified to exploit the sparsity prior on $v$ by using HTP.

\begin{algorithm}
\LinesNumbered
\SetAlgoNoLine
\caption{$\hat{X} = \texttt{SPF\_HTP}(\A,b,n_1,n_2,s_1,s_2,v_0)$}
\label{alg:spf}
\While{stop condition not satisfied}{
$t \leftarrow t+1$\;
$v_{t-1} \leftarrow \displaystyle \frac{v_{t-1}}{\norm{v_{t-1}}_2}$\;
\eIf{$s_1 < n_1$}{$u_t \leftarrow \texttt{HTP}(F(v_{t-1}),b,s_1)$\;}
{$u_t \leftarrow \displaystyle \argmin_x \norm{b - [F(v_{t-1})] x}_2^2$\;}
$u_t \leftarrow \displaystyle \frac{u_t}{\norm{u_t}_2}$\;
\eIf{$s_2 < n_2$}{$v_t \leftarrow \texttt{HTP}(G(u_t),\overline{b},s_2)$\;}
{$v_t \leftarrow \displaystyle \argmin_y \norm{\overline{b} - [G(u_t)] y}_2^2$\;}
}
\Return $\hat{X} \leftarrow u_t v_t^*$\;
\end{algorithm}

\begin{algorithm}
\LinesNumbered
\SetAlgoNoLine
\caption{$\hat{x} = \texttt{HTP}(\Phi,b,s)$}
\label{alg:htp}
\While{stop condition not satisfied}{
    $t \leftarrow t+1$\;
    $J \leftarrow \supp{H_s[x_{t-1} + \gamma \Phi^*(b - \Phi x_{t-1})]}$\;
    $\displaystyle x_t \leftarrow \argmin_x \left\{ \norm{b - \Phi x}_2 : \supp{x} \subset J \right\}$\;
}
\Return $\hat{x} \leftarrow x_t$\;
\end{algorithm}

The performance of alternating minimization algorithms usually depends critically on the initialization. A typical heuristic (cf. \cite{HalHer2009pf}) is to select the best solution $\widehat{X}$ that minimizes $\norm{b - \A(\widehat{X})}_2^2$ among solutions obtained by multiple random initializations. However, no theoretical guarantee has been shown for this heuristic. Instead, Jain et al. \cite{JaiNS2012} proposed to set $v_0$ to the leading right singular vector of the proxy matrix $\A^*(b) \in \bbC^{n_1 \times n_2}$, and provided a performance guarantee of PF with this initialization under the rank-restricted isometry property of $\A$. However, when applied to the sparse rank-one matrix recovery problem, this procedure does not exploit the sparsity of the eigenvectors, leading to highly suboptimal performance in the sparse regime.

\begin{algorithm}
\LinesNumbered
\SetAlgoNoLine
\DontPrintSemicolon
\caption{$\vzth = \texttt{thres\_init}(\A,b,n_1,n_2,s_1,s_2)$}
\label{alg:thres_proj}
$M \leftarrow \A^*(b)$\;
\For{$k=1,\dots,n_1$}{
$\zeta_k \leftarrow \text{$\ell_2$ norm of the $s_2$-sparse approx. of the $k$th row of $M$}$\;
}
$\widehat{J}_1 \leftarrow \text{indices of the $s_1$ entries of $\zeta$ with the largest magnitude}$\;
$\widehat{J}_2 \leftarrow \text{indices of the $s_2$ columns of $\Pi_{\widehat{J}_1} M$ with the largest $\ell_2$ norm}$\;
$\vzth \leftarrow \text{the first dominant right singular vector of $\Pi_{\widehat{J}_1} M \Pi_{\widehat{J}_2}$}$\;
\Return $\vzth$\;
\end{algorithm}

To achieve near optimal recovery of sparse rank-one matrices, we propose a simple initialization method that exploits the sparsity structure, which is summarized in Algorithm~\ref{alg:thres_proj}.
Although the initialization $v_0^\text{\rm Th}$ is practical thanks to its low computational cost,
the success of Algorithm~\ref{alg:thres_proj} requires an extra condition on the unknown singular vectors. It is of interest to design a more sophisticated initialization with a better performance at an increased computational cost. For example, similarly to the initialization for PF by Jain et al. \cite{JaiNS2012}, if the best sparse and rank-one approximation of the matrix $\A^*(b)$ is available, then one can compute a good initialization as follows. Compute estimates $\widehat{J}_1$ on the support $J_1$ of $u$ and $\widehat{J}_2$ on the support $J_2$ of $v$ by solving
\begin{equation}
(\widehat{J}_1,\widehat{J}_2) \triangleq \argmax_{|\widetilde{J}_1| = s_1, |\widetilde{J}_2| = s_2} \norm{\Pi_{\widetilde{J}_1} [\A^*(b)] \Pi_{\widetilde{J}_2}}.
\label{eq:initbydspca}
\end{equation}
The leading right singular of $\Pi_{\widehat{J}_1} [\A^*(b)] \Pi_{\widehat{J}_2}$, denoted by $\vzstar$, is used as the initialization for SPF. We refer to this procedure as \emph{optimal} initialization. Solving (\ref{eq:initbydspca}) involves searching over all possible support sets, which can be computationally demanding in high-dimensional settings.
Iterative algorithms developed for sparse principal component analysis (e.g., \cite{yuan2013truncated}) might be employed to get a good approximate solution to \eqref{eq:initbydspca}.
In this paper, we will focus only on the simple thresholding initialization $\vzth$ by Algorithm~\ref{alg:thres_proj} and the optimal initialization $\vzstar$.
Performance guarantees of the SPF algorithms equipped with these initialization schemes are presented in the next section.

\section{Rank-1 Recovery Guarantees}
\label{sec:ub}
In this section we provide upper bounds on the number of linear measurements
that guarantee the stable recovery of sparse and rank-one matrices by SPF with high probability. We consider Gaussian sensing schemes with the linear operator $\A: \bbC^{n_1 \times n_2} \to \bbC^m$ given by
\[
\A(Z) = [\langle M_1, Z \rangle, \ldots, \langle M_m, Z \rangle]^\transpose,
\]
where $M_\ell \in \bbC^{n_1 \times n_2}$ has i.i.d. $\calC \calN (0,1/m)$ entries. We call such an $\calA$ an \emph{i.i.d. Gaussian measurement operator}.

Recall the two initialization schemes defined in \prettyref{sec:alg-spf}. Our main results are stated in the following theorem.
\begin{theorem}
\label{thm:pgrip_init_ds_iidG}
Let $\A: \bbC^{n_1 \times n_2} \to \bbC^m$ be an i.i.d. Gaussian measurement operator.
There exist absolute constants $c_1$, $c_2$, $c_3$, $c_4$, and $C$ such that
for all $s_1 \in [n_1],s_2 \in [n_2]$,
the following statement holds with probability at least $1 - \exp(-c_2 m)$.
If $m \geq c_1 (s_1+s_2) \log(\max\{e n_1/s_1, e n_2/s_2\})$, then when initialized by $v_0^\text{\rm Th}$, SPF outputs $\widehat{X}$ that satisfies
\begin{equation}
\label{eq:thm:pgrip_init_ds_iidG:res}
\frac{\fnorm{\hat{X} - X}}{\fnorm{X}} \leq C \frac{\norm{z}_2}{\norm{\A(X)}_2}
\end{equation}
for all $(s_1,s_2)$-sparse and rank-one $X=\lambda uv^*$ with $\norm{u}_2 = \norm{v}_2 = 1$ and $\min(\norm{u}_{\infty}, \norm{v}_{\infty}) \geq c_4$, and for all $z$ with $\norm{z}_2 \leq c_3 \norm{\A(X)}_2$. In the special cases of $s_2=n_2$ (row sparsity), the ``peakiness'' condition $\norm{u}_{\infty} \norm{v}_{\infty} \geq c_4$ is replaced by $\norm{u}_{\infty} \geq c_4$.
\end{theorem}

The probability in Theorem~\ref{thm:pgrip_init_ds_iidG} is with respect to the selection of an i.i.d. Gaussian measurement operator, and the guarantee applies uniformly to the set of all matrices following the underlying model. In particular, this result achieves (to within a logarithmic factor) the fundamental limit on the number $m$ of measurements, in comparison to the corresponding necessary condition in Section~\ref{sec:lb}.

Theorem~\ref{thm:pgrip_init_ds_iidG} claims that SPF initialized by $v_0^\text{\rm Th}$ provides stable reconstruction when the singular vectors $u$ and $v$ of the unknown matrix $X$ are heavily peaked in the sense that both $\norm{u}_\infty$ and $\norm{v}_\infty$ are larger than an absolute constant. Intuitively, in the presence of a few dominant components in $u$ and $v$, the simple thresholding heuristic in Algorithm~\ref{alg:thres_proj} can capture the location of these peaks although it might not identify the entire support sets. This peakiness property is satisfied by certain classes of ``fast-decaying'' signals. Let $u^{(k)}$ denote the $k$th largest magnitudes of $u$. For example, if $u^{(k)} \leq ck^{-\alpha}$ for $\alpha > \frac{1}{2}$ or $u^{(k)} \leq c \beta^k$ for $\beta \in (0,1)$, then $\norm{u}_\infty$ is larger than an absolute constant. These fast-decaying-magnitudes models on the sparse vector $u$ are often relevant to practical applications. For example, the magnitudes of the wavelet coefficients of piecewise smooth signals decay geometrically across the scales of the wavelet tree \cite{mallat2008wavelet}.

\begin{proposition}
\label{prop:without_peakedness}
In the setup of Theorem~\ref{thm:pgrip_init_ds_iidG}, SPF initialized by $\vzth$ provides the same recovery guarantee from $m = O(s_1 s_2 \log(\max\{e n_1/s_1, e n_2/s_2\}))$ measurements without requiring the peakiness condition.
\end{proposition}

\begin{remark}
\label{rem:without_peakedness}
The performance guarantee in Proposition~\ref{prop:without_peakedness} is only as good as those for other recovery algorithms with provable guarantees, which ignore the rank-one prior in the matrix structure and only exploit the sparsity prior (e.g., basis pursuit).
We included Proposition~\ref{prop:without_peakedness} to demonstrate that SPF initialized by $\vzth$ is as good as existing guaranteed algorithms even when the peakiness condition is not satisfied. Furthermore, SPF is still preferable to other methods that do not exploit the rank-one prior because it solves the un-lifted formulation and has much lower computational cost.
\end{remark}

As we show in Section~\ref{subsec:pgrip_spf}, given a good initialization, the convergence of the subsequent iterations is shown without the heavily-peakedness condition. The following proposition demonstrates that the initialization $\vzstar$ from the exact solution to \eqref{eq:initbydspca} enables the performance guarantee for SPF without the heavily-peakedness condition. Recall that the computation of $\vzstar$ involves exhaustive search over all support sets of cardinality $s_1$ and $s_2$. In fact, with this enumeration, by applying guaranteed algorithms for low-rank matrix recovery for each choice of the support, one can get the same sample complexity result as in Proposition~\ref{prop:opt_init} easily. Nonetheless, the success of SPF initialized by $\vzstar$ opens up the possibility of finding better initialization schemes using a practical approximate algorithm to solve \eqref{eq:initbydspca}.

\begin{proposition}
\label{prop:opt_init}
In the setup of Theorem~\ref{thm:pgrip_init_ds_iidG}, SPF initialized by $\vzstar$ provides the same recovery guarantee from $m = O((s_1+s_2) \log(\max\{e n_1/s_1, e n_2/s_2\}))$ measurements without requiring the peakiness condition.
\end{proposition}

The rest of this section is devoted to proving \prettyref{thm:pgrip_init_ds_iidG}, \prettyref{prop:without_peakedness}, and \prettyref{prop:opt_init}. The outline of the proof is the following:
\begin{enumerate}
  \item \prettyref{thm:conv_spf} gives a deterministic guarantee for SPF under the condition that the linear operator satisfies certain RIP conditions and the initial value is reasonably close to the true singular vector.
  \item \prettyref{thm:ripiidG} shows that the Gaussian measurement operator satisfies the desired RIP if the number of measurements is lower bounded accordingly.
  \item The sufficiency of the initialization methods $\vzstar$ and $\vzth$, defined in \prettyref{sec:alg-spf}, to satisfy the conditions in \prettyref{thm:conv_spf}, is established
      under respective conditions.
\end{enumerate}

\subsection{Restricted isometry properties}

A sufficient condition for stable recovery of SPF is that the linear operator $\A$ satisfies
certain \emph{restricted isometry property} (RIP) conditions.
The original version of RIP \cite{CanTao2005decoding}, denoted by $s$-sparse RIP in this paper, refers to
a linear operator being a near isometry when restricted to the set of $s$-sparse vectors.
This notion has been extended to a similar near isometry property restricted to the set of rank-$r$ matrices \cite{RecFP2010}.
Here, the relevant RIP condition to the analysis of SPF is the near isometry on the set of rank-$r$ matrices with at most $s_1$ nonzero rows and at most $s_2$ nonzero columns.

\begin{definition}[Rank-$r$ and $(s_1,s_2)$-sparse RIP]
A linear operator $\A: \bbC^{n_1 \times n_2} \to \bbC^m$ satisfies the \textit{rank-$r$ and doubly $(s_1,s_2)$-sparse RIP}
with isometry constant $\delta$ if
\[
(1-\delta) \fnorm{Z}^2 \leq \norm{\A(Z)}_2^2 \leq (1+\delta) \fnorm{Z}^2
\]
for all $Z \in \bbC^{n_1 \times n_2}$ such that $\rank(Z) \leq r$, $\norm{Z}_{0,2} \leq s_1$, and $\norm{Z^*}_{0,2} \leq s_2$.
\label{def:rip}
\end{definition}

\begin{remark}
The special case of the rank-$r$ and $(s_1,s_2)$-sparse RIP with $s_2 = n_2$ (resp. $s_1 = n_1$) is called
the rank-$r$ and row-$s_1$-sparse RIP (resp. the rank-$r$ and column-$s_2$-sparse RIP).
\end{remark}

The following result gives a sufficient condition for the Gaussian measurement operator to satisfy the RIP condition defined in Definition~\ref{def:rip}.

\begin{theorem}
\label{thm:ripiidG}
Let $\A: \bbC^{n_1 \times n_2} \to \bbC^m$ be an i.i.d. Gaussian measurement operator. If
\[
m \geq c_1 r (s_1+s_2) \log\left( \max\left\{ \frac{e n_1}{s_1}, \frac{e n_2}{s_2} \right\} \right),
\]
then $\A$ satisfies the rank-$r$ and $(s_1,s_2)$-sparse RIP with isometry constant $\delta$
with probability at least $1-\exp(-c_2 \delta^2 m)$, where $c_1,c_2$ are absolute constants.
\end{theorem}

The proof of Theorem~\ref{thm:ripiidG} is rather straightforward using standard mathematical tools in the literature \cite{BarDdVW2008simple,RecFP2010,CanPla2011oracle}; hence, we only provide a sketch. It follows from the standard volume argument and the exponential concentration of the i.i.d. Gaussian measurement operator \cite{BarDdVW2008simple}. The only difference from the derivation of the standard $s$-sparse RIP \cite{BarDdVW2008simple} is to use the $\epsilon$-net for all unit-norm rank-$r$ matrices, the cardinality of which is bounded according to the following lemma.

\begin{lemma}[Size of $\epsilon$-net of rank-$r$ matrices \cite{CanPla2011oracle}]
Let $\mathcal{S} = \{ X \in \bbC^{n_1 \times n_2} :~ \text{\rm rank}(X) \leq r,~ \fnorm{X} = 1 \}$.
There exists a subset $\mathcal{S}_\epsilon$ of $\mathcal{S}$ such that
\[
\sup_{X \in \mathcal{S}} \inf_{\hat{X} \in \mathcal{S}_\epsilon} \fnorm{X - \hat{X}} \leq \epsilon
\]
and
\[
|\mathcal{S}_\epsilon| \leq \left(\frac{9}{\epsilon}\right)^{r(n_1+n_2+1)}.
\]
\end{lemma}

\subsection{RIP-based recovery guarantees for SPF}
\label{subsec:pgrip_spf}

Performance guarantees for recovery by SPF are derived using the rank-2 and $(3s_1,3s_2)$-sparse RIP of $\A$. The next theorem shows that given a good initialization, SPF provides stable recovery.

\begin{theorem}[RIP-guarantee for SPF with good initialization]
\label{thm:conv_spf}
Suppose the followings:
\begin{enumerate}
  \item $X = \lambda u v^*$ satisfies $\norm{u}_0 \leq s_1$ and $\norm{v}_0 \leq s_2$.
  \item $\A$ satisfies the rank-2 and $(3s_1,3s_2)$-sparse RIP with isometry constant $\delta = 0.08$.
  \item $b = \A(X) + z$ where $z$ and $\A(X)$ satisfy
  \begin{equation}
  \label{eq:defnu}
  \frac{\norm{z}_2}{\norm{\A(X)}_2} \leq \nu
  \end{equation}
  with $\nu = 0.08$.
  \item The initialization $v_0$ of SPF satisfies
  \begin{equation}
  \norm{P_{\R(v)^\perp} P_{\R(v_0)}} < 0.85.
  \label{eq:conv_spf_initcond}
  \end{equation}
\end{enumerate}
Then, the output $(X_t)_{t \in \bbN}$ of SPF satisfies
\begin{equation}
\limsup_{t \to \infty} \frac{\fnorm{X_t - X}}{\fnorm{X}} \leq 8.3 \frac{\norm{z}_2}{\norm{\A(X)}_2}.
\label{eq:conv_spf_res}
\end{equation}
Moreover, the convergence in (\ref{eq:conv_spf_res}) is superlinear,
i.e., for any $\epsilon > 0$, there exists $t_0 = O(\log(1/\epsilon))$ that satisfies
\begin{equation}
\frac{\fnorm{X_{t_0} - X}}{\fnorm{X}} \leq 8.3 \frac{\norm{z}_2}{\norm{\A(X)}_2} + \epsilon.
\label{eq:conv_spf_res2}
\end{equation}
\end{theorem}

\begin{proof}
See Section~\ref{subsec:pgspf}.
\end{proof}

Theorem~\ref{thm:conv_spf} implies that under the rank-2 and $(3s_1,3s_2)$-sparse RIP assumption on $\A$, With a good initialization $v_0$, which is close to the unknown $v$ in the principal angle, SPF converges superlinearly to a robust reconstruction of $X$ in the sense of \prettyref{eq:stability}. In particular, in the noiseless case ($z = 0$), SPF recovers $X$ perfectly. Note that the performance guarantee in Theorem~\ref{thm:conv_spf} is obtained under the conservative assumption that the noise variance is below a certain constant threshold. However, empirically, SPF still provides stable recovery of $X$ even when the additive noise is stronger than the threshold in Theorem~\ref{thm:conv_spf} (See Section~\ref{sec:numres}).

Next, we address the question of finding a good initialization. The performance guarantee for SPF in Theorem~\ref{thm:conv_spf} holds subject to the condition that the initialization satisfies \eqref{eq:conv_spf_initcond}. We study the performance of the two initialization methods proposed in Section~\ref{sec:alg} and present the corresponding performance guarantees below.

\begin{theorem}[RIP-guarantees: doubly sparse case]
\label{thm:pgrip_init_ds}
Suppose that $X = \lambda u v^*$ satisfies $\norm{u}_0 \leq s_1$ and $\norm{v}_0 \leq s_2$.
\begin{enumerate}
  \item Suppose that $\A$ satisfies the rank-2 and $(3s_1,3s_2)$-sparse RIP with isometry constant $\delta = 0.04$, and that the SNR condition in (\ref{eq:defnu}) holds with $\nu = 0.04$. Then, SPF initialized by $\vzstar$ provides a performance guarantee as in Theorem~\ref{thm:conv_spf}.
  \item Suppose that either one of the following conditions is satisfied:
  \begin{enumerate}
    \item $\A$ satisfies the rank-2 and $(3s_1,3s_2)$-sparse RIP with isometry constant $\delta = 0.04$, the SNR condition in (\ref{eq:defnu}) holds with $\nu = 0.04$, $\norm{u}_\infty \geq 0.78 \norm{u}_2$, and $\norm{v}_\infty \geq 0.78 \norm{v}_2$.
    \item $\A$ satisfies the RIP with isometry constant $\delta = 0.02$, when restricted to the set of matrices with up to $9 s_1 s_2$ nonzero entries, and the SNR condition in (\ref{eq:defnu}) holds with $\nu = 0.02$.
  \end{enumerate}
  Then, SPF initialized by $v_0^\text{\rm Th}$ provides a performance guarantee as in Theorem~\ref{thm:conv_spf}.
\end{enumerate}
\end{theorem}

\begin{proof}
See Section~\ref{subsec:initds}.
\end{proof}

\begin{remark}
It is noteworthy that the two different performance guarantees in Part 2 of Theorem~\ref{thm:pgrip_init_ds} are achieved by a single algorithm. In fact, when $\A$ satisfies the $(3s_1,3s_2)$-sparse RIP, nuclear norm minimization achieves a performance guarantee, which applies to all $(s_1,s_2)$-sparse matrices (not necessarily of rank-1). Part 2-(b) of Theorem~\ref{thm:pgrip_init_ds} asserts that SPF with the thresholding initialization provides a comparable performance guarantee in this scenario (only by the $(3s_1,3s_2)$-sparse RIP without any further condition on $u$ and $v$). The performance guarantee in Part 2-(b) of Theorem~\ref{thm:pgrip_init_ds} only recovers existing results. However, by Part 2-(a) of Theorem~\ref{thm:pgrip_init_ds}, when $u$ and $v$ have large entries, unlike the nuclear norm minimization that discards the rank-1 constraint, the same algorithm (SPF with the thresholding initialization) achieves a better performance guarantee by a weaker RIP.
\end{remark}

It is straightforward to check that the performance guarantees in Theorem~\ref{thm:pgrip_init_ds}
apply to the row-sparse (resp. the column sparse case) by letting $s_2 = n_2$ (resp. $s_1 = n_1$). However, in the row-sparse case (resp. the column-sparse case), the near optimal performance guarantee for SPF initialized by $v_0^\text{\rm Th}$ requires only the additional condition on $\norm{u}_\infty$ (resp. $\norm{v}_\infty$), as stated in the following result.

\begin{theorem}[RIP-guarantee: row-sparse case]
\label{thm:pgrip_init_rs}
Suppose the followings:
\begin{enumerate}
  \item $X = \lambda u v^*$ satisfies $\norm{u}_0 \leq s_1$ and $\norm{u}_\infty \geq 0.4 \norm{u}_2$.
  \item $\A$ satisfies the rank-2 and row-$3s_1$-sparse RIP with isometry constant $\delta = 0.08$.
  \item The SNR condition in (\ref{eq:defnu}) holds with $\nu = 0.04$.
\end{enumerate}
Then, SPF initialized by $v_0^\text{\rm Th}$ provides a performance guarantee as in Theorem~\ref{thm:conv_spf}.
\end{theorem}

\begin{proof}
See Section~\ref{subsec:initrs}.
\end{proof}

\begin{remark}
The constants in Theorems~\ref{thm:pgrip_init_ds} and \ref{thm:pgrip_init_rs} are not optimized,
but rather were chosen conservatively to simplify the proofs and the statement of the results.
\end{remark}

We conclude this section with the proofs of Theorem~\ref{thm:pgrip_init_ds_iidG}, Proposition~\ref{prop:without_peakedness}, and Proposition~\ref{prop:opt_init}.
\begin{proof}
In view of Theorem~\ref{thm:ripiidG}, the RIP conditions in Theorem~\ref{thm:pgrip_init_ds}
are satisfied by corresponding conditions on the number of i.i.d. Gaussian measurements as follows.
First, the performance guarantee for SPF initialized by $\vzstar$ is given by the rank-2 and $(3s_1,3s_2)$-sparse RIP;
hence, it holds with with high probability for $m = O((s_1 + s_2) \log(\max\{e n_1/s_1, e n_2/s_2\}))$ i.i.d. Gaussian measurements.
Thus, Proposition~\ref{prop:opt_init} follows from Part~1 of Theorem~\ref{thm:pgrip_init_ds}.
Next, the RIP conditions in Part 2-(a) of Theorem~\ref{thm:pgrip_init_ds}
and in Theorem~\ref{thm:pgrip_init_rs} are similarly expressed as conditions on $m$.
This proves Theorem~~\ref{thm:pgrip_init_ds_iidG}.
Finally, noting that the RIP condition in Part 2-(b) of Theorem~\ref{thm:pgrip_init_ds} holds for an i.i.d. Gaussian measurement operator with $m = O(s_1 s_2 \log(\max\{e n_1/s_1, e n_2/s_2\}))$
proves Proposition~\ref{prop:without_peakedness}.
\end{proof}

\section{Extension to the Rank-$r$ Case}
\label{sec:rankr}
In this section, we extend the results in Section~\ref{sec:ub} from the rank-1 case to a more general rank-$r$ case.

\subsection{Algorithms for the rank-$r$ case}

First, we generalize the definition of $F(\cdot)$ and $G(\cdot)$ in \eqref{eq:defFnG} to the rank-$r$ case. Recall that there exist matrices $M_1,M_2,\ldots,M_m \in \cz^{n_1 \times n_2}$ such that
\[
\A(X) = [\langle M_1, X \rangle, \ldots, \langle M_m, X \rangle]^\transpose.
\]
For $V \in \cz^{n_2 \times r}$, a linear operator $\calF(V): \cz^{n_1 \times r} \to \cz^m$ parameterized by $V$ is defined by
\begin{equation}
\label{eq:def_calF}
[\calF(V)](U) := [\langle M_1 V, U \rangle, \ldots, \langle M_m V, U \rangle]^\transpose,
\quad \forall U \in \cz^{n_1 \times r}.
\end{equation}
For $U \in \cz^{n_1 \times r}$, a linear operator $\calG(U): \cz^{n_2 \times r} \to \cz^m$ parameterized by $U$ is defined by
\begin{equation}
\label{eq:def_calG}
[\calG(U)](V) := [\langle M_1^* U, V \rangle, \ldots, \langle M_m^* U, V \rangle]^\transpose,
\quad \forall V \in \cz^{n_2 \times r}.
\end{equation}
Then,
\[
\A(U V^*) = [\calF(V)](U) = \overline{[\calG(U)](V)}.
\]
When $r=1$, the above definitions reduce to the corresponding part in Section~\ref{sec:alg}.

With $\calF$ and $\calG$ defined respectively in \eqref{eq:def_calF} and \eqref{eq:def_calG},
similarly to the PF algorithm \cite{HalHer2009pf}, SPF in Algorithm~\ref{alg:spf} extends naturally to the rank-$r$ case, which is summarized in Algorithm~\ref{alg:rspf}. We also extend the thresholding initialization in Section~\ref{sec:alg} to the rank-$r$ case for both the initial estimates $U_0$ and $V_0$. Theses algorithms are summarized in Algorithms~\ref{alg:initV0} and \ref{alg:initU0}, respectively.

\insrt{
\begin{algorithm}
\LinesNumbered
\SetAlgoNoLine
\DontPrintSemicolon
\caption{$\hat{X} = \texttt{rSPF\_HTP}(\A,b,n_1,n_2,r,s_1,s_2,U_0,V_0)$}
\label{alg:rspf}
\While{stop condition not satisfied}{
    $t \leftarrow t+1$\;
    $V_{t-1} \leftarrow \texttt{orth}(V_{t-1})$\tcp*{ortho-basis for $\mathcal{R}(V_{t-1})$}
    \eIf{$s_1 < n_1$}{$U_t \leftarrow \texttt{B-HTP}(\calF(V_{t-1}),b,s_1)$\;}{$U_t \leftarrow \displaystyle \argmin_{U'} \norm{b - [\calF(V_{t-1})] (U')}_2^2$\;}
    $U_t \leftarrow \texttt{orth}(U_t)$\;
    \eIf{$s_2 < n_2$}{$V_t \leftarrow \texttt{B-HTP}(\calG(U_t),\overline{b},s_2)$\;}{$V_t \leftarrow \displaystyle \argmin_{V'} \norm{\overline{b} - [\calG(U_t)] (V')}_2^2$\;}
}
\Return $\hat{X} \leftarrow U_t V_t^*$\;
\end{algorithm}
}

\begin{algorithm}
\LinesNumbered
\SetAlgoNoLine
\caption{$\hat{x} = \texttt{B-HTP}(\Phi,b,s)$}
\label{alg:bhtp}
\While{stop condition not satisfied}{
    $t \leftarrow t+1$\;
    $\widetilde{X} \leftarrow X_{t-1} + \gamma \Phi^*(b - \Phi (X_{t-1}))$\;
    $J \leftarrow \mbox{indices of the $s$ rows of $\widetilde{X}$ with the largest $\ell_2$ norm}$\;
    $\displaystyle X_t \leftarrow \argmin_{X'} \left\{ \norm{b - \Phi(X')}_2 : \Pi_J X' = X' \right\}$\;
}
\Return $\widehat{X} \leftarrow X_t$\;
\end{algorithm}

\begin{algorithm}
\LinesNumbered
\SetAlgoNoLine
\caption{$V_0^\text{\rm Th} = \texttt{INIT\_SC\_SPF\_V}(\A,b,n_1,n_2,r,s_1,s_2)$}
\label{alg:initV0}
$M \leftarrow \A^*(b)$\;
\For{$k=1,\dots,n_1$}{
$\zeta_k \leftarrow \text{$\ell_2$ norm of the $s_2$-sparse approx. of the $k$th row of $M$}$\;
}
$\widehat{J}_1 \leftarrow \text{indices of the $s_1$ entries of $\zeta$ with the largest magnitude}$\;
$\widehat{J}_2 \leftarrow \text{indices of the $s_2$ columns of $\Pi_{\widehat{J}_1} M$ with the largest $\ell_2$ norm}$\;
$V_0^\text{\rm Th} \leftarrow \text{$r$ leading right singular vectors of $\Pi_{\widehat{J}_1} M \Pi_{\widehat{J}_2}$}$\;
\end{algorithm}

\begin{algorithm}
\LinesNumbered
\SetAlgoNoLine
\caption{$U_0^\text{\rm Th} = \texttt{INIT\_SC\_SPF\_U}(\A,b,n_1,n_2,r,s_1,s_2)$}
\label{alg:initU0}
$M \leftarrow \A^*(b)$\;
\For{$k=1,\dots,n_2$}{
$\zeta_k \leftarrow \text{$\ell_2$ norm of the $s_1$-sparse approx. of the $k$th column of $M$}$\;
}
$\widehat{J}_2 \leftarrow \text{indices of the $s_2$ entries of $\zeta$ with the largest magnitude}$\;
$\widehat{J}_1 \leftarrow \text{indices of the $s_1$ rows of $M \Pi_{\widehat{J}_2}$ with the largest $\ell_2$ norm}$\;
$V_0^\text{\rm Th} \leftarrow \text{$r$ leading left singular vectors of $\Pi_{\widehat{J}_1} M \Pi_{\widehat{J}_2}$}$\;
\end{algorithm}

As demonstrated in Section~\ref{subsec:numres:lrds}, empirically, the natural extension of SPF (Algorithm~\ref{alg:rspf}) outperformed the convex method. However, in our attempt to extend Theorem~\ref{thm:pgrip_init_ds_iidG} to the rank-$r$ case, instead of the linear dependence, the sample complexity for performance guarantees had higher order dependence on the rank $r$. This is suboptimal in order compared to the matching lower bound. To overcome this limitation, we modify the natural rank-$r$ sparse power factorization (Algorithm~\ref{alg:rspf}) into the subspace-concatenated sparse power factorization (SC-SPF), summarized in Algorithm~\ref{alg:scspf}. The most important difference between SC-SPF and SPF is that in every iteration of SC-SPF, the initial estimate is used in concatenation with the estimate from the previous iteration. As we show in the next section, with this subspace concatenation, the sample complexity for recovering a low-rank and sparse matrix scales linearly in the rank, which is optimal.

\begin{algorithm}
\LinesNumbered
\SetAlgoNoLine
\caption{$\hat{X} = \texttt{SC\_SPF\_HTP}(\A,b,n_1,n_2,r,s_1,s_2,U_0,V_0)$}
\label{alg:scspf}
\While{stop condition not satisfied}{
    $t \leftarrow t+1$\;
    $\widetilde{V} \leftarrow \texttt{orth}([V_{t-1}, V_0])$\;
    \eIf{$s_1 < n_1$}{$\widetilde{U} \leftarrow \texttt{B-HTP}(\calF(\widetilde{V}),b,s_1)$\;}{$\widetilde{U} \leftarrow \displaystyle \argmin_{U'} \norm{b - [\calF(\widetilde{V})] (U')}_2^2$\;}
    $U_t \leftarrow \mbox{(the best rank-$r$ approximation of $\widetilde{U}$)}$\;
    $\widetilde{U} \leftarrow \texttt{orth}([U_t, U_0])$\;
    \eIf{$s_2 < n_2$}{$\widetilde{V} \leftarrow \texttt{B-HTP}(\calG(\widetilde{U}),\overline{b},s_2)$\;}{$\widetilde{V} \leftarrow \displaystyle \argmin_{V'} \norm{\overline{b} - [\calG(\widetilde{U})] (V')}_2^2$\;}
    $V_t \leftarrow \mbox{(the best rank-$r$ approximation of $\widetilde{V}$)}$\;
}
\Return $\hat{X} \leftarrow U_t V_t^*$\;
\end{algorithm}

\subsection{Performance guarantees}

Similar to the performance guarantee in Section~\ref{sec:ub}, we derive a sufficient condition
for sparse recovery by SC-SPF from i.i.d. Gaussian measurements.

\begin{theorem}
\label{thm:pgrip_ds_iidG_rankr}
Let $\A: \bbC^{n_1 \times n_2} \to \bbC^m$ be an i.i.d. Gaussian measurement operator.
There exist absolute constants $c_1,\ldots,c_5$, and $C$ such that the following statement holds.
Let $X \in \mathbb{C}^{n_1 \times n_2}$ be a fixed matrix of rank-$r$,
where $X = U \Lambda V^*$ denotes the singular value decomposition of $X$.
Suppose that the following conditions hold:
\begin{enumerate}
  \item The condition number of $X$ is no greater than $\kappa \leq c_1$.
  \item $U \in \mathbb{C}^{n_1 \times r}$ and $V \in \mathbb{C}^{n_2 \times r}$ are row-$s_1$-sparse and row-$s_2$-sparse, respectively.
  \item $\sigma_r(\Pi_{\widetilde{J}_1} U) \geq c_2$ and $\sigma_r(\Pi_{\widetilde{J}_2} V) \geq c_2$, where $\widetilde{J}_1$ and $\widetilde{J}_2$ are respectively defined by
  \begin{equation}
  \label{eq:def_widetildeJs}
  \widetilde{J}_1 \triangleq \argmax_{\Upsilon \subset [n_1]: |\Upsilon| = r} \fnorm{\Pi_{\Upsilon} U} \quad \text{and} \quad
  \widetilde{J}_2 \triangleq \argmax_{\Upsilon \subset [n_2]: |\Upsilon| = r} \fnorm{\Pi_{\Upsilon} V}.
  \end{equation}
  \item $b = \A(X) + z$ where $\norm{z}_2 \leq c_3 (\norm{X}/\fnorm{X}) \norm{\A(X)}_2$.
  \item $m \geq c_4 \kappa^2 r (s_1+s_2) \log(\max\{e n_1/s_1, e n_2/s_2\})$.
\end{enumerate}
Then initialized by $(U_0^\text{\rm Th},V_0^\text{\rm Th})$, SC-SPF outputs $\widehat{X}$ that satisfies
\begin{equation}
\label{eq:tildeJ_rankr}
\frac{\fnorm{\hat{X} - X}}{\fnorm{X}} \leq C \kappa^2 \left( \frac{\norm{z}_2}{\norm{\A(X)}_2} \right).
\end{equation}
with probability at least $1 - \exp(-c_5 m)$.
\end{theorem}

\begin{remark}
Assumption 3 in Theorem~\ref{thm:pgrip_ds_iidG_rankr} generalizes the peakiness assumption, $\norm{u}_\infty \geq c_2$ and $\norm{v}_\infty \geq c_2$, of Theorem~\ref{thm:pgrip_init_ds_iidG}. In the rank-1 case, $\widetilde{J}_1$ (resp. $\widetilde{J}_2$) in \eqref{eq:def_widetildeJs} reduces to the index of the largest entry of $u$ (resp. $v$) in magnitude. Therefore, Assumption 3 in Theorem~\ref{thm:pgrip_ds_iidG_rankr} reduces to the corresponding peakiness assumption in Theorem~\ref{thm:pgrip_init_ds_iidG}.
\end{remark}

\begin{remark}
The dependence of the sample complexity on the condition number $\kappa$ is due to the estimation of subspaces. We can always apply the algorithm with parameter $r'<r$ to decrease $\kappa$ and the reconstruction error will depends on the best rank-$r'$ approximation of $X$ (which is now absorbed into the measurement error $z$).
\end{remark}

When the unknown rank-$r$ matrix $X$ is not sparse ($s_1 = n_1$ and $s_2 = n_2$), there is no need to estimate the support sets and the initialization $(U_0^\text{\rm Th},V_0^\text{\rm Th})$ by Algorithms~\ref{alg:initV0} and \ref{alg:initU0} is trivially obtained as the singular vectors of the rank-$r$ approximation of $\A^*(b)$. Furthermore, the iterative updates in Algorithm~\ref{alg:scspf} are done by solving least squares problems. In this scenario, the guarantee in Theorem~\ref{thm:pgrip_ds_iidG_rankr} holds only with Assumption 4 (sufficiently high SNR) as shown in the next corollary.

\begin{corollary}[Non-sparse case]
\label{cor:pgrip_ds_iidG_rankr}
Let $\A: \bbC^{n_1 \times n_2} \to \bbC^m$ be an i.i.d. Gaussian measurement operator.
There exist absolute constants $c_3,c_4,c_5$, and $C$ such that the following statement holds.
Let $X \in \mathbb{C}^{n_1 \times n_2}$ be a fixed matrix of rank-$r$, where $X = U \Lambda V^*$ denotes the singular value decomposition of $X$. Suppose that Assumption 4 in Theorem~\ref{thm:pgrip_ds_iidG_rankr} holds.
If $m \geq c_2 \kappa^2 r (n_1+n_2)$, then initialized by $(U_0^\text{\rm Th},V_0^\text{\rm Th})$, SC-SPF outputs $\widehat{X}$ that satisfies \eqref{eq:tildeJ_rankr} with probability at least $1 - \exp(-c_5 m)$.
\end{corollary}

Corollary~\ref{cor:pgrip_ds_iidG_rankr} implies that the recovery of an $n$-by-$n$ matrix of rank $r$ is guaranteed from $m = O(\kappa^2 r n)$ measurements\footnote{The logarithmic term comes from the unknown support and disappears in the non-sparse case.}, which significantly improves on the previous result $m = O(\kappa^4 r^3 n)$ \cite{JaiNS2012} also using alternating minimization and i.i.d. Gaussian measurements.

In the remainder of this section we provide a proof of Theorem~\ref{thm:pgrip_ds_iidG_rankr}. To this end, we first derive a deterministic RIP-based guarantee for SC-SPF assuming good initialization, then a RIP-based guarantee using the initialization in Algorithms~\ref{alg:initV0} and \ref{alg:initU0}, and finally the sample complexity for the relevant RIP using a Gaussian measurement operator.

\begin{theorem}[RIP-guarantee for SC-SPF with good initialization]
\label{thm:conv_scspf}
Suppose the followings:
\begin{enumerate}
  \item $X = U \Lambda V^*$ denotes the singular value decomposition of a rank-$r$ matrix $X \in \mathbb{C}^{n_1 \times n_2}$, where $U \in \mathbb{C}^{n_1 \times r}$ and $V \in \mathbb{C}^{n_2 \times r}$ are row-$s_1$-sparse and row-$s_2$-sparse, respectively.
  \item The condition number of $X$ is no greater than $\kappa$.
  \item $\A$ satisfies the rank-$2r$ and $(3s_1,3s_2)$-sparse RIP with isometry constant $\delta = \frac{0.04}{\kappa}$.
  \item $b = \A(X) + z$ where $z$ and $\A(X)$ satisfy
  \begin{equation}
  \label{eq:snrcond_rankr}
  \frac{\fnorm{X}}{\norm{X}} \cdot \frac{\norm{z}_2}{\norm{\A(X)}_2} \leq \nu
  \end{equation}
  with $\nu = \frac{0.04}{\kappa}$.
  \item The initialization $(U_0,V_0)$ of SC-SPF satisfies
  \begin{equation}
  \label{eq:conv_scspf_initcond}
  \max(\norm{P_{\R(U)^\perp} P_{\R(U_0)}}, ~ \norm{P_{\R(V)^\perp} P_{\R(V_0)}}) < 0.95.
  \end{equation}
\end{enumerate}
Then the output $(X_t)_{t \in \bbN}$ of SC-SPF satisfies
\begin{equation}
\limsup_{t \to \infty} \frac{\fnorm{X_t - X}}{\fnorm{X}} \leq (55 \kappa^2 + 3 \kappa + 3) \frac{\norm{z}_2}{\norm{\A(X)}_2}.
\label{eq:conv_scspf_res}
\end{equation}
Moreover, the convergence in \eqref{eq:conv_scspf_res} is linear,
i.e., for any $\epsilon > 0$, there exists $t_0 = O(\log \frac{1}{\epsilon} )$ that satisfies
\begin{equation}
\frac{\fnorm{X_{t_0} - X}}{\fnorm{X}} \leq (55 \kappa^2 + 3 \kappa + 3) \frac{\norm{z}_2}{\norm{\A(X)}_2} + \epsilon.
\label{eq:conv_scspf_res2}
\end{equation}
\end{theorem}

Theorem~\ref{thm:conv_scspf} implies that starting from good initial estimates $U_0$ and $V_0$, SC-SPF provide stable recovery of the unknown matrix $X$ under the rank-$2r$ and doubly $(3s_1,3s_2)$-sparse RIP of $\A$. In particular, when the unknown matrix is only of rank $r$ without sparsity ($s_1 = n_1$ and $s_2 = n_2$), one can obtain good initial estimates satisfying \eqref{eq:conv_scspf_initcond} by a single step of the singular value projection \cite{JaiNS2012}. In this case, Theorem~\ref{thm:conv_scspf}, combined with Theorem \ref{thm:ripiidG} and Lemma~\ref{lemma:rip_spectral},
shows that stable recovery of an unknown $n \times n$ matrix of rank-$r$ from $O(\kappa^2 r n)$ i.i.d. Gaussian measurements is guaranteed.
With the sparsity model, we provide a performance guarantee of SC-SPF with initial estimates using Algorithms~\ref{alg:initV0} and \ref{alg:initU0} in the following theorem.

\begin{theorem}[RIP-guarantee]
\label{thm:pgrip_init_ds_rankr}
Suppose the followings:
\begin{enumerate}
  \item $X = U \Lambda V^*$ is the singular value decomposition of a rank-$r$ matrix $X \in \mathbb{C}^{n_1 \times n_2}$, where $U \in \mathbb{C}^{n_1 \times r}$ and $V \in \mathbb{C}^{n_2 \times r}$ are row-$s_1$-sparse and row-$s_2$-sparse, respectively.
  \item $\sigma_r(\Pi_{\widetilde{J}_1} U) \geq 0.9$ and $\sigma_r(\Pi_{\widetilde{J}_2} V) \geq 0.9$, where $\widetilde{J}_1$ and $\widetilde{J}_2$ are defined in \eqref{eq:def_widetildeJs}.
  \item The condition number of $X$ is no greater than $\kappa \leq 4$.
  \item $\A$ satisfies the rank-$2r$ and $(3s_1,3s_2)$-sparse RIP with isometry constant $\delta = \frac{0.04}{\kappa}$.
  \item $\A$ and $X$ satisfy
  \[
  \sup_{|J_1| \leq s_1} \sup_{|J_2| \leq s_2} \norm{\Pi_{J_1} [(\A^*\A - \id)(X)] \Pi_{J_2}} \leq \delta \norm{X}
  \]
  for $\delta = \frac{0.04}{\kappa}$.
  \item $b = \A(X) + z$ where the SNR condition in (\ref{eq:snrcond_rankr}) holds with $\nu = \frac{0.04}{\kappa}$.
\end{enumerate}
Then Algorithm~\ref{alg:scspf} initialized by $(U_0^\text{\rm Th},V_0^\text{\rm Th})$ provides a performance guarantee as in Theorem~\ref{thm:conv_scspf}.

In particular, when the unknown rank-$r$ matrix $X$ is not sparse ($s_1 = n_1$ and $s_2 = n_2$), the same guarantee holds without Assumptions 2 and 3.
\end{theorem}

\begin{remark}
In the rank-1 case,
$\sigma_r(\Pi_{\widetilde{J}_1} U)$ and $\sigma_r(\Pi_{\widetilde{J}_2} V)$ reduce to $\norm{u}_\infty$ and $\norm{v}_\infty$, respectively.
\end{remark}

The next lemma provides an RIP-like property of an i.i.d. Gaussian operator. Theorem~\ref{thm:pgrip_ds_iidG_rankr} is then obtained by combining Theorems~\ref{thm:pgrip_init_ds_rankr}, \ref{thm:ripiidG}, and Lemma~\ref{lemma:rip_spectral}.

\begin{lemma}
\label{lemma:rip_spectral}
Let $X \in \mathbb{C}^{n_1 \times n_2}$ be an arbitrarily fixed rank-$r$ matrix.
Let $\A$ be defined in (\ref{eq:defcalA})
where $M_1,\ldots,M_m$ are mutually independent random matrices whose entries are i.i.d. following $\calC \calN(0,1/m)$.
Then, with probability $1-\epsilon$,
\[
\sup_{|J_1| \leq s_1} \sup_{|J_2| \leq s_2} \norm{\Pi_{J_1} [(\A^*\A - \id)(X)] \Pi_{J_2}} \leq \delta r^{-1/2} \fnorm{X} \leq \delta \norm{X}
\]
provided
\[
m \geq C r \delta^{-2} \max\left[ (s_1 + s_2) \log(\max\{e n_1/s_1, e n_2/s_2\}), \log(\epsilon^{-1}) \right]
\]
for an absolute constant $C > 0$.
\end{lemma}

\begin{proof}[Proof of Lemma~\ref{lemma:rip_spectral}]
Let
\[
V_Z \triangleq \frac{1}{\sqrt{m}} (I_m \otimes \vect(Z)^\transpose),
\]
where $\otimes$ denotes the Kronecker product and
\[
\xi \triangleq [\vect(\overline{M_1})^\transpose, \ldots, \vect(\overline{M_m})^\transpose]^\transpose.
\]
Note that $\xi \in \cz^{m n_1 n_2}$ is a Gaussian vector with $\mathbb{E} \xi \xi^* = I_{m n_1 n_2}$.
Then, the action of $\A$ on $Z \in \cz^{n_1 \times n_2}$ is expressed as
\[
\A(Z) = V_Z \xi.
\]

Therefore, it follows that
\[
\sup_{|J_1| \leq s_1} \sup_{|J_2| \leq s_2} \norm{\Pi_{J_1} [(\A^*\A - \id)(X)] \Pi_{J_2}}
= \sup_{V_Z \in \Xi} \big| \langle V_X \xi, V_Z \xi \rangle - \mathbb{E} \langle V_X \xi, V_Z \xi \rangle \big|,
\]
where
\[
\Xi \triangleq \{ V_Z :~ Z = u v^*, ~ \norm{u}_2 = 1, ~ \norm{u}_0 \leq s_1, ~ \norm{v}_2 = 1, ~ \norm{v}_0 \leq s_2 \}.
\]

Then, the radius of $\Xi$ in the spectral norm is
\[
d_{2 \to 2}(\Xi) = \frac{1}{\sqrt{m}}
\]
and the radius of $\Xi$ in the Frobenius norm is
\[
d_{\mathrm{F}}(\Xi) = 1.
\]
Furthermore, Talagrand's $\gamma_2$ functional \cite{talagrand2006generic} of $\Xi$ with respect to the spectral norm is upper bounded by
\[
\gamma_2(\Xi,\norm{\cdot})
\lesssim \sqrt{\frac{s_1 \log(e n_1/s_1)}{m}} + \sqrt{\frac{s_2 \log(e n_2/s_2)}{m}}.
\]
Then, the conclusion follows from \cite[Theorem~2.3]{LeeJunge2015}.
\end{proof}

\section{Fundamental Limits: Necessary Conditions for Robust Reconstruction}
\label{sec:lb}



In this section we give necessary conditions for robust reconstruction of sparse rank-one matrices by considering a Bayesian version of \prettyref{eq:model}.
Denote the $r$-dimensional complex Stiefel manifold $V(\complex^n, r) \triangleq \{V\in\complex^{n\times r}: V^*V=I_r\}$.
Following the information-theoretic approach in \cite{WV12}, we prove a non-asymptotic lower bound that holds for any non-linear measurement mechanisms (encoders) and reconstruction algorithms (decoders).
The setup is illustrated in \prettyref{fig:bayesian},
where \begin{itemize}
  \item $X=UV^* \in \complex^{n_1\times n_2}$ is a rank-$r$ random matrix, where $U\in\complex^{n_1\times r}$ and $V\in\complex^{n_2\times r}$ are independent random matrices. Furthermore, $U$ is row-sparse with row support $S=\supp{U}$ chosen uniformly at random from all subsets of $[n_1]$ of cardinality $s_1$, and the non-zero rows $U_S$ are uniformly distributed on $\sfV(\complex^{s_1},r)$. Similarly, $V$ is $s_2$-sparse with uniformly chosen support and the non-zero component is uniform on $\sfV(\complex^{s_2},r)$.
  \item $Z \sim \calC\calN(0,\sigma^2 I_m)$ denotes additive complex Gaussian noise, whose real and imaginary parts are independently distributed according to $\calN(0, \frac{\sigma^2}{2} I_m)$.
  \item The encoder $f: \complex^{n_1 \times n_2} \to \complex^m$ satisfies the average power constraint
  \begin{equation}
  \Expect \|f(X)\|^2 \leq m
  \label{eq:avg-power}
  \end{equation}
  \item The decoder $g: \complex^m\times  \complex^{n_1 \times n_2}$ outputs a rank-$r$ matrix, namely, $g(f(X)+Z) = \hU\hV^*$ with $\hU\in \sfV(\complex^{n_1},r)$ and $\hV\in \sfV(\complex^{n_2},r)$.
\end{itemize}

\begin{figure}[ht]
	\centering
\begin{tikzpicture}[scale=1.2,transform shape,node distance=2.5cm,auto,>=latex']
    \node [int] (a) {$\substack{\textrm{~~~~Encoder~~~~} \\ f:~\complex^{n_1 \times n_2} \to \complex^m}$};
    \node (b) [left of=a,node distance=2.8cm, coordinate] {};
    \node [dot, pin={[init]above:{$Z$}}] (d) [right of=a] {$+$};
    \node [int] (c) [right of=d] {$\substack{\textrm{~~~~Decoder~~~~} \\ g:~\complex^m \to \complex^{n_1 \times n_2}}$};
    \node [coordinate] (end) [right of=c, node distance=2.8cm]{};
    \path[->] (b) edge node {$X=UV^*$~} (a);
    \draw[->] (c) edge node {~$\hat X = \hat U \hat V^*$} (end);
    \path[->] (a) edge node {$Y$} (d);
    \path[->] (d) edge node {$\hat Y$} (c);
\end{tikzpicture}
	\caption{Bayesian setup for the compressed sensing problem \prettyref{eq:model}, allowing possibly non-linear measurement mechanisms.}
	\label{fig:bayesian}
\end{figure}
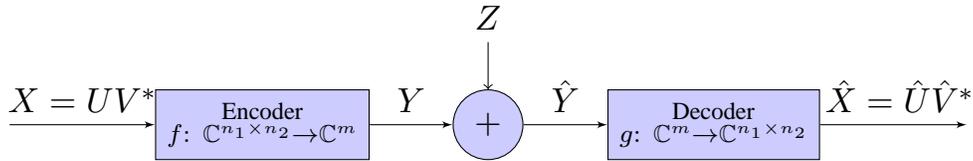

\begin{theorem}
	\label{thm:lb}
Let $f$ satisfies the average power constraint \prettyref{eq:avg-power}. Let $g$ achieve the reconstruction error
\begin{equation}
	\Expect \fnorm{\hat{X} - X}^2  \leq r D 
\label{eq:loss}
\end{equation}
where $D > 0$. Then for any $s_1,s_2,r\in\naturals$ with $r \leq \min\{s_1,s_2\}$,
\begin{equation}
m \geq  \pth{(s_1+s_2) r - \frac{3r^2}{2}}  \frac{ \log \frac{c}{ D}}{\log \pth{1+\frac{1}{\sigma^2}}}.	
	\label{eq:lb}
\end{equation}
where $c$ is a universal constant.
\end{theorem}
\begin{proof}
	\prettyref{sec:pf-lb}.
\end{proof}

As a consequence of \prettyref{thm:lb}, we conclude that the minimum number of measurements (sample complexity) for robust reconstruction of doubly sparse rank-$r$ matrices, \ie, achieving finite noise sensitivity with $D \leq C \sigma^2$ for all $\sigma^2 > 0$ and some absolute constant $C$, must satisfy
\begin{equation}
m \geq \pth{s_1+s_2  - \frac{3r}{2}} r.	
	\label{eq:lb-stable}
\end{equation}
which holds even if the recovery algorithm is allowed to depend on the noise level.
This follows from \prettyref{eq:lb} by sending $\sigma\to0$.
When $r=1$, since $m$ is an integer, we conclude that stable recovery of doubly-sparse rank-one matrices requires at least
\begin{equation}
m \geq s_1+s_2-1.	
	\label{eq:lb-stable1}
\end{equation}
number of measurements.

In view of the lower bounds \prettyref{eq:lb-stable}--\prettyref{eq:lb-stable1}, we conclude that the number of measurements in Theorems~\ref{thm:pgrip_init_ds_iidG} and \ref{thm:pgrip_ds_iidG_rankr} are optimal within logarithmic factors. To see this, we argue that the additional spikiness assumption on the matrix imposed in Theorems~\ref{thm:pgrip_init_ds_iidG} and \ref{thm:pgrip_ds_iidG_rankr} does not change the sample complexity of the problem. To adapt \prettyref{thm:lb} to this scenario, one can simply consider $X=UV^*$, where $U = \frac{1}{\sqrt{2}}\Big[\begin{smallmatrix} I_r \\ \tU \end{smallmatrix}\Big]$ and $V = \frac{1}{\sqrt{2}}\Big[\begin{smallmatrix} I_r \\ \tV \end{smallmatrix}\Big]$, and $\tU$ and $\tV$ are distributed according to \prettyref{thm:lb} with ambient dimension $n_1$ (resp.~$n_2$) replaced by $n_2-r$ (resp.~$n_1-r$). 
Then U and V satisfy the spikiness assumptions in Theorems~\ref{thm:pgrip_init_ds_iidG} and \ref{thm:pgrip_ds_iidG_rankr}. Note that since $r \leq \min\{s_1,s_2\}$ by definition, the lower bound \prettyref{eq:lb-stable} is always at least $\frac{1}{4}r(s_1+s_2)$.
Therefore as long as the rank is not too high, namely, $r \leq \min\{n_1-s_1, n_2-s_2\}$, the information-theoretic lower bounds \prettyref{eq:lb-stable}--\prettyref{eq:lb-stable1} continue to hold and we conclude that SPF algorithm achieves the fundamental limits for stably recovering sparse low-rank matrices within logarithmic factors.

\begin{remark}
Compared to \cite[Theorem 10]{WV12} for compressed sensing of sparse vectors proved under the high-dimensional scaling, the lower bound \prettyref{eq:lb-stable} is non-asymptotic. Moreover, even if we relax the stability requirement in \prettyref{eq:loss} to $\Expect \fnorm{\hat{X} - X}^2  \leq \sigma^{2\alpha}$ for some $\alpha \in (0,1)$, the number of measurements still satisfies the lower bound $m \geq \alpha(s_1+s_2-2r)r$.	
\end{remark}

\begin{remark}
The lower bound \prettyref{eq:lb-stable} can be heuristically understood by counting the number of (real) degrees of freedom in $X=UV^*$, which turns out to be $2(s_1+s_2)r-r^3$. Note that
the Stiefel manifold $\sfV(\complex^{n},r)$ has topological (real) dimension\footnote{This follows from choosing the first column $v_1$ of $V$ from the complex unit sphere which has $2n_1-1$ real variables, then the second column $v_2\perp v_1$ which gives two equations (real and imaginary parts) and leaves $2n_1-3$ free variables, etc.}
 $\sum_{i=1}^r (2n-2i+1) = 2nr-r^2$.
Since the non-zero parts of $U$ and $V$ belongs to $\sfV(\complex^{s_1},r)$ and $\sfV(\complex^{s_2},r)$ respectively,
the total number of the (real) degrees of freedom in $U$ and $V$ is $2(s_1+s_2)r-2r^2$.
However, since $UV^* = UR (VR)^*$, we need to quotient out the orthogonal group $O(r)$ in $\complex^r$, which has dimension $\sum_{i=1}^r (2r-2i+1) = r^2$. Therefore the total degrees of freedom in $X$ is $\dim \sfV(\complex^{s_1},r) + \dim \sfV(\complex^{s_1},r) - \dim O(r) = 2(s_1+s_2)r-3r^2$.
Hence, intuitively, we expect as least half of this number of \emph{complex} linear measurements for stable recovery. \prettyref{thm:lb} gives a rigorous information-theoretic justification of this heuristic.
See \prettyref{rmk:dof-vv} for more detailed discussion on counting degrees of freedom.
	\label{rmk:dof}
\end{remark}

\begin{remark}[Bayesian v.s. minimax lower bound]
The lower bound in \prettyref{thm:lb} is obtained under a Bayesian setup where the left and right singular vectors have uniformly drawn support and non-zeros. On the other hand, the upper bounds in Theorems~\ref{thm:pgrip_init_ds_iidG} and \ref{thm:pgrip_ds_iidG_rankr} are obtained under an adversarial setting where both the unknown matrix $X$ and the noise $z$ are deterministic. It is unclear whether the extra logarithmic factor for the number of Gaussian measurements is necessary in a minimax setting, where, for instance, the noise is additive Gaussian and the unknown rank-one matrix $X$ is adversarial.
	\label{rmk:bayesian-minimax}
\end{remark}

The proof of \prettyref{thm:lb} relies on a rate-distortion lower bound for random subspaces, given in \prettyref{thm:RDv} in \prettyref{sec:pf-lb}, which might be of independent interest.

\section{Numerical Results}
\label{sec:numres}
In this section, we compare the empirical performance of SPF
to those of PF and other popular recovery algorithms based on convex optimization \cite{oymak2015simultaneously}.
The simulation setup is as follows: Let
$\A$ be an i.i.d. Gaussian measurement operator.
The unknown sparse rank-one matrix is generated with support uniformly drawn at random
and the nonzero elements of the singular vectors are uniform on the complex sphere.
The recovery performance of different procedures was compared in a Monte-Carlo study,
averaging over 100 instances of signal matrices and measurement operators drawn at random.

\subsection{Row sparsity}
We compare the recovery performance of a row-sparse rank-one matrix by SPF in both noiseless and noisy cases
against PF as well as the following convex-optimization approaches:
basis pursuit with the row-sparsity prior (BP\_RS), with the low-rank prior (BP\_LR),
and with both priors combined (BP\_RSLR), \ie,
\begin{align*}
\text{\tt BP\_RS}: {} & \qquad \widehat{X} = \argmin_{\widetilde{X}} \left\{ \norm{{\widetilde{X}}}_{1,2} :~ \A({\widetilde{X}}) = b \right\} \\
\text{\tt BP\_LR}: {} & \qquad \widehat{X} = \argmin_{\widetilde{X}} \left\{ \norm{{\widetilde{X}}}_* :~ \A({\widetilde{X}}) = b \right\} \\
\text{\tt BP\_RSLR}: {} & \qquad \widehat{X} = \argmin_{\widetilde{X}} \left\{ \max\left( \frac{\norm{{\widetilde{X}}}_{1,2}}{\norm{X}_{1,2}} ,~ \frac{\norm{{\widetilde{X}}}_*}{\norm{X}_*} \right) :~ \A({\widetilde{X}}) = b \right\}
\end{align*}
where $\norm{\cdot}_*$ denotes the nuclear norm (sum of singular values), and $\norm{\cdot}_{1,2}$ denotes the mixed norm (sum of row norms).
BP\_RSLR corresponds to the optimal convex approach among all methods that minimize combinations of the nuclear norm and mixed norm \cite{oymak2015simultaneously}. The weights used in BP\_RSLR are functions of the unknown matrix $X$ and therefore BP\_RSLR is considered as an oracle method.

\begin{figure}
  \centering
  \begin{tabular}{m{0.18\textwidth}m{0.18\textwidth}m{0.18\textwidth}m{0.18\textwidth}m{0.18\textwidth}}
  \includegraphics[width=30mm]{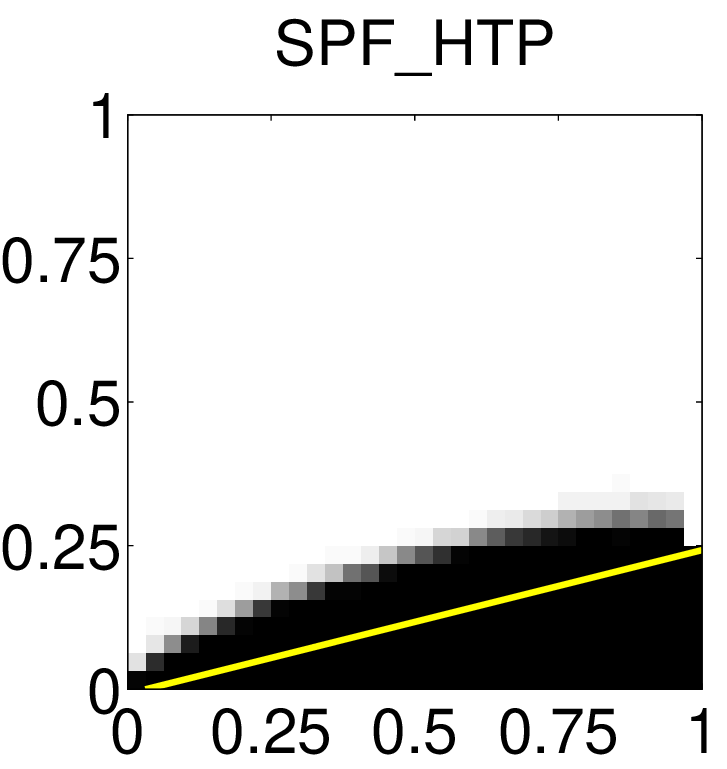} &
  \includegraphics[width=30mm]{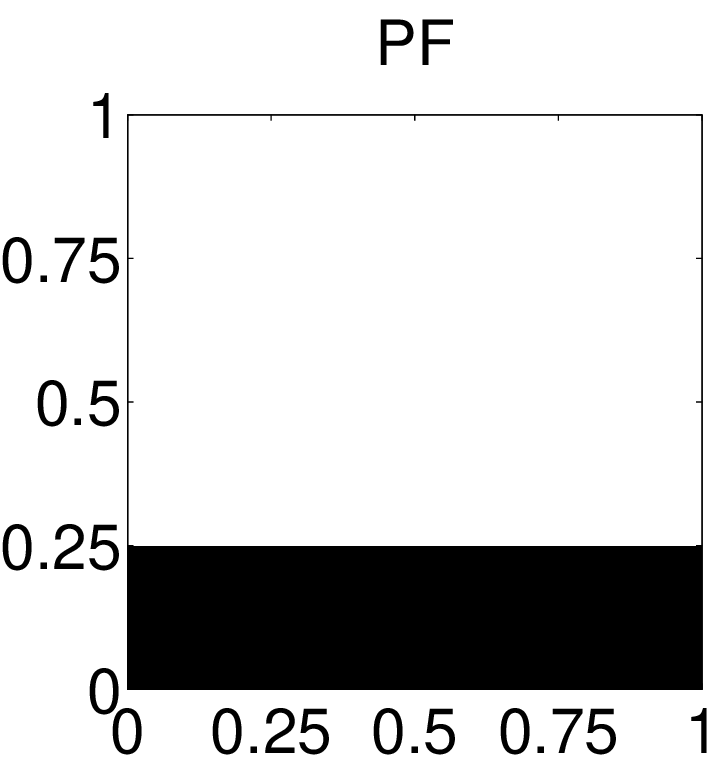} &
  \includegraphics[width=30mm]{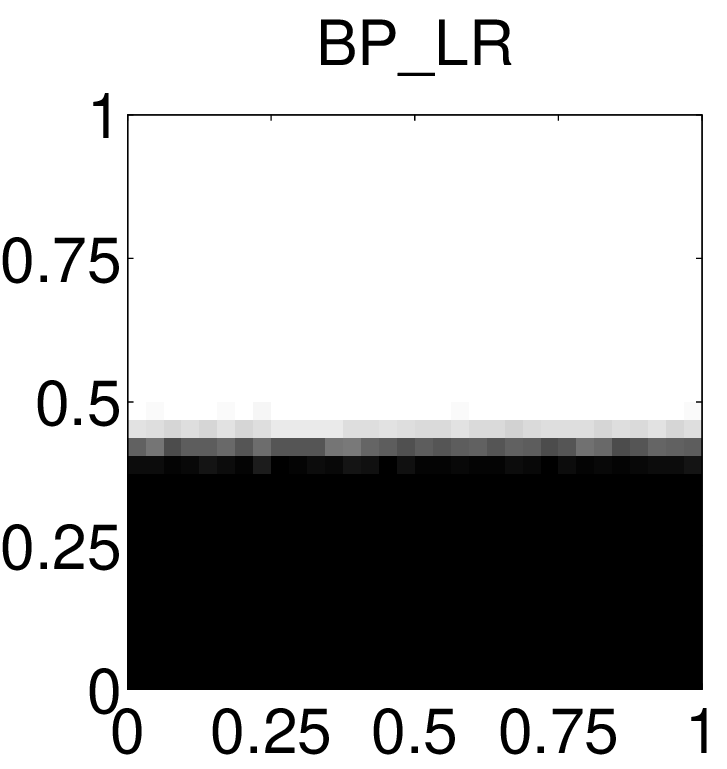} &
  \includegraphics[width=30mm]{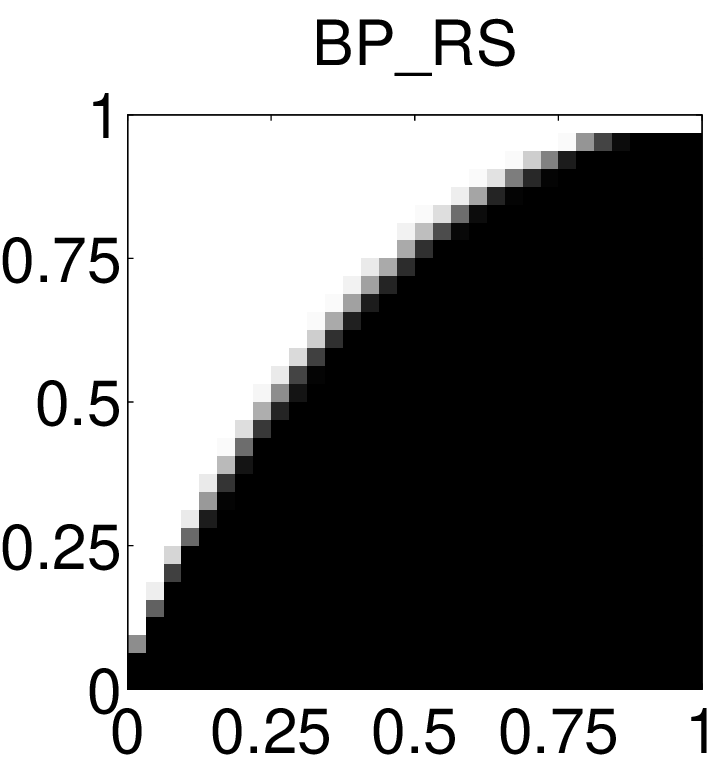} &
  \includegraphics[width=30mm]{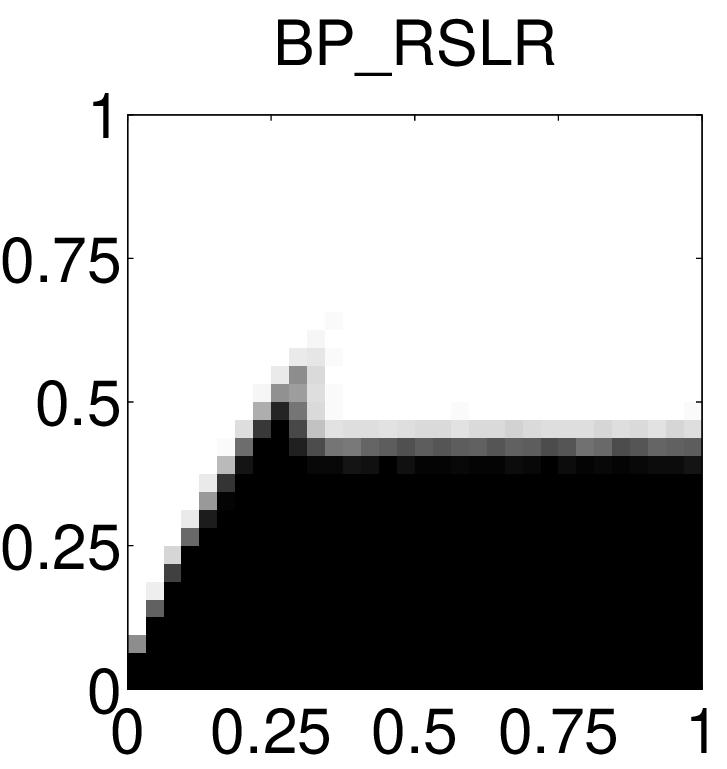} \\
  \multicolumn{5}{c}{\footnotesize (a) $n_2 = 4$}
  \end{tabular}
  \\
  \vspace{3mm}
  \begin{tabular}{m{0.18\textwidth}m{0.18\textwidth}m{0.18\textwidth}m{0.18\textwidth}m{0.18\textwidth}}
  \includegraphics[width=30mm]{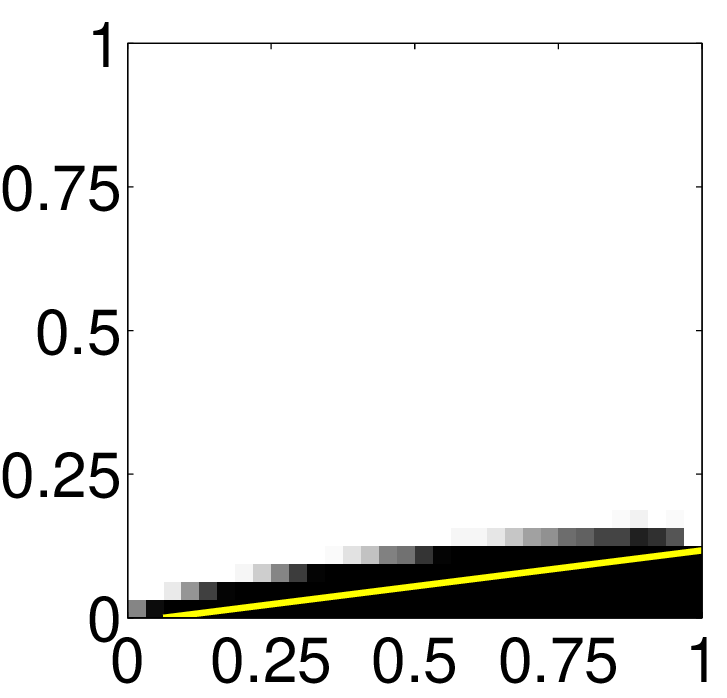} &
  \includegraphics[width=30mm]{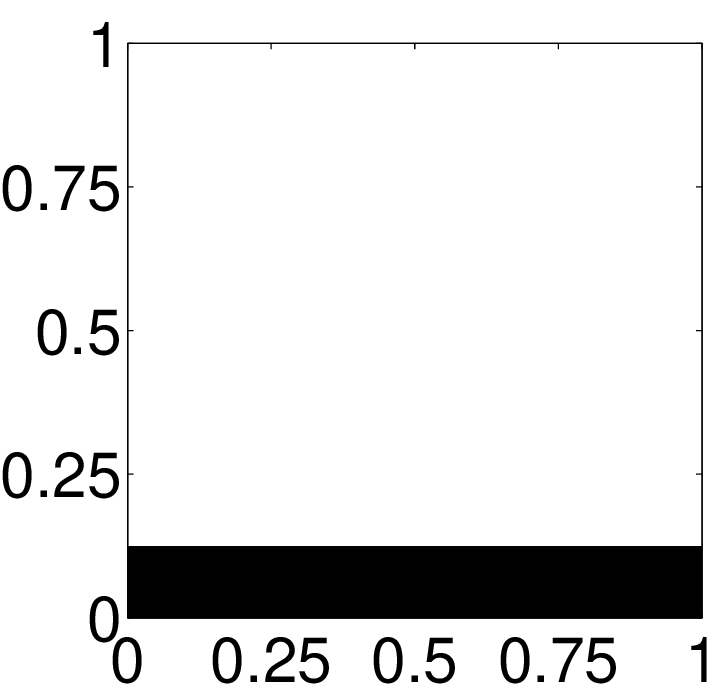} &
  \includegraphics[width=30mm]{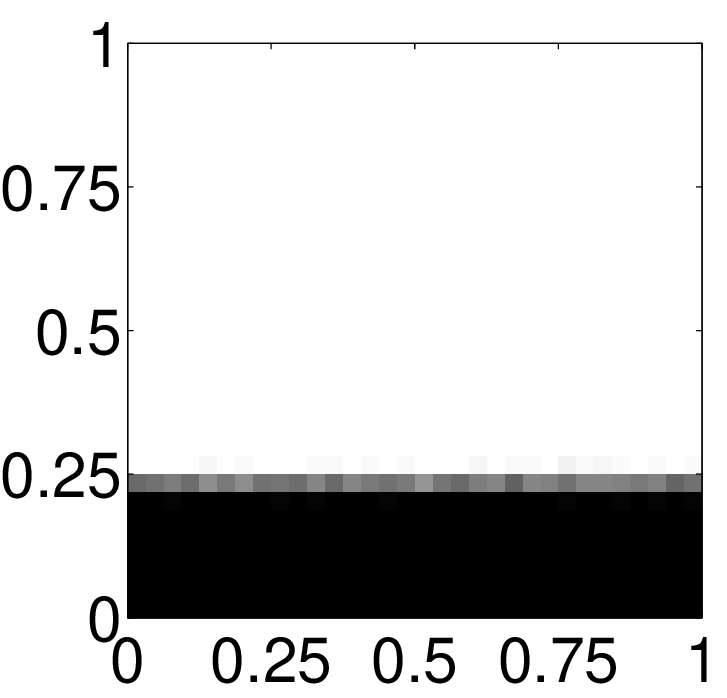} &
  \includegraphics[width=30mm]{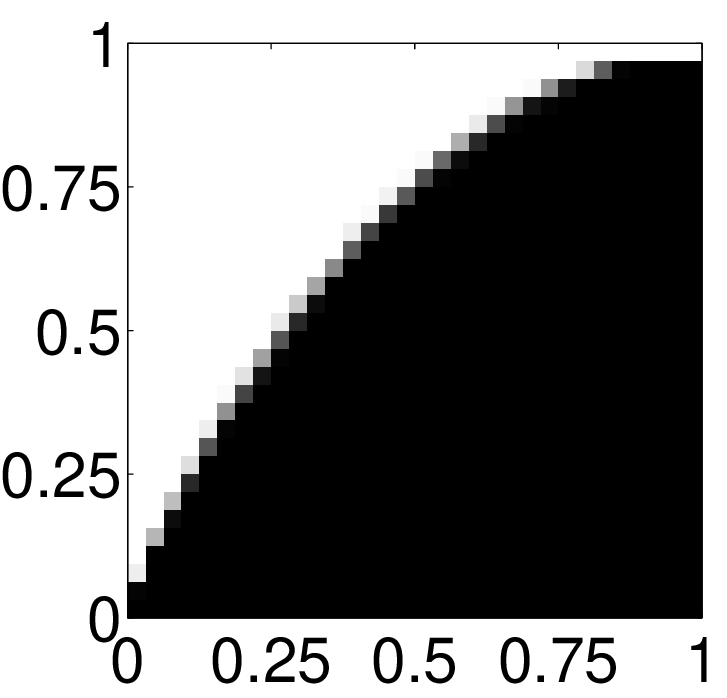} &
  \includegraphics[width=30mm]{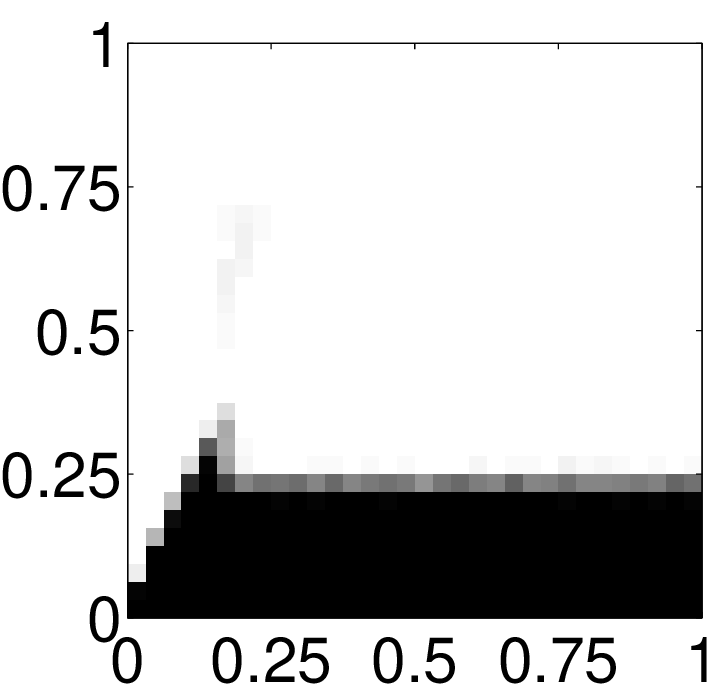} \\
  \multicolumn{5}{c}{\footnotesize (b) $n_2 = 8$}
  \end{tabular}
  \caption{Phase transition of the empirical success rates in the recovery of a rank-1 and row-$s$-sparse matrix of size $n_1 \times n_2$ ($x$-axis: $s/n_1$, $y$-axis: $m/n_1 n_2$, $n_1 = 128$). The empirical success rate is on a grayscale (black for zero and white for one). The yellow line corresponds to the fundamental limit $m \geq s+n_2-2$.}
  \label{fig:ptm_s}
\end{figure}

\begin{figure}[ht]
  \centering
  \begin{subfigure}[b]{0.45\textwidth}
  \centering
  \includegraphics[width=35mm]{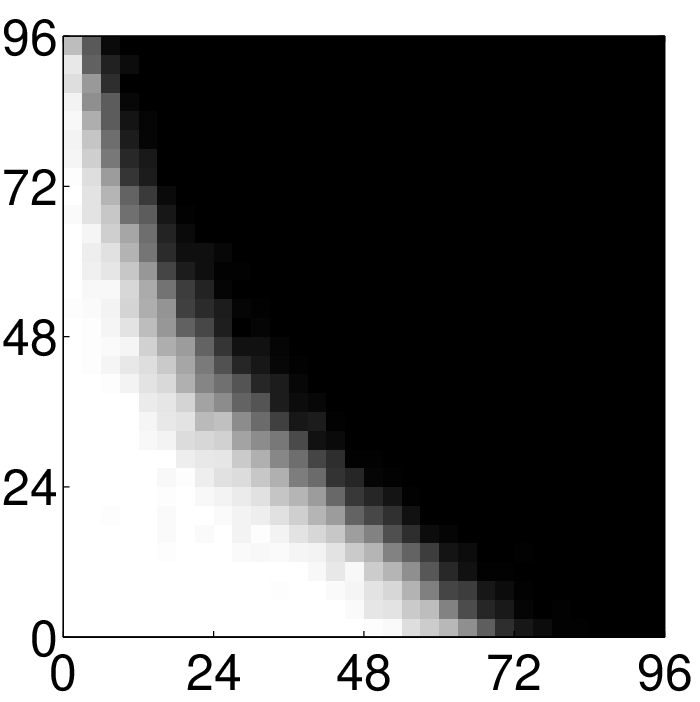}
  \caption{SPF}
  \end{subfigure}
  \begin{subfigure}[b]{0.45\textwidth}
  \centering
  \includegraphics[width=35mm]{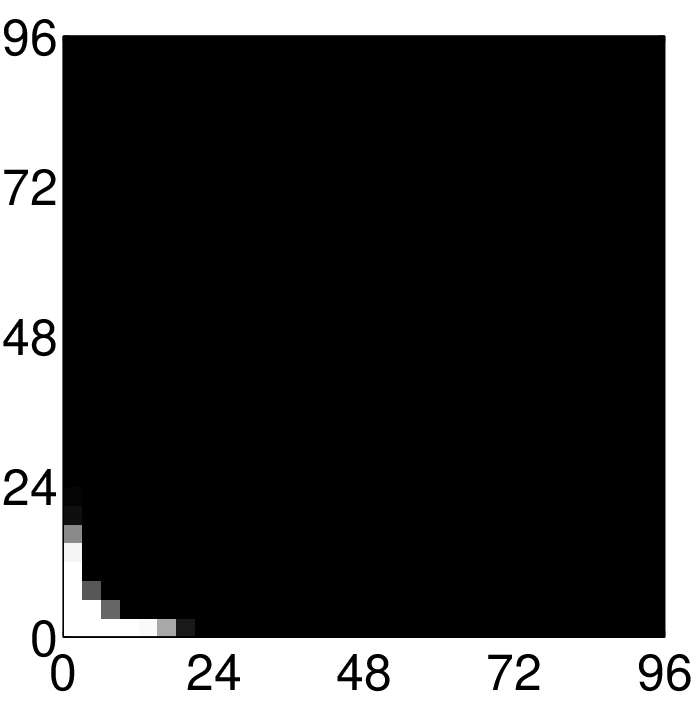}
  \caption{BP\_RSLR}
  \end{subfigure}
  \caption{Phase transition for the empirical success rate in the recovery of a rank-1 and row-$s$-sparse matrix of size $n_1 \times n_2$ when the number of measurements $m$ is fixed ($x$-axis: $s$, $y$-axis: $n_2$, $m = n_1 = 128$).} 
  \label{fig:pt2_spf}
\end{figure}


\paragraph{Noiseless case}
The empirical phase transitions of success rates of various algorithms are given in \prettyref{fig:ptm_s},
where success corresponds to achieving a reconstruction signal-to-noise-ratio (SNR) higher than 50 dB.
SPF improves the performance of PF significantly by exploiting the sparsity structure
and outperforms all three convex approach.
Moreover, empirically the phase transition boundary of SPF is close to the fundamental limit $m\geq s_1+n_2-2$ given in \prettyref{sec:lb}.
A detailed comparison of different phase transition boundaries is given below:
\begin{itemize}
	\item
Both PF and BP\_LR exploit the rank-one constraint requiring
\begin{equation}
m \geq C \max(n_1,n_2)
	\label{eq:mpf}
\end{equation}
number of measurements where $C$ denotes some generic absolute constant. The difference is that PF iterates between left and right singular vectors and BP\_LR promotes low-rankness by minimizing the nuclear norm. As the number of columns $n_2$ in the unknown matrix increases, the number of singular values $\min(n_1,n_2)$ increases; hence, the rank-one constraint becomes more informative in the sense that it further narrows down the solution set. Therefore both PF and BP\_LR benefit from a larger $n_2$. However, since neither procedure takes into account the row sparsity, their performance does not improve with smaller relative sparsity $s/n_1$. This results in a flat phase transition boundary.
On the other hand, SPF achieves a better performance by exploiting the sparsity in the left singular vectors.

\item
Similarly, BP\_RS, which minimizes the mixed norm to promote row sparsity,
provides perfect reconstruction with number of measurements
\begin{equation}
m \geq C s n_2 \log(e n_1/s)
	\label{eq:mbprs}
\end{equation}
which decreases as the relative sparsity $s/n_1$ decreases,
but its performance is indifferent to the relative low-rankness $1/\min(n_1,n_2)$.

\item
BP\_RSLR, which mixes BP\_RS and BP\_LR optimally, performs only as well as, but no better than, the best of BP\_LR and BP\_RS, achieving a success region as the union of the two.
It requires number of measurements to exceed the smallest of the right-hand sides of \prettyref{eq:mpf} and \prettyref{eq:mbprs}, which can far exceed the fundamental limit $m \geq s_1+n_2$.
The intrinsic suboptimality of BP\_RSLR stems from the following observation: Although equivalent under the rank-one constraint, row-sparsity of the matrix is excessive relaxation for the sparsity of the left singular vector, and simply promoting low-rankness seems to be not enough. By thresholding the singular vectors in the power iteration, SPF achieves near-optimal performance.
\end{itemize}

Figure~\ref{fig:pt2_spf} compares the phase transition of empirical success rates for SPF and BP\_RSLR
when the number of columns $n_2$ and the sparsity level $s$ vary while the number of measurements $m$ is fixed to be 128.
SPF outperforms BP\_RSLR by succeeding roughly under the condition $m \geq 4 (s+n_2)$.
This empirical observation is aligned with the theoretical analysis in this paper.

\paragraph{Noisy case}
Figure~\ref{fig:ptm_n} visualizes the performance of SPF in the presence of noise.
While the theoretical analysis of SPF in this paper is restricted to the conservative case where SNR exceeds an absolute constant level,
empirically the reconstruction error of SPF is robust with respect to increased noise level.
In particular, with $m \geq 3 (s + n_2)$, the noise amplification in the reconstruction remains less than 1 for all SNR great than 6 dB.
\begin{figure}[ht]
  \centering
  \begin{subfigure}[b]{0.45\textwidth}
  \centering
  \includegraphics[width=50mm]{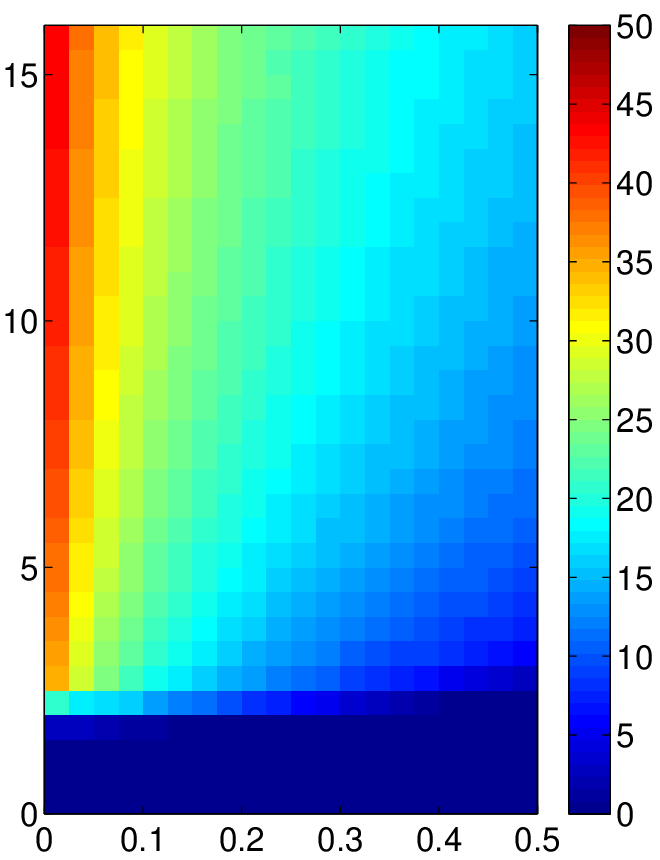}
  \caption{Reconstruction SNR}
  \end{subfigure}
  \begin{subfigure}[b]{0.45\textwidth}
  \centering
  \includegraphics[width=50mm]{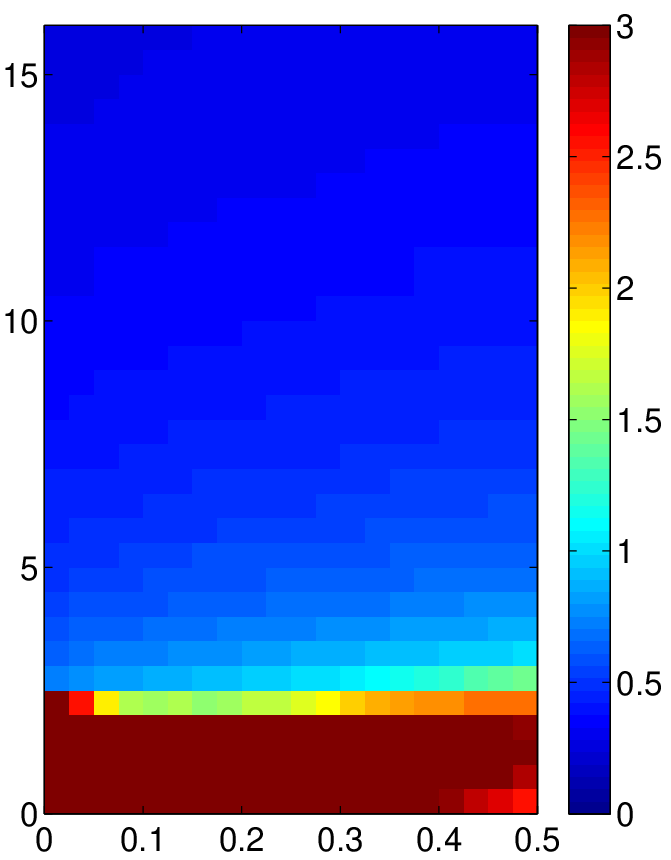}
  \caption{Noise amplification}
  \end{subfigure}
  \caption{Relative error and noise amplification constant of SPF for varying noise strength $\nu$ ($x$-axis: $\nu$, $y$-axis: $m/(s+n_2)$, $n_1 = 256$, $s_1 = 32$, $n_2 = 64$). Reconstruction SNR = $\min\big\{50, 20 \log_{10} \frac{\fnorm{X}}{\fnorm{\widehat{X} - X}} \big\}$. Noise amplification = $\min\big\{3, \log_{10} \frac{\fnorm{\widehat{X} - X}}{\nu \fnorm{X}}\big\}$.}
  \label{fig:ptm_n}
\end{figure}

\subsection{Row and column sparsity}

When the unknown rank-one matrix $X$ is both column and row-sparse,
we compare the recovery performance by SPF
to the recovery using the following convex optimization approaches:
\begin{align*}
\text{\tt BP\_DS}: {} & \qquad \widehat{X} = \arg\min_Z \left\{ \max\left( \frac{\norm{Z}_{1,2}}{\norm{X}_{1,2}} ,~ \frac{\norm{Z^*}_{1,2}}{\norm{X^*}_{1,2}} \right)  :~ \A(Z) = b \right\} \\
\text{\tt BP\_DSLR}: {} & \qquad \widehat{X} = \arg\min_Z \left\{ \max\left( \frac{\norm{Z}_{1,2}}{\norm{X}_{1,2}} ,~ \frac{\norm{Z^*}_{1,2}}{\norm{X^*}_{1,2}} ,~ \frac{\norm{Z}_*}{\norm{X}_*} \right) :~ \A(Z) = b \right\}.
\end{align*}
Both BP\_DS and BP\_DSLR use the oracle optimal weights, which are functions of the unknown matrix $X$.

\begin{figure}[ht]
  \centering
  \begin{tabular}{m{0.18\textwidth}m{0.18\textwidth}m{0.18\textwidth}m{0.18\textwidth}m{0.18\textwidth}}
  \includegraphics[width=27mm]{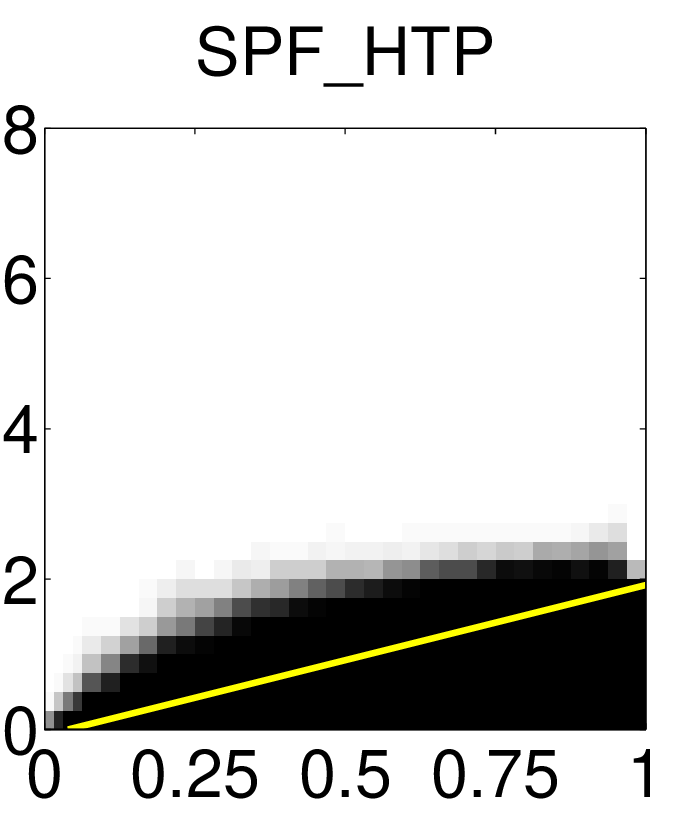} &
  \includegraphics[width=27mm]{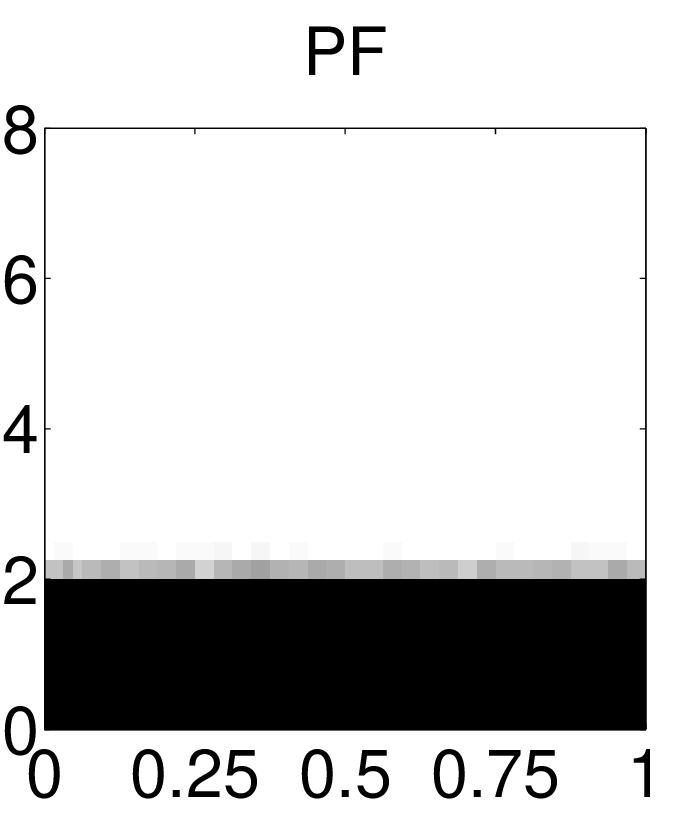} &
  \includegraphics[width=27mm]{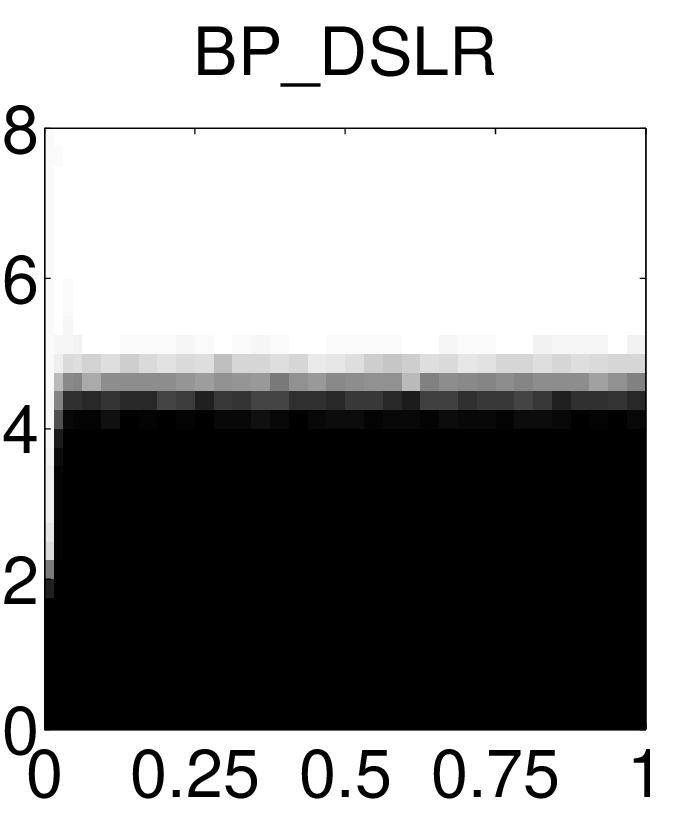} &
  \includegraphics[width=27mm]{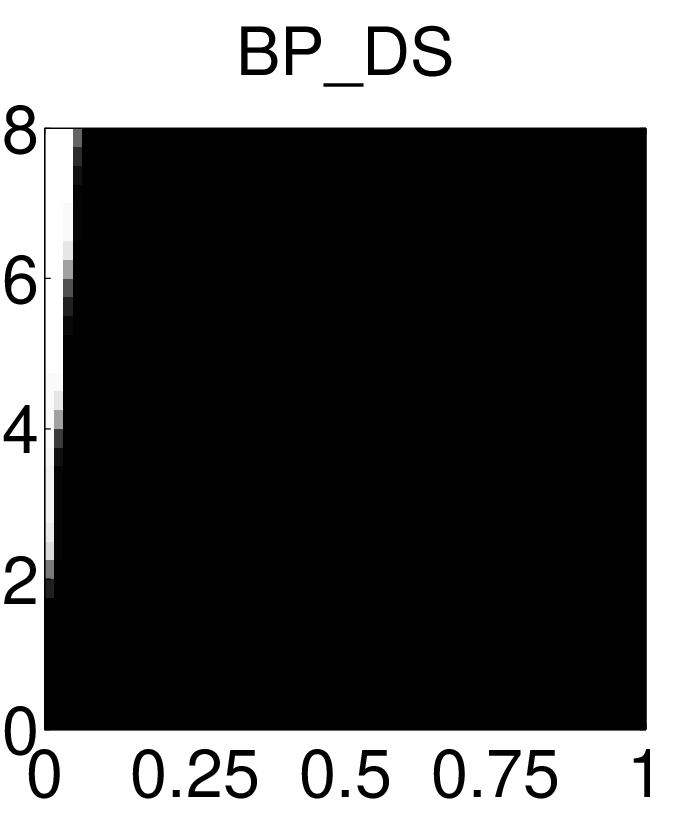} &
  \includegraphics[width=27mm]{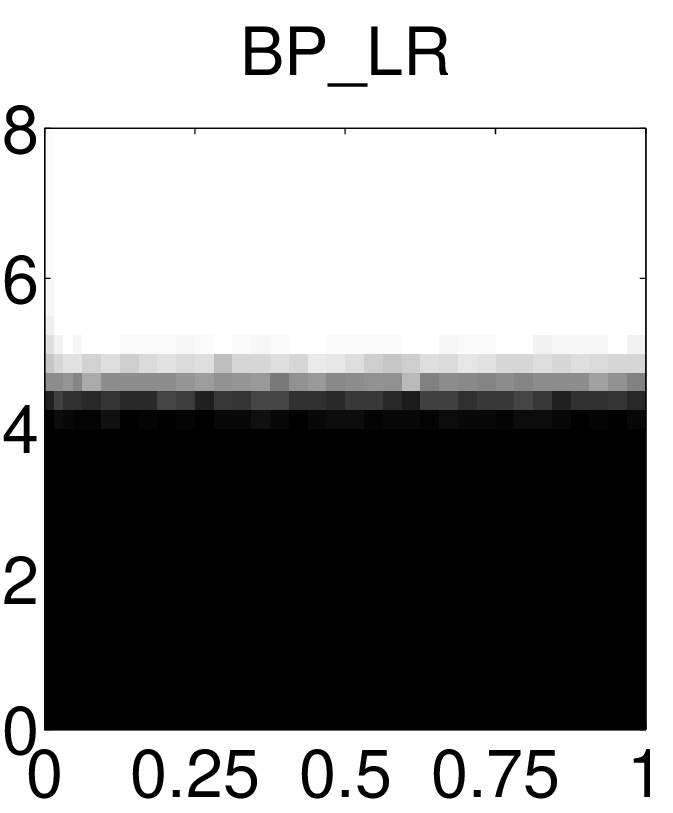}
  \end{tabular}
  \caption{Phase transition for the empirical success rate in the recovery of a rank-1 and doubly $(s,s)$-sparse matrix of size $n \times n$ ($x$-axis: $s/n$, $y$-axis: $m/n$, $n = 64$). The yellow line corresponds to the fundamental limit $m=2s$.}
  \label{fig:ptm_ds}
\end{figure}

Figure~\ref{fig:ptm_ds} plots the empirical success rates of SPF and competing convex approaches.
The unknown rank-one matrix of size $n \times n$ is assumed doubly $(s,s)$-sparse.
In this setup, BP\_DSLR performs as well as BP\_DS in the very sparse case ($s=1$).
Otherwise, the rank-one prior dominates the double sparsity prior and the performance of BP\_DSLR coincides with that of BP\_LR. On the contrary, SPF significantly improves on PF by exploiting the double sparsity prior. \prettyref{fig:pt2_dspfa} (resp. \prettyref{fig:pt2_dspfb}) shows that,
when $m = 1.5 n$ (resp. $m=2 n$), SPF empirically succeed roughly under the condition $m \geq 6 s$,
which is aligned with the performance guarantee of SPF in \prettyref{sec:ub} and the lower bound $m \geq 2s-2$ in \prettyref{sec:lb}.
As shown in \prettyref{fig:pt2_dspfc}, even with more measurements,
BP\_DSLR completely failed in this setup except for the corner case of either $s_1 = 1$ or $s_2 = 1$.

\begin{figure}[ht]
  \centering
  \begin{subfigure}[b]{0.31\textwidth}
  \centering
  \includegraphics[width=35mm]{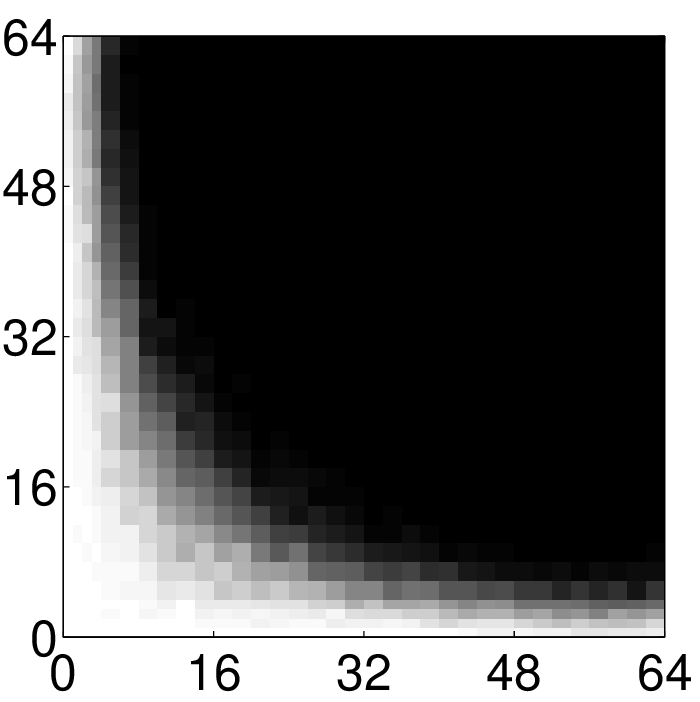}
  \caption{SPF ($m = 96$)}
  \label{fig:pt2_dspfa}
  \end{subfigure}
  \begin{subfigure}[b]{0.31\textwidth}
  \centering
  \includegraphics[width=35mm]{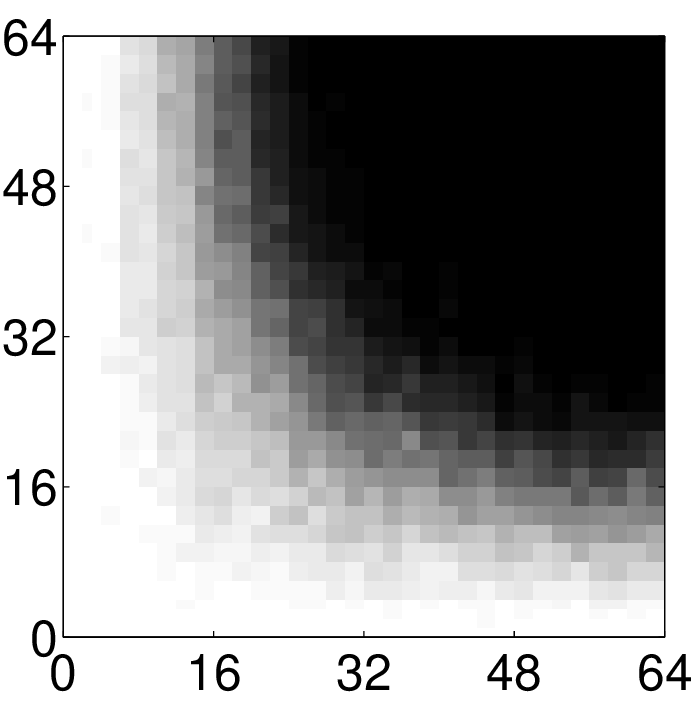}
  \caption{SPF ($m = 128$)}
  \label{fig:pt2_dspfb}
  \end{subfigure}
  \begin{subfigure}[b]{0.31\textwidth}
  \centering
  \includegraphics[width=35mm]{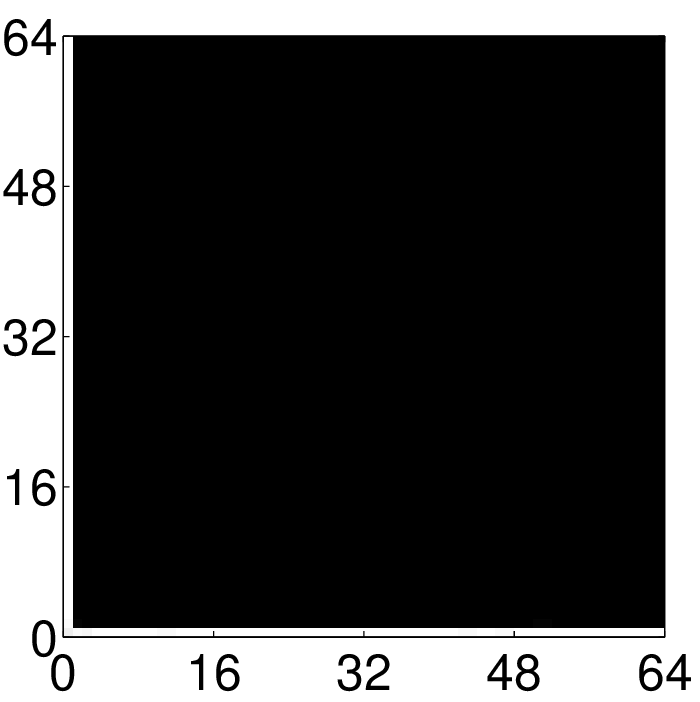}
  \caption{BP\_DSLR ($m = 192$)}
  \label{fig:pt2_dspfc}
  \end{subfigure}
  \caption{Phase transition for the empirical success rate in the recovery of a rank-1 and doubly $(s_1,s_2)$-sparse matrix of size $64 \times 64$ using SPF when the number of measurements $m$ is fixed ($x$-axis: $s_1$, $y$-axis: $s_2$).}
  \label{fig:pt2_dspf}
\end{figure}

\subsection{Sparse and low rank matrices}
\label{subsec:numres:lrds}

In the scenario where the rank of the unknown sparse matrix is low but larger than 1, we compare the recovery performance of SCSPF to the natural rank-$r$ extension of SPF without subspace concatenation (Algorithm~\ref{alg:rspf}) and to BP\_DSLR. Figure~\ref{fig:ptm_dslr_kappa1} compares the three algorithms when the condition number of the unknown matrix is ideally fixed as 1. SCSPF and SPF outperform BP\_DSLR similarly to the previous sections. The empirical phase transition occurs at a sample complexity proportional to the rank of the unknown matrix, which is aligned with the presented theory for SCSPF. Although the sample complexity for the performance guarantee for SCSPF in Theorem~\ref{thm:pgrip_ds_iidG_rankr} is proportional to the square of the condition number $\kappa$, as shown in Figure~\ref{fig:ptm_dslr_kappa5}, when $\kappa$ is small, 5 in this figure, the empirical performance of the three algorithms looked similar to the ideal case in Figure~\ref{fig:ptm_dslr_kappa1}. In both figures, in fact, the natural rank-$r$ SPF provided a better empirical performance compared to SCSPF. We suspect that the subspace concatenation in SCSPF was necessary because of artifacts in our proof techniques. It might be possible that a more careful analysis of the natural rank-$r$ SPF provide a performance guarantee at the sample complexity depending linearly on the rank.

\begin{figure}[ht]
  \centering  
  \begin{subfigure}[b]{0.31\textwidth}
  \centering
  \includegraphics[height=50mm]{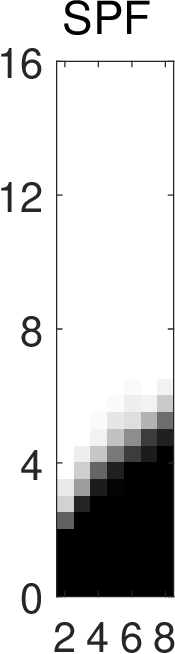}
  \end{subfigure}  
  \begin{subfigure}[b]{0.31\textwidth}
  \centering
  \includegraphics[height=50mm]{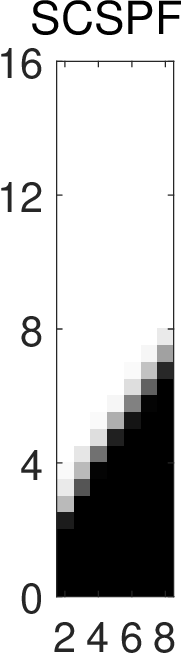}
  \end{subfigure}  
  \begin{subfigure}[b]{0.31\textwidth}
  \centering
  \includegraphics[height=50mm]{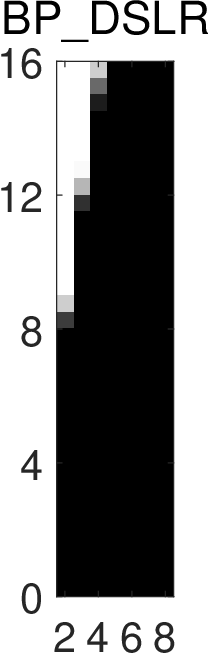}
  \end{subfigure}  
  \caption{Phase transition for the empirical success rate in the recovery of a rank-$r$ and doubly $(s,s)$-sparse matrix of size $n \times n$ ($x$-axis: $r$, $y$-axis: $m/n$, $n = 64$, $s = 16$, $\kappa = 1$).}
  \label{fig:ptm_dslr_kappa1}
\end{figure}

\begin{figure}[ht]
  \centering
  \begin{subfigure}[b]{0.31\textwidth}
  \centering
  \includegraphics[height=50mm]{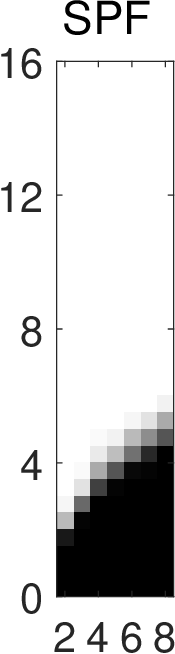}
  \end{subfigure}  
  \begin{subfigure}[b]{0.31\textwidth}
  \centering
  \includegraphics[height=50mm]{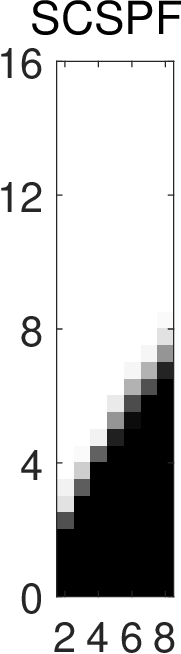}
  \end{subfigure}  
  \begin{subfigure}[b]{0.31\textwidth}
  \centering
  \includegraphics[height=50mm]{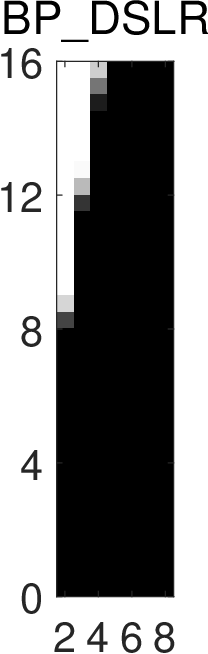}
  \end{subfigure}  
  \caption{Phase transition for the empirical success rate in the recovery of a rank-$r$ and doubly $(s,s)$-sparse matrix of size $n \times n$ ($x$-axis: $r$, $y$-axis: $m/n$, $n = 64$, $s = 16$, $\kappa = 5$).}
  \label{fig:ptm_dslr_kappa5}
\end{figure}

\section{Conclusion}
\label{sec:conclusion}
We proposed an alternating minimization algorithm called sparse power factorization (SPF)
that reconstructs sparse rank-one matrices from their linear measurements. We showed that under variants of the restricted isometry property corresponding to the underlying matrix models, SPF provides provable performance guarantees. Information-theoretic determination of the minimal number of measurements is another contribution of this paper. In particular, when the measurements are given using an i.i.d. Gaussian measurement operator, SPF with the initialization $\vzstar$ provides a near-optimal performance guarantee that holds with a number of measurements, which exceed the fundamental lower bound only by a logarithmic factor. On the other hand, when the unknown matrix has sparse singular vectors with fast-decaying magnitudes, the performance guarantee for SPF with the computationally efficient initialization $v_0^\text{\rm Th}$ holds with the same number of measurements. Similar performance guarantees in the context of blind deconvolution are presented in a companion paper \cite{LeeLJB2015}.
Furthermore, through numerical experiments, we showed that the empirical performance of SPF dominates that of competing convex approaches, which is consistent with the theoretical analysis.

We conclude the paper by discussing the computational aspect of our sensing problem. As mentioned in \prettyref{sec:ub}, the performance of SPF hinges on the initialization. The near optimal number of measurements $O((s_1+s_2)\log \max(e n_1/s_1, e n_2/s_2))$ in \prettyref{thm:pgrip_init_ds_iidG} with no additional assumption is obtained using the initial value $\vzstar$ defined in \prettyref{sec:alg-spf}, which can be expensive to compute. In view of the failure of convex programming via nuclear/mixed norms \cite{oymak2015simultaneously}, as well as the recently established hardness of sparse PCA \cite{Berthet13} and biclustering \cite{MW13b}, it is an open problem to determine whether achieving stable recovery using $O(s_1+s_2)$ or $O((s_1+s_2)\log \max(e n_1/s_1, e n_2/s_2))$ measurements is computationally intractable.

\section{Proofs}
\label{sec:pf}

\subsection{Proof of Theorem~\ref{thm:conv_spf}}
\label{subsec:pgspf}
For $t\geq 0$, denote the angle between $u$ and $u_t$ (resp. between $v$ and $v_t$)
by $\theta_t$ (resp. $\phi_t$):
\begin{align}
\theta_t \triangleq \cos^{-1} \left(\frac{|u^* u_t|}{\norm{u}_2 \norm{u_t}_2} \right)
\quad \text{and} \quad
\phi_t \triangleq \cos^{-1} \left( \frac{|v^* v_t|}{\norm{v}_2 \norm{v_t}_2} \right), \label{eq:defangles}
\end{align}
which are in $[0, \pi/2]$.

We derive recursive relations between the sequences $(\theta_t)_{t \in \bbN}$ and $(\phi_t)_{t \in \bbN}$. This roadmap is similar to that in a previous work \cite{JaiNS2012}
deriving a performance guarantee in terms of the rank-$2r$ RIP for the recovery of a rank-$r$ matrix using PF.

The updates of $u_t$ and/or $v_t$ using HTP, in the presence of the corresponding sparsity prior,
are new components in SPF compared to PF, and require a different analysis (Lemma~\ref{lemma:update_theta}). Although the least squares update of $u_t$ or $v_t$, in the absence of the corresponding sparsity prior, is common to both SPF and PF, we provide a new analysis of this step (Lemma~\ref{lemma:update_phi}) that improves on the previous analysis \cite{JaiNS2012} by sharpening the inequalities involved.

First, we analyze the update of $u_t$ using HTP. To this end, we present a modified guarantee for HTP in the following lemma.

\begin{lemma}[RIP-guarantee for HTP]
Let $b = \Phi x + z \in \bbC^m$ denote the noisy measurement vector of an $s$-sparse $x \in \bbC^n$ obtained using a sensing matrix $\Phi \in \bbC^{m \times n}$. Let $(x_k)_{k \in \bbN}$ be the iterates of HTP using a perturbed sensing matrix $\widehat{\Phi}$. Suppose that $\widehat{\Phi}$ satisfies the $3s$-sparse RIP with isometry constant $\delta < 0.5$,
and that there exists $\vartheta \in [0,1)$ such that
\begin{equation}
\norm{[\Pi_{\widetilde{J}} \widehat{\Phi}^* (\Phi - \widehat{\Phi}) \Pi_{\widetilde{J}}} \leq \vartheta \delta,
\quad \forall \widetilde{J} \subset [n],~ |\widetilde{J}| \leq 3s.
\label{eq:RBOP_AtE}
\end{equation}
Then, there exist $L_\delta^{\mathrm{HTP}} \in \bbN$ and $K_\delta^{\mathrm{HTP}} > 0$ such that
\[
\norm{x_k - x}_2 \leq C_\delta^\mathrm{HTP} (\vartheta \delta \norm{x}_2 + \sqrt{1+\delta} \norm{z}_2)
\]
for all $k \geq \lceil L_\delta^{\mathrm{HTP}} + K_\delta^{\mathrm{HTP}} s \rceil$,
where the constants $L_\delta^{\mathrm{HTP}}$, $K_\delta^{\mathrm{HTP}}$, $C_\delta^{\mathrm{HTP}}$ depend only on $\delta$.
\label{lemma:htp}
\end{lemma}

\begin{remark}
Given $\delta$, the constants $L_\delta^{\mathrm{HTP}}$, $K_\delta^{\mathrm{HTP}}$, and $C_\delta^\mathrm{HTP}$ are computed explicitly.
For example, if $\delta = 0.08$, then $L_\delta^{\mathrm{HTP}} = 3$, $K_\delta^{\mathrm{HTP}} = 3.17$, and $C_\delta^\mathrm{HTP} = 2.86$.
\end{remark}

Unlike the original performance guarantee for HTP \cite{Fou2011htp}, in Lemma~\ref{lemma:htp},
the recovery of the unknown $x$ via HTP uses a perturbed sensing matrix $\widehat{\Phi}$.
Furthermore, the error bound in Lemma~\ref{lemma:htp} applies to the estimate after $O(s)$ iterations rather than to the limit to which HTP converges to. This non-asymptotic bound is useful in the sense that the resulting final guarantee for SPF in Theorem~\ref{thm:conv_spf} applies to the case where the number of inner iterations for the HTP steps is bounded. As a result, compared to \cite{Fou2011htp}, the bound increases slightly by a constant factor 1.01, which appears in the definition of $C_\delta^{\mathrm{HTP}}$ in the proof of Lemma~\ref{lemma:htp}.

\begin{proof}[Proof of Lemma~\ref{lemma:htp}]
Let $E \triangleq \Phi - \widehat{\Phi}$. Then
\[
b = \widehat{\Phi} x + E x + z,
\]
which can be viewed as the measurement vector $\widehat{\Phi} x$ corrupted by the additive noise $E x + z$. Then, we can apply the conventional analysis of HTP \cite{Fou2011htp} to $\widehat{\Phi}$. However, due to the near bi-orthogonality between $\widehat{\Phi}_{\widetilde{J}}$ and $E_J$ in (\ref{eq:RBOP_AtE}), the noise term $E x$ and $z$ propagate with different amplification coefficients. The resulting guarantee is summarized as follows: Under the assumptions in Lemma~\ref{lemma:htp}, we have
\begin{equation}
\norm{x_k - x}_2 \leq \rho \norm{x_{k-1} - x}_2 + \tau (\vartheta \delta \norm{x}_2 + \sqrt{1+\delta} \norm{z}_2), \quad \forall k \in \bbN
\label{eq:idgp:recursion}
\end{equation}
and
\begin{equation}
\norm{x_k - x}_2 \leq \rho' \norm{\Pi_{J_k}^\perp x}_2 + \tau' (\vartheta \delta \norm{x}_2 + \sqrt{1+\delta} \norm{z}_2), \quad \forall k \in \bbN,
\label{eq:idgp:uncapenergy}
\end{equation}
where $J_k$ denotes the support of $x_k$ and the constants are given by
\[
\rho = \sqrt{\frac{2\delta^2}{1-\delta^2}}, \quad
\tau = \sqrt{\frac{2}{1-\delta^2}} + \frac{1}{1-\delta}, \quad
\rho' = \frac{1}{\sqrt{1-\delta^2}}, \quad
\tau' = \frac{1}{1-\delta}.
\]
(The modification of the conventional analysis of HTP to this version is straightforward. Hence, we omit the detail for deriving this step.)

Applying (\ref{eq:idgp:recursion}) and (\ref{eq:idgp:uncapenergy}) to \cite[Lemma~2.11]{LeeBJ2013oblique} yields that if
\begin{align*}
k > L + \frac{ s \ln\left( 1 + 2 \left( \rho' + \frac{\tau'}{\tau} \frac{1-\rho}{2} \right) \right) }
{ \ln \left(\frac{2}{1+\rho}\right) }
\end{align*}
for $L \in \bbN$, then
\[
\norm{x_k - x}_2 \leq \underbrace{ \frac{\tau [1 + \rho^L (2 \rho' - \rho)]}{1-\rho} }_{ (\star)} ~\cdot~ (\vartheta \delta \norm{x}_2 + \sqrt{1+\delta} \norm{z}_2).
\]

Set the parameter $L$ to
\[
L = \left\lceil \frac{\ln(100(2\rho'-\rho))}{\ln(\rho^{-1})} \right\rceil.
\]
Then, $(\star)$ is bounded from above by $\frac{1.01 \tau}{1-\rho}$; hence, if
\[
k >
\underbrace{ \left\lceil \frac{\ln(100(2\rho'-\rho))}{\ln(\rho^{-1})} \right\rceil }_{ L_\delta^{\mathrm{HTP}}}
+ s \cdot \underbrace{ \frac{ \ln\left[ 1 + 2 \left( \rho' + \frac{\tau'}{\tau} \frac{1-\rho}{2} \right) \right] }{ \ln \left(\frac{2}{1+\rho}\right) } }_{ K_\delta^{\mathrm{HTP}} },
\]
then
\[
\norm{x_k - x}_2 \leq
\underbrace{ \frac{1.01 \tau}{1-\rho} }_{ C_\delta^{\mathrm{HTP}} } ~\cdot~ (\vartheta \delta \norm{x}_2 + \sqrt{1+\delta} \norm{z}_2).
\]
\end{proof}

Next, using Lemma~\ref{lemma:htp}, we analyze the update of $u_t$ using HTP.

\begin{lemma}[Sparse update of $u_t$]
\label{lemma:update_theta}
Suppose the hypotheses of Theorem~\ref{thm:conv_spf} hold.
Let $u_t$ be obtained as the $\lceil L_\delta^{\mathrm{HTP}} + K_\delta^{\mathrm{HTP}} s \rceil$th iterate of HTP applied to the sensing matrix $F(v_{t-1})$ and the measurement vector $b$,
where $L_\delta^{\mathrm{HTP}}$ and $K_\delta^{\mathrm{HTP}}$ are constants determined by $\delta$ in the proof of Lemma~\ref{lemma:htp}.
Then,
\begin{align}
\sin \theta_t
{} & \leq
C_\delta^\mathrm{HTP} \left[\delta \tan \phi_{t-1} + (1+\delta) \sec \phi_{t-1} \frac{\norm{z}_2}{\norm{\A(X)}_2}\right].
\label{eq:update_theta}
\end{align}
where $\theta_t$ and $\phi_{t-1}$ are defined in (\ref{eq:defangles}).
\end{lemma}

\begin{proof}
Recall that $v_{t-1}$ was normalized in the $\ell_2$ norm before the update of $u_t$.
Decomposing $v$ as
\begin{equation}
v = P_{\R(v_{t-1})} v + \underbrace{P_{\R(v_{t-1})^\perp} v}_{ \zeta}
= (v_{t-1}^* v) v_{t-1} + \zeta,
\label{eq:decomp_checkv}
\end{equation}
we have
\[
\norm{\zeta}_2^2 = 1 - |v_{t-1}^* v|^2 = \sin^2 \phi_{t-1}
\]
and
\[
\norm{\zeta}_0 \leq \norm{v}_0 + \norm{v_{t-1}}_0 \leq 2 s_2.
\]

By (\ref{eq:decomp_checkv}), we have
\begin{equation}
(v_{t-1}^* v) F(v_{t-1}) = F(v) - F(\zeta).
\label{eq:decomp_Fvt}
\end{equation}

By (\ref{eq:decomp_Fvt}), we can rewrite the measurement vector $b$ as
\begin{align*}
b {} & = \A(X) + z \\
{} & = \A(\lambda u v^*) + z \\
{} & = [F(v)] (\lambda u) + z \\
{} & = [(v_{t-1}^* v) F(v_{t-1}) + F(\zeta)](\lambda u) + z \\
{} & = [\underbrace{ F(v_{t-1}) }_{ \widehat{\Phi} } + \underbrace{ (v_{t-1}^* v)^{-1} F(\zeta) }_{ E } ] \underbrace{ ((v_{t-1}^* v) \lambda u) }_{ x } + z.
\end{align*}

Since $\norm{v_{t-1}}_2 = 1$ and $\norm{v_{t-1}}_0 \leq s_2$, the rank-2 and $(3s_1,3s_2)$-sparse RIP of $\A$ implies by Lemma~\ref{lemma:ds_RIP2FG} that
$F(v_{t-1})$ satisfies the $3s_1$-sparse RIP with isometry constant $\delta$.
Similarly, since $\langle v_{t-1}, \zeta \rangle = 0$ and $\norm{v_{t-1}}_0 + \norm{\zeta}_0 \leq 3s_2$, the rank-2 and $(3s_1,3s_2)$-sparse RIP of $\A$ implies by Lemma~\ref{lemma:ds_RIP2FG} that
\[
\norm{\Pi_{\widetilde{J}} [F(v_{t-1})]^* [F(\zeta)] \Pi_{\widetilde{J}}} \leq \delta \norm{\zeta}_2,
\quad \forall \widetilde{J} \subset [n_1],~ |\widetilde{J}| \leq 3s_1.
\]

Therefore, by Lemma~\ref{lemma:htp}, we have
\begin{align*}
\norm{u_t - (v_{t-1}^* v) \lambda u}_2
{} & \leq C_\delta^\mathrm{HTP} \delta \norm{\zeta}_2 \norm{\lambda u}_2
+ C_\delta^\mathrm{HTP} \sqrt{1+\delta} ~ \norm{z}_2 \\
{} & \leq \lambda C_\delta^\mathrm{HTP} \left( \delta \sin \phi_{t-1} + \sqrt{1+\delta} ~ \frac{\norm{z}_2}{\lambda} \right) \\
{} & = \lambda \underbrace{ C_\delta^\mathrm{HTP} \left( \delta \sin \phi_{t-1} + \sqrt{1+\delta} ~ \frac{\norm{z}_2}{\fnorm{X}} \right) }_{ \alpha }
\end{align*}
and $\sin \theta_t$ is rewritten as
\begin{equation}
\begin{aligned}
\sin \theta_t
{} & = \norm{P_{\R(u_t)^\perp} P_{\R(u)}} \\
{} & = \norm{P_{\R(u_t)^\perp} u}_2 \\
{} & = \frac{\norm{P_{\R(u_t)^\perp} (v_{t-1}^* v) \lambda u}_2}{|(v_{t-1}^* v) \lambda|} \\
{} & = \frac{\norm{P_{\R(u_t)^\perp} (v_{t-1}^* v) \lambda u}_2}{\lambda \cos \phi_{t-1}} \\
{} & \overset{\text{(c)}}{\leq} \frac{\norm{(v_{t-1}^* v) \lambda u - u_t}_2}{\lambda \cos \phi_{t-1}} \\
{} & \leq \frac{\alpha}{\cos \phi_{t-1}},
\end{aligned}
\label{eq:ubtheta}
\end{equation}
where (c) follows since
\[
\norm{P_{\R(u_t)^\perp} (v_{t-1}^* v) \lambda u}_2 = \min_{\tilde{u}} \{ \norm{(v_{t-1}^* v) \lambda u - \tilde{u}}_2 :~ \tilde{u} \in \R(u_t) \}.
\]

It remains to apply the following upper bound on $\norm{z}_2/\fnorm{X}$
\begin{equation}
\frac{\norm{z}_2}{\norm{\A(X)}_2} \geq \frac{1}{\sqrt{1+\delta}} \cdot \frac{\norm{z}_2}{\fnorm{X}}
\label{eq:nuineq}
\end{equation}
implied by the rank-2 and row-$3s$-sparse RIP of $\A$ with isometry constant $\delta$.
\end{proof}

In the doubly sparse case ($s_1 < n_1$ and $s_2 < n_2$), by the symmetry in the updates of $u_t$ and $v_t$,
we get the following corollary as a direct implication of Lemma~\ref{lemma:update_theta}.

\begin{corollary}[Sparse update of $v_t$]
\label{cor:update_phi}
Suppose the hypotheses of Theorem~\ref{thm:conv_spf} hold.
Let $v_t$ be obtained as the $\lceil L_\delta^{\mathrm{HTP}} + K_\delta^{\mathrm{HTP}} s \rceil$th iterate of HTP applied to the sensing matrix $G(u_t)$ and the measurement vector $\overline{b}$,
where $L_\delta^{\mathrm{HTP}}$ and $K_\delta^{\mathrm{HTP}}$ are constants determined by $\delta$ in the proof of Lemma~\ref{lemma:htp}.
Then,
\begin{align}
\sin \phi_t
{} & \leq
C_\delta^\mathrm{HTP} \left[\delta \tan \theta_t + (1+\delta) \sec \theta_t \frac{\norm{z}_2}{\norm{\A(X)}_2}\right].
\label{eq:update_phi_htp}
\end{align}
where $\theta_t$ and $\phi_t$ are defined in (\ref{eq:defangles}).
\end{corollary}

In the row-sparse case ($s_2 = n_2$), the right factor $v_t$ is updated using $u_t$ by solving the least squares problem in (\ref{eq:lsupdate2_v}).
The angle between the resulting $v_t$ and $v$ is analyzed in the following lemma.

\begin{lemma}[LS update of $v_t$]
\label{lemma:update_phi}
Suppose the hypotheses of Theorem~\ref{thm:conv_spf} hold.
Then,
\begin{align}
\sin \phi_t
{} & \leq \frac{1}{1-\delta} \left[ \delta \tan \theta_t + (1+\delta) \sec \theta_t \frac{\norm{z}_2}{\norm{\A(X)}_2}\right].
\label{eq:update_phi}
\end{align}
\end{lemma}

\begin{proof}
Recall that $u_t$ is normalized in the $\ell_2$ norm before the update of $v_t$ and that $v_t$ is updated as
\[
v_t = [G(u_t)]^\dagger \bar{b}
\]
where $[G(u_t)]^\dagger$ denotes the pseudo-inverse of $[G(u_t)]$.

Decomposing $u$ as
\begin{align*}
u
= P_{\R(u_t)} u + \underbrace{P_{\R(u_t)^\perp} u}_{ \xi }
= (u_t^* u) u_t + \xi
\end{align*}
yields
\[
\norm{\xi}_2^2 = 1 - |u_t^* u|^2 = \sin^2 \theta_t.
\]

By the anti-linearity of $G$, we have
\[
\overline{(u_t^* u)} G(u_t) = G(u) - G(\xi).
\]

We can rewrite the complex conjugate of the measurement vector as follows:
\begin{align*}
\overline{b} {} & = \overline{\A(X)} + \overline{z} \\
{} & = \overline{\A(\lambda u v^*)} + \overline{z} \\
{} & = [G(u)](\lambda v) + \overline{z} \\
{} & = [\overline{(u_t^* u)} G(u_t) + G(\xi)](\lambda v) + \overline{z} \\
{} & = [G(u_t)](\overline{(u_t^* u)} \lambda v)
+ [G(\xi)](\lambda v) + \overline{z}.
\end{align*}

Since $\norm{u_t}_2 = 1$ and $\norm{u_t}_0 \leq s_1$, the rank-2 and $(3s_1,3s_2)$-sparse RIP of $\A$ implies by Lemma~\ref{lemma:ds_RIP2FG} that
\begin{equation}
\norm{([G(u_t)]^* [G(u_t)])^{-1}} \leq \frac{1}{1-\delta}
\label{eq:lemma:update_phi:GtG}
\end{equation}
and
\begin{equation}
\norm{[G(u_t)]^* z}_2 \leq \sqrt{1+\delta} \norm{z}_2, \quad \forall z \in \bbC^n.
\label{eq:lemma:update_phi:Gtz}
\end{equation}
Similarly, since $\langle u_t, \xi \rangle = 0$ and $\norm{u_t}_0 + \norm{\xi}_0 \leq 3s_1$,
the rank-2 and $(3s_1,3s_2)$-sparse RIP of $\A$ implies by Lemma~\ref{lemma:ds_RIP2FG} that
\begin{equation}
\norm{G(u_t)^* G(\xi)} \leq \delta \norm{\xi}_2.
\label{eq:lemma:update_phi:GtGxi}
\end{equation}

Note that $v_t$ is written as
\begin{align*}
v_t {} & = [G(u_t)]^\dagger \big([G(u_t)](\overline{(u_t^* u)} \lambda v)
+ [G(\xi)](\lambda v) + \overline{z}\big) \\
{} & = \overline{(u_t^* u)} \lambda v
+ [G(u_t)]^\dagger [G(\xi)] (\lambda v)
+ [G(u_t)]^\dagger \overline{z}.
\end{align*}

By (\ref{eq:lemma:update_phi:GtG}), (\ref{eq:lemma:update_phi:Gtz}), and (\ref{eq:lemma:update_phi:GtGxi}), it follows that
\begin{align*}
\norm{ v_t - \overline{(u_t^* u)} \lambda v }_2
{} & \leq \norm{ [G(u_t)]^\dagger [G(\xi)] (\lambda v) }_2
+ \norm{ [G(u_t)]^\dagger \overline{z} }_2 \\
{} & \leq \frac{\delta}{1-\delta} \lambda \norm{\xi}_2
+ \frac{\sqrt{1+\delta}}{1-\delta} \norm{z}_2 \\
{} & = \lambda \left( \frac{\delta}{1-\delta} \sin \theta_t + \frac{\sqrt{1+\delta}}{1-\delta} \cdot \frac{\norm{z}_2}{\lambda} \right) \\
{} & = \lambda \cdot \frac{1}{1-\delta} \left( \delta \sin \theta_t + \sqrt{1+\delta} ~ \frac{\norm{z}_2}{\fnorm{X}} \right) \\
{} & \leq \lambda \cdot \underbrace{ \frac{1}{1-\delta} \left( \delta \sin \theta_t + (1+\delta) \frac{\norm{z}_2}{\norm{\A(X)}_2} \right) }_{ \alpha }
\end{align*}
where the last step follow by (\ref{eq:nuineq}).

Then, $\sin \phi_t$ is bounded from above by
\begin{align*}
\sin \phi_t
{} & = \norm{P_{\R(v_t)^\perp} P_{\R(v)}} \\
{} & = \norm{P_{\R(v_t)^\perp} v}_2 \\
{} & = \frac{\norm{P_{\R(v_t)^\perp} \overline{(u_t^* u)} \lambda v}_2}{|(u_t^* u)| \lambda } \\
{} & = \frac{\norm{P_{\R(v_t)^\perp} \overline{(u_t^* u)} \lambda v}_2}{\lambda \cos \theta_t} \\
{} & \leq \frac{\norm{\overline{(u_t^* u)} \lambda v - v_t}_2}{\lambda \cos \theta_t} \\
{} & \leq \frac{\alpha}{\cos \theta_t}
\end{align*}
where the first inequality holds using a similar argument to (a) in the proof of Lemma~\ref{lemma:update_theta}.
\end{proof}

\begin{remark}
In fact, since $C_\delta^\mathrm{HTP} \geq (1-\delta)^{-1}$,
the upper bound on $\phi_t$ by Lemma~\ref{lemma:update_phi} is tighter than that by Corollary~\ref{cor:update_phi}.
In other words, Corollary~\ref{cor:update_phi} also applies to the row-sparse case.
By symmetry, we also conclude that Lemma~\ref{lemma:update_theta} also applies to the column-sparse case ($s_1 = n_1$).
\end{remark}

Next, Lemma~\ref{lemma:update_theta} and Corollary~\ref{cor:update_phi} provide recursive relations
that alternate between the two sequences $(\phi_t)_{t \in \bbN}$ and $(\theta_t)_{t \in \bbN}$.
From these results, we deduce a recursive relation on $(\theta_t)_{t \in \bbN}$ that leads to a convergence of SPF
in the following lemma, the proof of which is deferred to later.

\begin{lemma}
\label{lemma:convtheta}
Suppose the hypotheses of Theorem~\ref{thm:conv_spf} hold.
Define
\[
\Omega \triangleq \{\omega \in [0, \pi/2) :~ \omega \geq \sin^{-1} (C_\delta^{\mathrm{HTP}} [ \delta \tan \omega + (1+\delta) \nu \sec \omega ]) \}
\]
and let $\omega_{\sup}$ is the supremum of $\Omega$.
Suppose \begin{equation}
\norm{P_{\R(v)^\perp} P_{\R(v_0)}} < \sin \omega_{\sup}.
\label{eq:conv_spf_initcond_general}
\end{equation}
In fact, for $\delta = 0.08$ and $\nu = 0.08$ as in Theorem~\ref{thm:conv_spf},
then \eqref{eq:conv_spf_initcond_general} coincides with \eqref{eq:conv_spf_initcond}.
Then,
\[
\limsup_{t \to \infty} \sin \theta_t \leq C \frac{\norm{z}_2}{\norm{\A(X)}_2},
\]
where the constant $C$ depends only on $\delta$ and $\nu$.
Furthermore, $\max(0, \sin \theta_t - C \norm{z}_2 / \norm{\A(X)}_2)$ converges to 0 superlinearly.
\end{lemma}

The next lemma provides an upper bound on the normalized estimation error $\fnorm{X - u_t v_t^*}/\fnorm{X}$ in terms of $\sin \theta_t$.
\begin{lemma}
\label{lemma:convspfphi2X}
Suppose the hypotheses of Theorem~\ref{thm:conv_spf} hold.
Then,
\begin{equation}
\label{eq:lemma:convspfphi2X}
\frac{\fnorm{X - u_t v_t^*}}{\fnorm{X}}
\leq \sqrt{1 + 2 (\delta C_\delta^{\mathrm{HTP}})^2} ~ \sin \theta_t + \sqrt{2} (1+\delta) C_\delta^{\mathrm{HTP}} \frac{\norm{z}_2}{\norm{\A(X)}_2}.
\end{equation}
\end{lemma}

\begin{remark}
For fixed $\delta$ and $\nu$, the noise amplification in Lemma~\ref{lemma:convtheta} is explicitly bounded.
For example, if $\delta \leq 0.08$ and $\nu \leq 0.08$, then
\[
\limsup_{t \to \infty} \sin \theta_t \leq 3.75 ~ \frac{\norm{z}_2}{\norm{\A(X)}_2}
\]
By Lemma~\ref{lemma:convspfphi2X}, this implies
\[
\limsup_{t \to \infty} \frac{\norm{X - X_t}_{\mathrm{F}}}{\norm{X}_{\mathrm{F}}} \leq 8.3 ~ \frac{\norm{z}_2}{\norm{\A(X)}_2}.
\]
\end{remark}

Combining Lemmas~\ref{lemma:convtheta} and \ref{lemma:convspfphi2X} yields the proof of Theorem~\ref{thm:conv_spf}. We conclude this section with the proofs of Lemmas~\ref{lemma:convtheta} and \ref{lemma:convspfphi2X}.

\begin{proof}[Proof of Lemma~\ref{lemma:convtheta}]
Define $f: [0,\pi/2) \to [0,\pi/2)$ by
\[
f(\omega) \triangleq \sin^{-1} \left( C_\delta^{\mathrm{HTP}} \left[ \delta \tan \omega + (1+\delta) \sec \omega \frac{\norm{z}_2}{\norm{\A(X)}_2} \right] \right)
\]
and let $\Omega_f \triangleq \{\omega :~ \omega \geq f(\omega)\}$.
Let $\omega_{f,\inf}$ and $\omega_{f,\sup}$ denote the infimum and supremum of $\Omega_f$, respectively.
Since $\sin \omega$ is monotone increasing and concave on $[0,\pi/2]$,
it follows that $\sin^{-1} \upsilon$ is monotone increasing and convex on $[0,1]$.
Furthermore, both $\tan \omega$ and $\sec \omega$ are monotone increasing and convex on $[0,\pi/2)$.
Therefore, $f$ is monotone increasing and convex on $[0,\pi/2)$.
Then it follows that
\begin{equation}
\label{eq:lemma:convtheta:set}
\begin{aligned}
\omega {} & \geq f(\omega), \quad \forall \omega \in [\omega_{f,\inf},\omega_{f,\sup}], \\
\omega {} & < f(\omega), \quad \forall \omega \in [0,\omega_{f,\inf}) \cup (\omega_{f,\sup}, \pi/2).
\end{aligned}
\end{equation}

First, we show that
\begin{equation}
\label{eq:lemma:convtheta:phi0}
\phi_0 < \omega_{f,\sup}.
\end{equation}
Since $\norm{z}_2/\norm{\A(X)}_2 \leq \nu$, $\Omega$ is a subset of $\Omega_f$.
Then it follow that $\omega_{\sup} \leq \omega_{f,\sup}$.
The inequality in \eqref{eq:lemma:convtheta:phi0} holds since we assumed that $\phi_0 < \omega_{\sup}$.

In the first iteration, by Lemma~\ref{lemma:update_theta}, $\theta_1$ is upper-bounded by
\begin{equation}
\label{eq:lemma:convtheta:update_theta1}
\theta_1 \leq f(\phi_0).
\end{equation}
Then, by \eqref{eq:lemma:convtheta:set} and \eqref{eq:lemma:convtheta:update_theta1}, either of the following cases holds:
if $\phi_0 > \omega_{f,\inf}$, then $\theta_1 < \phi_0$;
if $\phi_0 \leq \omega_{f,\inf}$, then $\theta_1 \leq \omega_{f,\inf}$.

Next, in the second iteration, by Corollary~\ref{cor:update_phi}, $\phi_1$ is upper-bounded by
\begin{equation}
\label{eq:lemma:convtheta:update_phi1}
\phi_1 \leq f(\theta_1).
\end{equation}
Similarly to the previous case, since $\theta_1 < \omega_{\sup}$, \eqref{eq:lemma:convtheta:update_phi1} implies that $\phi_1$ satisfies
either $\phi_1 < \theta_1$ or $\phi_1 \leq \omega_{f,\inf}$.

By induction, both $(\phi_t)_{t \in \bbN}$ and $(\theta_t)_{t \in \bbN}$ converge to the set $[0, \omega_{\inf}]$.
The convergence is superlinear because of the convexity of $f$.

Finally, it remains to compute an upper bound on $\sin \omega_{f,\inf}$.
Define $\hat{f}: [0,\pi/2) \to [0,\pi/2)$ by
\[
\hat{f}(\omega) \triangleq \sin^{-1} \left(\frac{C_\delta^{\mathrm{HTP}}}{\cos \omega_{\sup}} \left[ \delta \sin \omega + (1+\delta) \frac{\norm{z}_2}{\norm{\A(X)}_2} \right] \right)
\]
and let $\widehat{\Omega}_f \triangleq \{\omega  \in [0, \omega_{\sup}] :~ \omega \geq \hat{f}(\omega)\}$.
Then $\widehat{\Omega}_f$ is a subset of $\Omega_f$
and it follows that $\omega_{f,\inf}$ is no greater than the infimum of $\widehat{\Omega}_f$.
Thus we have
\[
\sin \omega_{f,\inf} \leq \underbrace{ \left(\frac{1+\delta}{\cos \omega_{\sup} - \delta C_\delta^{\mathrm{HTP}}}\right) }_{ C } \frac{\norm{z}_2}{\norm{\A(X)}_2},
\]
where $C$ depends only on $\delta$ and $\nu$.
\end{proof}

\begin{proof}[Proof of Lemma~\ref{lemma:convspfphi2X}]
Without loss of generality, we may assume that $\norm{u_t}_2 = 1$.
Note that the estimation error is decomposed as
\[
\lambda u v^* - u_t v_t^*
= \lambda P_{R(u_t)} u v^* - u_t v_t^* + \lambda P_{R(u_t)^\perp} u v^*
= u_t [(u^* u_t) \lambda v - v_t]^* + \lambda P_{R(u_t)^\perp} u v^*.
\]

As shown in the proof of Lemma~\ref{lemma:update_theta}, we have
\[
\norm{v_t - (u^* u_t) \lambda v}_2
\leq \lambda C_\delta^{\mathrm{HTP}} \left[ \delta \norm{P_{R(u_t)^\perp} u}_2 + (1+\delta) \frac{\norm{z}_2}{\norm{\A(X)}_2} \right].
\]

Therefore,
\begin{align*}
\fnorm{\lambda u v^* - u_t v_t^*}^2
{} & = \fnorm{u_t [(u^* u_t) \lambda v - v_t]^*}^2 + \lambda^2 \fnorm{P_{R(u_t)^\perp} u v^*}^2 \\
{} & = \norm{(u^* u_t) \lambda v - v_t}_2^2 + \lambda^2 \norm{P_{R(u_t)^\perp} u}_2^2 \\
{} & \leq \lambda^2 \left[1 + 2 (\delta C_\delta^{\mathrm{HTP}})^2\right] \norm{P_{R(u_t)^\perp} P_{R(u)}}^2
+ 2 \left[\lambda (1+\delta) C_\delta^{\mathrm{HTP}} \frac{\norm{z}_2}{\norm{\A(X)}_2} \right]^2,
\end{align*}
which implies (\ref{eq:lemma:convspfphi2X}).
\end{proof}

\subsection{Proof of Theorem~\ref{thm:pgrip_init_ds}}
\label{subsec:initds}
The lemmas in this section assumes that $\A$ satisfies the rank-2 and $(3s_1,3s_2)$-sparse RIP with isometry constant $\delta$,
unless otherwise stated.
Through the following results, we show that the problem of finding a good initialization,
which satisfies (\ref{eq:conv_spf_initcond}), reduces to finding good support estimates of $u$ and $v$.
Once the satisfaction of (\ref{eq:conv_spf_initcond}) is shown, the performance guarantee is automatically implied by Theorem~\ref{thm:conv_spf}.

First, we present the following lemma that computes an upper bound on the angle between $v$ and $v_0$.

\begin{lemma}
\label{lemma:init_w_good_est_supp2}
Let $\widehat{J}_1 \subset [n_1]$ and $\widehat{J}_2 \subset [n_2]$ satisfy $|\widehat{J}_1| \leq s_1$ and $|\widehat{J}_1| \leq s_2$, respectively.
Let $v_0$ be obtained as the leading right singular vector of $\Pi_{\widehat{J}_1} [\A^*(b)] \Pi_{\widehat{J}_2}$.
Then,
\begin{equation}
\norm{P_{\R(v)^\perp} P_{\R(v_0)}} \leq
\frac{\norm{\Pi_{\widehat{J}_1} u}_2 \norm{\Pi_{\widehat{J}_2}^\perp v}_2 + \delta + (1+\delta) \nu}{\norm{\Pi_{\widehat{J}_1} u}_2 - \delta - (1+\delta) \nu}.
\label{eq:lemma:init_w_good_est_supp2:ub}
\end{equation}
\end{lemma}
\begin{proof}
See Appendix~\ref{sec:proof:lemma:init_w_good_est_supp2}.
\end{proof}

By Lemma~\ref{lemma:init_w_good_est_supp2}, if the right-hand-side of (\ref{eq:lemma:init_w_good_est_supp2:ub})
is smaller than $\sin \omega_{\sup}$, then the inequality in (\ref{eq:conv_spf_initcond}) holds.
Rearranging this, we get a sufficient condition for (\ref{eq:conv_spf_initcond}) given by
\begin{equation}
\delta + (1+\delta) \nu < \frac{\norm{\Pi_{\widehat{J}_1} u}_2 \left(\sin \omega_{\sup} - \norm{\Pi_{\widehat{J}_2}^\perp v}_2\right)}{1 + \sin \omega_{\sup}}.
\label{eq:pgrip_spf_init_good_supp2:cond}
\end{equation}

Finding estimates $\widehat{J}_1$ and $\widehat{J}_2$ of the supports of $u$ and $v$ that satisfy (\ref{eq:pgrip_spf_init_good_supp2:cond})
is much easier than finding the supports of $u$ and $v$ exactly.
We show that the two support estimation schemes in the initialization of SPF, proposed in Section~\ref{sec:alg},
are guaranteed to satisfy (\ref{eq:pgrip_spf_init_good_supp2:cond}) under the RIP assumption.

\subsubsection{Proof of Part 1}

Recall that $\vzstar$ was the leading right singular vector of $\Pi_{\widehat{J}_1} [\A^*(b)] \Pi_{\widehat{J}_2}$, where $\widehat{J}_1$ and $\widehat{J}_2$ are given in (\ref{eq:initbydspca}).
In this case, the product of $\norm{\Pi_{\widehat{J}_1} u}_2$ and $\norm{\Pi_{\widehat{J}_2} v}_2$ is lower-bounded by the following lemma.

\begin{lemma}
\label{lemma:est_supp_dspca}
Let $\widehat{J}_1 \subset [n_1]$ and $\widehat{J}_2 \subset [n_2]$ be given by (\ref{eq:initbydspca}).
Then,
\begin{equation}
\norm{\Pi_{\widehat{J}_1} u}_2 \cdot \norm{\Pi_{\widehat{J}_2} v}_2
\geq \sin \left[
\sin^{-1} \left( \frac{1 - \delta - 2 (1+\delta) \nu}{\sqrt{\delta^2 + (1+\delta)^2}} \right)
- \sin^{-1} \left( \frac{\delta}{\sqrt{\delta^2 + (1+\delta)^2}} \right) \right].
\label{eq:lemma:est_supp_dspca:res}
\end{equation}
\end{lemma}
\begin{proof}
See Appendix~\ref{sec:proof:lemma:est_supp_dspca}.
\end{proof}

Let
\[
\varphi(\delta,\nu) \triangleq \sin^{-1} \left( \frac{1 - \delta - 2 (1+\delta) \nu}{\sqrt{\delta^2 + (1+\delta)^2}} \right)
- \sin^{-1} \left( \frac{\delta}{\sqrt{\delta^2 + (1+\delta)^2}} \right).
\]
Since $\norm{\Pi_{\widehat{J}_1} u}_2 \leq 1$ and $\norm{\Pi_{\widehat{J}_2} v}_2 \leq 1$,
by Lemma~\ref{lemma:est_supp_dspca}, we have
\[
\min(\norm{\Pi_{\widehat{J}_1} u}_2, \norm{\Pi_{\widehat{J}_2} v}_2) \geq \sin \varphi(\delta,\nu),
\]
which also implies
\[
\norm{\Pi_{\widehat{J}_2}^\perp v}_2 \leq \cos \varphi(\delta,\nu).
\]
Therefore, a sufficient condition for $\vzstar$ to satisfy (\ref{eq:conv_spf_initcond}) is given by
\begin{equation}
\label{eq:thm:pgrip_dspf_dspca:cond}
\delta + (1+\delta) \nu < \frac{\sin \varphi(\delta,\nu) \left(\sin \omega_{\sup} - \cos \varphi(\delta,\nu) \right)}{1 + \sin \omega_{\sup}}.
\end{equation}

Owing to the monotonicity of components in (\ref{eq:thm:pgrip_dspf_dspca:cond}), a set of $(\delta,\nu)$ where (\ref{eq:thm:pgrip_dspf_dspca:cond}) holds is determined. For example, it holds if $\delta \leq 0.04$ and $\nu \leq 0.04$.

\subsubsection{Proof of Part 2-(a)}

Toward a performance guarantee using a computationally efficient algorithm, we analyze the performance of the initialization with $v_0^\text{\rm Th}$.

The assumptions on $\delta$, $\nu$, $\norm{u}_\infty$, and $\norm{v}_\infty$ imply that
\begin{equation}
\delta + \nu + \delta \nu \leq \frac{\norm{u}_\infty \norm{v}_\infty + \sin \omega_{\sup} - \sqrt{2}}{3 + \sin \omega_{\sup}},
\label{eq:thm:pgrip_init_ds:2a:cond_gen}
\end{equation}
where $\sin \omega_{\sup} \geq 0.97$.

Let $j_0 \in [n_1]$ denote the index of the largest entry of $u$ in magnitude, i.e.,
\[
j_0 \triangleq \argmax_{j \in [n_1]} |[u]_j|.
\]
Then, $|[u]_{j_0}| = \norm{\Pi_{\{j_0\}} u}_2 = \norm{u}_\infty$.

First, we show that $j_0 \in \widehat{J}_1$ using the following lemma.
\begin{lemma}
\label{lemma:est_supp_dthres}
Let $\widehat{J}_1 \subset [n_1]$ be the support estimate for $v_0^\text{\rm Th}$.
If there exists $\widetilde{J}_1 \subset \supp{u}$ such that
\begin{equation}
2\delta + 2(1+\delta)\nu < \min_{j \in \widetilde{J}_1} |[u]_j|,
\label{eq:lemma:est_supp_dthres:cond1}
\end{equation}
then $\widetilde{J}_1 \subset \widehat{J}_1$.
\end{lemma}
\begin{proof}
See Appendix~\ref{sec:proof:lemma:est_supp_dthres}.
\end{proof}

We apply Lemma~\ref{lemma:est_supp_dthres} with $\widetilde{J}_1 = \{j_0\}$. Since (\ref{eq:thm:pgrip_init_ds:2a:cond_gen}) implies (\ref{eq:lemma:est_supp_dthres:cond1}),
we have shown $j_0 \in \widehat{J}_1$.

Next, we derive a lower bound on $\norm{\Pi_{\widehat{J}_1} u}_2 \norm{\Pi_{\widehat{J}_2} v}_2$. Define
\begin{align*}
k_0 {} & \triangleq \argmax_{k \in [n_2]} |[v]_k|, \\
k_1 {} & \triangleq \argmax_{k \in [n_2]} \fnorm{\Pi_{\widehat{J}_1} [\A^*(b)] \Pi_{\{k\}}}, \\
k_2 {} & \triangleq \argmax_{k \in [n_2]} \fnorm{\Pi_{\{j_0\}} [\A^*(b)] \Pi_{\{k\}}}.
\end{align*}
Then, by the selection of $\widehat{J}_2$ in computing $v_0^\text{\rm Th}$, we have $k_1 \in \widehat{J}_2$.
On the other hand, by definition of $k_1$ and $k_2$, it follow that
\begin{equation}
\label{eq:proof:thm2a:pgrip_dspf_dspca:ineqJ2}
\begin{aligned}
\underbrace{ \fnorm{\Pi_{\widehat{J}_1} [\A^*(b)] \Pi_{\{k_1\}}} }_{ \text{(a)} }
{} & \geq \fnorm{\Pi_{\widehat{J}_1} [\A^*(b)] \Pi_{\{k_2\}}} \\
{} & \geq \fnorm{\Pi_{\{j_0\}} [\A^*(b)] \Pi_{\{k_2\}}} \\
{} & \geq \underbrace{ \fnorm{\Pi_{\{j_0\}} [\A^*(b)] \Pi_{\{k_0\}}} }_{ \text{(b)} }
\end{aligned}
\end{equation}
where the second inequality holds since $j_0 \in \widehat{J}_1$.

The left-hand-side of (\ref{eq:proof:thm2a:pgrip_dspf_dspca:ineqJ2}) is further upper-bounded by
\begin{equation}
\label{eq:proof:thm2a:pgrip_dspf_dspca:ineqJ2:lhs}
\begin{aligned}
\text{(a)}
{} & \leq \fnorm{\Pi_{\widehat{J}_1} X \Pi_{\{k_1\}}} + \fnorm{\Pi_{\widehat{J}_1} (\A^*\A - \idm)(X) \Pi_{\{k_1\}}} + \fnorm{\Pi_{\widehat{J}_1} [\A^*(z)] \Pi_{\{k_1\}}} \\
{} & \leq \lambda \norm{\Pi_{\widehat{J}_1} u}_2 \norm{\Pi_{\{k_1\}} v}_2 + \delta \lambda + \sqrt{1+\delta} \norm{z}_2 \\
{} & \leq \lambda \norm{\Pi_{\widehat{J}_1} u}_2 \norm{\Pi_{\widehat{J}_2} v}_2 + \delta \lambda + \sqrt{1+\delta} \norm{z}_2,
\end{aligned}
\end{equation}
where the second inequality follows by Lemmas~\ref{lemma:ds_bndopnormrip} and \ref{lemma:ds_bndnoise},
and the last step holds since $k_1 \in \widehat{J}_2$.

The right-hand-side of (\ref{eq:proof:thm2a:pgrip_dspf_dspca:ineqJ2}) is lower-bounded by
\begin{equation}
\label{eq:proof:thm2a:pgrip_dspf_dspca:ineqJ2:rhs}
\begin{aligned}
\text{(b)}
{} & \geq \fnorm{\Pi_{\{j_0\}} X \Pi_{\{k_0\}}} - \fnorm{\Pi_{\{j_0\}} (\A^*\A - \idm)(X) \Pi_{\{k_0\}}} - \fnorm{\Pi_{\{j_0\}} [\A^*(z)] \Pi_{\{k_0\}}} \\
{} & \geq \lambda \norm{u}_\infty \norm{v}_\infty - \delta \lambda - \sqrt{1+\delta} \norm{z}_2.
\end{aligned}
\end{equation}

Applying (\ref{eq:proof:thm2a:pgrip_dspf_dspca:ineqJ2:lhs}) and (\ref{eq:proof:thm2a:pgrip_dspf_dspca:ineqJ2:rhs}) to (\ref{eq:proof:thm2a:pgrip_dspf_dspca:ineqJ2}) with (\ref{eq:nuineq}) yields
\begin{equation}
\label{eq:proof:thm2a:pgrip_dspf_dspca:bnduvnorms}
\norm{\Pi_{\widehat{J}_1} u}_2 \norm{\Pi_{\widehat{J}_2} v}_2 \geq \norm{u}_\infty \norm{v}_\infty - 2 (\delta+ \nu + \delta\nu).
\end{equation}

Then, by applying (\ref{eq:proof:thm2a:pgrip_dspf_dspca:bnduvnorms}) to (\ref{eq:pgrip_spf_init_good_supp2:cond}) with some rearrangement,
we get a sufficient condition for (\ref{eq:pgrip_spf_init_good_supp2:cond}) given by
\begin{equation}
\label{eq:proof:thm2a:pgrip_dspf_dspca:suff1}
\begin{aligned}
\norm{\Pi_{\widehat{J}_1} u}_2^2
{} & < \left[ \sin \omega_{\sup} \norm{\Pi_{\widehat{J}_1} u}_2 - (1+\sin \omega_{\sup}) (\delta+\nu+\delta\nu) \right]^2 \\
{} & \quad + \left[\norm{u}_\infty \norm{v}_\infty - 2(\delta+ \nu + \delta\nu)\right]^2.
\end{aligned}
\end{equation}
By convexity of the scalar quadratic function, (\ref{eq:proof:thm2a:pgrip_dspf_dspca:suff1}) is implied by
\begin{equation}
\label{eq:proof:thm2a:pgrip_dspf_dspca:suff2}
\begin{aligned}
\sqrt{2} \norm{\Pi_{\widehat{J}_1} u}_2
< \sin \omega_{\sup} \norm{\Pi_{\widehat{J}_1} u}_2 + \norm{u}_\infty \norm{v}_\infty - (3+\sin \omega_{\sup}) (\delta+ \nu + \delta\nu).
\end{aligned}
\end{equation}
Here, we used the fact that $\sin \omega_{\sup} \norm{\Pi_{\widehat{J}_1} u}_2 - (1+\sin \omega_{\sup}) (\delta+\nu+\delta\nu) > 0$,
which is implied by (\ref{eq:thm:pgrip_init_ds:2a:cond_gen}).

Note that (\ref{eq:proof:thm2a:pgrip_dspf_dspca:suff2}) is equivalently rewritten as
\[
\delta+ \nu + \delta\nu < \frac{\norm{u}_\infty \norm{v}_\infty - (\sqrt{2} - \sin \omega_{\sup}) \norm{\Pi_{\widehat{J}_1} u}_2}{3+\sin \omega_{\sup}},
\]
which is implied by (\ref{eq:thm:pgrip_init_ds:2a:cond_gen}) since $\sin \omega_{\sup} < \sqrt{2}$.

\subsubsection{Proof of Part 2-(b)}

Unlike the previous parts of Theorem~\ref{thm:pgrip_init_ds},
Part 2-(b) assumes a stronger RIP of $\A$ that applies to all matrices with up to $9 s_1 s_2$ nonzero entries.

Let $\calS$ denote the set of $n_1$-by-$n_2$ matrices such that there are at most $s_1$ nonzero rows and each row has at most $s_2$ nonzero elements.
Let $\calP_\calS$ denote the orthogonal projection onto $\calS$.
Then, $\widehat{J}_1$ coincides with the row-support of $\calP_\calS[\A^*(b)]$.

The RIP assumption admits an upper bound on $\fnorm{\calP_\calS[\A^*(b)] - X}$,
which is direct implication of the analogous result in compressed sensing of a sparse vector.
We use the version by Foucart \cite[Theorem~3]{Fou2012sparse} given by
\begin{equation}
\fnorm{\calP_\calS[\A^*(b)] - X} \leq 2 \delta \fnorm{X} + 2 \sqrt{1+\delta} \norm{z}_2.
\label{eq:proxybnd}
\end{equation}

Define a two-dimensional coordinate projection $\widetilde{\Pi}_{\calJ}: \bbC^{n_1 \times n_2} \to \bbC^{n_1 \times n_2}$
associated with an index set $\calJ \subset [n_1] \times [n_2]$ by
\[
[\widetilde{\Pi}_{\calJ}(Z)]_{i,j}
= \begin{cases}
[Z]_{i,j} & (i,j) \in \calJ, \\
0 & \text{else},
\end{cases}
\]
for all $Z \in \bbC^{n_1 \times n_2}$.
Then, it follows that $\widetilde{\Pi}_{\widehat{J}_1 \times \widehat{J}_2}(Z) = \Pi_{\widehat{J}_1} Z \Pi_{\widehat{J}_1}$.
Furthermore, there exists a subset $\widehat{\calJ} \subset [n_1] \times [n_2]$ such that
\begin{equation}
\label{eq:2dprojid}
\calP_\calS[\A^*(b)] = \widetilde{\Pi}_{\widehat{\calJ}}[\A^*(b)].
\end{equation}

Let $\widehat{M} \in \bbC^{n_1 \times n_2}$ denote the best rank-1 approximation
of $\Pi_{\widehat{J}_1} [\A^*(b)] \Pi_{\widehat{J}_2}$ in the Frobenius norm.
Then, we have
\begin{align*}
\fnorm{\Pi_{\widehat{J}_1} [\A^*(b)] \Pi_{\widehat{J}_2} - \widehat{M}}
\leq \fnorm{\Pi_{\widehat{J}_1} [\A^*(b)] \Pi_{\widehat{J}_2} - X}.
\end{align*}
Therefore,
\begin{equation}
\label{eq:proxybnd2}
\begin{aligned}
\fnorm{\widehat{M} - X}
{} & \leq \fnorm{\widehat{M} - \Pi_{\widehat{J}_1} [\A^*(b)] \Pi_{\widehat{J}_2}} + \fnorm{\Pi_{\widehat{J}_1} [\A^*(b)] \Pi_{\widehat{J}_2} - X} \\
{} & \leq 2 \fnorm{\Pi_{\widehat{J}_1} [\A^*(b)] \Pi_{\widehat{J}_2} - X} \\
{} & \leq 2 \fnorm{\Pi_{\widehat{J}_1} [\A^*(b)] \Pi_{\widehat{J}_2} - \widetilde{\Pi}_{(\widehat{J}_1 \times \widehat{J}_2) \cup \widehat{\calJ}}[\A^*(b)]} + 2 \fnorm{\widetilde{\Pi}_{(\widehat{J}_1 \times \widehat{J}_2) \cup \widehat{\calJ}}[\A^*(b)] - X} \\
{} & \leq 4 \fnorm{\widetilde{\Pi}_{(\widehat{J}_1 \times \widehat{J}_2) \cup \widehat{\calJ}}[\A^*(b)] - X},
\end{aligned}
\end{equation}
where the last step follows since
\[
\Pi_{\widehat{J}_1} [\A^*(b)] \Pi_{\widehat{J}_2}
= \argmin_Z \{ \fnorm{\widetilde{\Pi}_{(\widehat{J}_1 \times \widehat{J}_2) \cup \widehat{\calJ}}[\A^*(b)] - Z} :~ \text{$Z$ is $(s_1,s_2)$-sparse} \}.
\]

The last term in (\ref{eq:proxybnd2}) is further upper-bounded by
\begin{equation}
\label{eq:proxybnd3}
\fnorm{\widetilde{\Pi}_{(\widehat{J}_1 \times \widehat{J}_2) \cup \widehat{\calJ}}[\A^*(b)] - X}
\leq \fnorm{\widetilde{\Pi}_{\widehat{\calJ}}[\A^*(b)] - X}
+ \fnorm{\widetilde{\Pi}_{(\widehat{J}_1 \times \widehat{J}_2) \setminus \widehat{\calJ}}[\A^*(b)]}.
\end{equation}

Since $\widetilde{\Pi}_{\widehat{\calJ}}[\A^*(b)] = \calP_\calS[\A^*(b)]$,
the first term in the right-hand-side of (\ref{eq:proxybnd3}) is upper-bounded by (\ref{eq:proxybnd}).

Next, by the triangle inequality, the second term in the right-hand-side of (\ref{eq:proxybnd3}) is upper-bounded by
\begin{align*}
{} & \fnorm{\widetilde{\Pi}_{(\widehat{J}_1 \times \widehat{J}_2) \setminus \widehat{\calJ}}[\A^*(b)]} \\
{} & = \fnorm{\widetilde{\Pi}_{(\widehat{J}_1 \times \widehat{J}_2) \setminus \widehat{\calJ}}[X + (\A^*\A-\id)(X)+\A^*(z)]} \\
{} & \leq \fnorm{\widetilde{\Pi}_{(\widehat{J}_1 \times \widehat{J}_2) \setminus \widehat{\calJ}}(X)}
+ \fnorm{\widetilde{\Pi}_{(\widehat{J}_1 \times \widehat{J}_2) \setminus \widehat{\calJ}}[(\A^*\A - \id)(X)]}
+ \fnorm{\widetilde{\Pi}_{(\widehat{J}_1 \times \widehat{J}_2) \setminus \widehat{\calJ}}[\A^*(z)]} \\
{} & \leq \underbrace{ \fnorm{\widetilde{\Pi}_{(\widehat{J}_1 \times \widehat{J}_2) \setminus \widehat{\calJ}}(X)} }_{=\text{(a)}}
+ \delta \fnorm{X} + \sqrt{1+\delta} \norm{z}_2,
\end{align*}
where the last step follows by the RIP of $\A$ with $|(J_1 \times J_2) \setminus \widehat{\calJ}| \leq s_1 s_2$.

Let $J_1$ and $J_2$ denote the support of $u$ and $v$, respectively.
In other words, $J_1$ and $J_2$ are the row-support and column-support of $X$, respectively.
Then,
\[
\widetilde{\Pi}_{(\widehat{J}_1 \times \widehat{J}_2) \setminus \widehat{\calJ}}(X)
= \widetilde{\Pi}_{[(\widehat{J}_1 \cap J_1) \times (\widehat{J}_2 \cap J_2)] \setminus \widehat{\calJ}}(X).
\]
Therefore, $\text{(a)}$ is upper-bounded by
\begin{align*}
{} & \fnorm{\widetilde{\Pi}_{(J_1 \times J_2) \setminus \widehat{\calJ}}(X)} \\
{} & = \fnorm{\widetilde{\Pi}_{\widehat{\calJ}}(X) - X} \\
{} & = \fnorm{\widetilde{\Pi}_{\widehat{\calJ}}[\A^*(b) - (\A^*\A-\id)(X) - \A^*(z)] - X} \\
{} & \leq \fnorm{\widetilde{\Pi}_{\widehat{\calJ}}[\A^*(b)] - X} + \fnorm{\widetilde{\Pi}_{\widehat{\calJ}}[(\A^*\A-\id)(X)]}
+ \fnorm{\widetilde{\Pi}_{\widehat{\calJ}}[\A^*(z)]} \\
{} & \leq 3 \delta \fnorm{X} + 3 \sqrt{1+\delta} \norm{z}_2,
\end{align*}
where the last step follows by (\ref{eq:2dprojid}), (\ref{eq:proxybnd}), and the RIP of $\A$ with $|\widehat{\calJ}| \leq s_1 s_2$.

By applying these upper bounds back to (\ref{eq:proxybnd2}), we get
\begin{equation}
\label{eq:proxybnd4}
\fnorm{\widehat{M} - X} \leq 24 \delta \fnorm{X} + 24 \sqrt{1+\delta} \norm{z}_2.
\end{equation}

Let $v_0^\text{\rm Th}$ be the leading right singular vector of $M$.
Then, by the non-Hermitian $\sin \theta$ theorem \cite{Wed1972perturbation} (Lemma~\ref{lemma:sintheta}),
\[
\norm{P_{\R(v)^\perp} P_{\R(v_0^\text{\rm Th})}} \leq \frac{\fnorm{M - X}}{\fnorm{X}} \leq 24\delta + 24(1+\delta)\nu,
\]
where the second step follows from (\ref{eq:proxybnd4}) and (\ref{eq:nuineq}).

Therefore, a sufficient for (\ref{eq:conv_spf_initcond}) is given by
\[
24\delta + 24(1+\delta)\nu < \sin \omega_{\sup},
\]
which holds, for example, if $\delta \leq 0.02$ and $\nu \leq 0.02$.

\subsection{Proof of Theorem~\ref{thm:pgrip_init_rs}}
\label{subsec:initrs}

In the row-sparse case ($s_2 = n_2$), the support estimate for $v$ is trivially given as $\widehat{J}_2 = [n_2]$.
The assumptions of Theorem~\ref{thm:pgrip_init_rs} imply
\begin{equation}
\delta+\nu+\delta\nu < \min\left( \frac{\sin \omega_{\sup} \norm{u}_\infty}{1 + \sin \omega_{\sup}}, \frac{\norm{u}_\infty}{2} \right)
\label{eq:proof:thm:pgrip_init_rs:cond}
\end{equation}

Let $j_0 \in [n_1]$ denote the index of the largest entry of $u$ in magnitude, i.e.,
\[
j_0 \triangleq \argmax_{j \in [n_1]} |[u]_j|.
\]
Then, $|[u]_{j_0}| = \norm{\Pi_{\{j_0\}} u}_2 = \norm{u}_\infty$.

Since (\ref{eq:proof:thm:pgrip_init_rs:cond}) implies
\begin{equation}
2\delta + 2(1+\delta)\nu < \norm{u}_\infty,
\end{equation}
by Lemma~\ref{lemma:est_supp_dthres} with $\widetilde{J}_1 = \{j_0\}$,
it follows that $j_0 \in \widehat{J}_1$; hence,
\begin{equation}
\norm{\Pi_{\widehat{J}_1} u}_2 \geq \norm{u}_\infty.
\label{eq:proof:thm:pgrip_init_rs:bnd}
\end{equation}

On the other hands, since $\Pi_{\widehat{J}_2}^\perp v = 0$, by Lemma~\ref{lemma:init_w_good_est_supp2},
we get a sufficient condition for (\ref{eq:conv_spf_initcond}) given by
\begin{equation}
\delta + (1+\delta) \nu < \frac{\sin \omega_{\sup} \norm{\Pi_{\widehat{J}_1} u}_2}{1 + \sin \omega_{\sup}}.
\label{eq:proof:thm:pgrip_init_rs:cond2}
\end{equation}

Finally, note that (\ref{eq:proof:thm:pgrip_init_rs:cond}) also implies (\ref{eq:proof:thm:pgrip_init_rs:cond2})
since (\ref{eq:proof:thm:pgrip_init_rs:bnd}) holds.

\subsection{Proof of Theorem~\ref{thm:conv_scspf}}
\label{subsec:proof_rankr_conv}
The proof will use the following lemma.
\begin{lemma}[Iterative update with subspace concatenation]
\label{lemma:iter_update_rankr}
Suppose the hypotheses of Theorem~\ref{thm:conv_scspf} hold.
Then, for all $t \geq 1$,
\[
\frac{\fnorm{P_{\R(U_t)^\perp P_{\R(U)}}}}{\sqrt{r}}
\leq \frac{2 \kappa C_\delta^{\mathrm{HTP}}}{\sqrt{1-\norm{P_{\R(V_0)^\perp} P_{\R(V)}}^2}}
\left[
\frac{\delta \fnorm{P_{\R(V_{t-1})^\perp} P_{\R(V)}}}{\sqrt{r}}
+ \frac{(1+\delta) \norm{z}_2}{\norm{\A(X)}_2}
\right]
\]
and
\[
\frac{\fnorm{P_{\R(V_t)^\perp P_{\R(V)}}}}{\sqrt{r}}
\leq \frac{2 \kappa C_\delta^{\mathrm{HTP}}}{\sqrt{1-\norm{P_{\R(U_0)^\perp} P_{\R(U)}}^2}}
\left[
\frac{\delta \fnorm{P_{\R(U_{t-1})^\perp} P_{\R(U)}}}{\sqrt{r}}
+ \frac{(1+\delta) \norm{z}_2}{\norm{\A(X)}_2}
\right].
\]
\end{lemma}

\begin{proof}[Proof of Lemma~\ref{lemma:iter_update_rankr}]
We provide the proof only for the first part.
The proof for the second part follows by symmetry.

Let $\widehat{V} \in \cz^{n_2 \times r}$ denote an orthonormal basis for the subspace spanned by $V_{t-1}$,
where $V_{t-1}$ is the estimate of $V$ obtained in the previous iteration.
Let $V_0 \in \cz^{n_2 \times r}$ denote the initial estimate of $V$.
Note that $V_0$ and $V$ were respectively obtained as singular vectors of a certain matrix.
Then we have $\widehat{V}^* \widehat{V} = V_0^* V_0 = I_r$.

The algorithm uses the subspace spanned by $\widetilde{V}$ and does not depend on a specific choice of an orthonormal basis.
Our proof will use $\widetilde{V}$ constructed as follow.
Let $Q$ be a matrix representing an orthonormal basis of $P_{\R(\widehat{V})^\perp} \R(V_0)$.
Let $\widetilde{V} = [\widehat{V}, Q]$.
Then, $\widetilde{V}$ satisfies
\[
\R(\widetilde{V}) = \R(\widehat{V}) + \R(V_0)
\]
and
\[
\widetilde{V} e_k = \widehat{V} e_k, \quad \forall k = 1,\ldots,r.
\]
Let $\tilde{r}$ denote the rank of $\widetilde{V}$.
Then, $r \leq \tilde{r} \leq 2r$.

Since $\widetilde{V}^* \widetilde{V}$ is an identity matrix,
the projection operator $P_{\R(\widetilde{V})}$ is expressed as
\[
P_{\R(\widetilde{V})} = \widetilde{V} (\widetilde{V}^* \widetilde{V})^{-1} \widetilde{V}^* = \widetilde{V} \widetilde{V}^*.
\]

The measurement vector $b$ is then written as
\begin{align*}
b {} & = \A(U \Lambda V^*) + z \\
{} & = \A[U \Lambda V^* (P_{\R(\widetilde{V})} + P_{\R(\widetilde{V})^\perp})] + z \\
{} & = \A(U \Lambda V^* P_{\R(\widetilde{V})}) + \A(U \Lambda V^* P_{\R(\widetilde{V})^\perp}) + z \\
{} & = \A( U \Lambda V^* \widetilde{V} \widetilde{V}^*) + \A(U \Lambda V^* P_{\R(\widetilde{V})^\perp}) + z \\
{} & = [\calF(\widetilde{V})] (U \Lambda V^* \widetilde{V}) + [\calF(P_{\R(\widetilde{V})^\perp} V)] (U \Lambda) + z.
\end{align*}

In the scenario $s_1 < n_1$,
$\widetilde{U}$ is obtained by B-HTP with the sensing linear operator $\calF(\widetilde{V}): \cz^{n_1 \times \tilde{r}} \to \cz^m$.
The performance guarantee for HTP in Lemma~\ref{lemma:htp} extends to the case of B-HTP.
The derivation is done in a straightforward way that involves replacing a few symbols.
By this RIP-guarantee for B-HTP, the error in the estimate $\widetilde{U}$ of $U \Lambda V^* \widetilde{V}$ is upper-bounded by
\begin{equation}
\label{eq:pf:lemma:iter_update_rankr:bnd1}
\fnorm{\widetilde{U} - U \Lambda V^* \widetilde{V}}
\leq C_\delta^{\mathrm{HTP}} |\langle W, [\calF(\widetilde{V})]^* [\A(U \Lambda V^* P_{\R(\widetilde{V})^\perp}) + z]\rangle|
\end{equation}
for some $W \in \cz^{n_1 \times \tilde{r}}$ such that $\fnorm{W} = 1$ and $W$ is row $s_1$-sparse.
By applying the triangle inequality to the right-hand-side of \eqref{eq:pf:lemma:iter_update_rankr:bnd1},
$\fnorm{\widetilde{U} - U \Lambda V^* \widetilde{V}}$ is further upper-bounded by
\begin{equation}
\label{eq:pf:lemma:iter_update_rankr:bnd2}
\fnorm{\widetilde{U} - U \Lambda V^* \widetilde{V}}
\leq C_\delta^{\mathrm{HTP}} |\langle W, [\calF(\widetilde{V})]^* \A(U \Lambda V^* P_{\R(\widetilde{V})^\perp}) \rangle|
+ C_\delta^{\mathrm{HTP}} |\langle W, [\calF(\widetilde{V})]^* z \rangle|.
\end{equation}
Then the right-hand-side of \eqref{eq:pf:lemma:iter_update_rankr:bnd2} is upper-bounded as follows.
Note that $\langle W, [\calF(\widetilde{V})]^* \A(U \Lambda V^* P_{\R(\widetilde{V})^\perp}) \rangle$ is rewritten as
\begin{align*}
\langle W, [\calF(\widetilde{V})]^* \A(U \Lambda V^* P_{\R(\widetilde{V})^\perp}) \rangle
{} & = \langle [\calF(\widetilde{V})] W, \A(U \Lambda V^* P_{\R(\widetilde{V})^\perp}) \rangle \\
{} & = \langle \A(W \widetilde{V}^*), \A(U \Lambda V^* P_{\R(\widetilde{V})^\perp}) \rangle \\
{} & = \langle W \widetilde{V}^*, \A^*\A(U \Lambda V^* P_{\R(\widetilde{V})^\perp}) \rangle \\
{} & = \langle W \widetilde{V}^*, (\A^*\A-\id)(U \Lambda V^* P_{\R(\widetilde{V})^\perp}) \rangle,
\end{align*}
where the last step follows since $\langle W \widetilde{V}^*, U \Lambda V^* P_{\R(\widetilde{V})^\perp} \rangle = 0$.
Thus, by the rank-$2r$ and doubly $(3s_1,3s_2)$-sparse RIP of $\A$,
the first term in the right-hand-side of \eqref{eq:pf:lemma:iter_update_rankr:bnd2} is upper-bounded by
\begin{align*}
C_\delta^{\mathrm{HTP}} |\langle W, [\calF(\widetilde{V})]^* \A(U \Lambda V^* P_{\R(\widetilde{V})^\perp}) \rangle|
{} & \leq C_\delta^{\mathrm{HTP}} \delta \fnorm{W \widetilde{V}^*} \fnorm{U \Lambda V^* P_{\R(\widetilde{V})^\perp}} \\
{} & \leq C_\delta^{\mathrm{HTP}} \delta \norm{\Lambda} \fnorm{P_{\R(V)} P_{\R(\widetilde{V})^\perp}} \\
{} & \leq C_\delta^{\mathrm{HTP}} \delta \norm{\Lambda} \fnorm{P_{\R(V)} P_{\R(\widehat{V})^\perp}},
\end{align*}
where the last step holds since $\R(\widehat{V}) \subset \R(\widetilde{V})$.
On the other hand, since $\langle W, [\calF(\widetilde{V})]^* z \rangle$ is rewritten as
\[
\langle W, [\calF(\widetilde{V})]^* z \rangle = \langle [\calF(\widetilde{V})] W, z \rangle = \langle \A(W \widetilde{V}^*), z \rangle,
\]
by the RIP of $\A$ together with the fact that $W$ is row-$s_1$-sparse and $\widetilde{V}$ is row-$2s_2$-sparse,
the second term in the right-hand-side of \eqref{eq:pf:lemma:iter_update_rankr:bnd2} is upper-bounded by
\begin{align*}
C_\delta^{\mathrm{HTP}} |\langle W, [\calF(\widetilde{V})]^* z \rangle|
\leq C_\delta^{\mathrm{HTP}} \norm{\A(W \widetilde{V}^*)}_2 \norm{z}_2
\leq C_\delta^{\mathrm{HTP}} \sqrt{1+\delta} \fnorm{W \widetilde{V}^*} \norm{z}_2
\leq C_\delta^{\mathrm{HTP}} \sqrt{1+\delta} \norm{z}_2.
\end{align*}
Therefore, \eqref{eq:pf:lemma:iter_update_rankr:bnd1} implies
\begin{equation}
\label{eq:pf:lemma:iter_update_rankr:bnd3}
\fnorm{\widetilde{U} - U \Lambda V^* \widetilde{V}}
\leq C_\delta^{\mathrm{HTP}} \left( \delta \norm{\Lambda} \fnorm{P_{\R(V)} P_{\R(\widehat{V})^\perp}} + \sqrt{1+\delta} \norm{z}_2 \right).
\end{equation}

On the other scenario $s_1 = n_1$, matrix $\widetilde{U}$ is obtained by
\begin{align*}
\widetilde{U} {} & = [\calF(\widetilde{V})]^\dagger b \\
{} & = [\calF(\widetilde{V})]^\dagger
\left(
[\calF(\widetilde{V})] (U \Lambda V^* \widetilde{V}) + [\calF(P_{\R(\widetilde{V})^\perp} V)] (U \Lambda) + z
\right) \\
{} & = U \Lambda V^* \widetilde{V} + [\calF(\widetilde{V})]^\dagger [\calF(P_{\R(\widetilde{V})^\perp} V)] (U \Lambda) +
[\calF(\widetilde{V})]^\dagger z
\end{align*}
Then, by Lemma~\ref{lemma:reduced_rip}, we have
\begin{equation}
\label{eq:pf:lemma:iter_update_rankr:bnd3LS}
\fnorm{\widetilde{U} - U \Lambda V^* \widetilde{V}}
\leq \frac{1}{1-\delta} \left( \delta \norm{\Lambda} \fnorm{P_{\R(V)} P_{\R(\widehat{V})^\perp}} + \sqrt{1+\delta} \norm{z}_2 \right).
\end{equation}
Since $C_\delta^{\mathrm{HTP}} \geq \frac{1}{1-\delta}$, \eqref{eq:pf:lemma:iter_update_rankr:bnd3LS} implies \eqref{eq:pf:lemma:iter_update_rankr:bnd3}. Therefore, we will use \eqref{eq:pf:lemma:iter_update_rankr:bnd3} regardless of whether $s_1 < n_1$ or not.

Next, we derive a lower bound on the left-hand-side of \eqref{eq:pf:lemma:iter_update_rankr:bnd3}.
Let $\widehat{U} \in \cz^{n_1 \times r}$ be a matrix with the first $r$ left singular vectors of $\widetilde{U}$.
Then, $\widehat{U} \widehat{U}^* \widetilde{U}$ is the best rank-$r$ approximation of $\widetilde{U}$
and by the optimality we have
\begin{equation}
\label{eq:pf:lemma:iter_update_rankr:bnd4}
\fnorm{\widehat{U} \widehat{U}^* \widetilde{U} - U \Lambda V^* \widetilde{V}}
\leq \fnorm{\widehat{U} \widehat{U}^* \widetilde{U} - \widetilde{U}}
+ \fnorm{\widetilde{U} - U \Lambda V^* \widetilde{V}}
\leq 2 \fnorm{\widetilde{U} - U \Lambda V^* \widetilde{V}}.
\end{equation}
On the other hand, since $\widehat{U} \widehat{U}^* \widetilde{U}$ is spanned by $\widehat{U}$, we have
\begin{equation}
\label{eq:pf:lemma:iter_update_rankr:bnd5}
\begin{aligned}
\fnorm{\widehat{U} \widehat{U}^* \widetilde{U} - U \Lambda V^* \widetilde{V}}
{} & \geq \fnorm{P_{\R(\widehat{U})^\perp} U \Lambda V^* \widetilde{V}} \\
{} & = \fnorm{P_{\R(\widehat{U})^\perp} U \Lambda V^* \widetilde{V} \widetilde{V}^*} \\
{} & = \fnorm{P_{\R(\widehat{U})^\perp} U \Lambda V^* P_{\R(\widetilde{V})}} \\
{} & \geq \fnorm{P_{\R(\widehat{U})^\perp} U \Lambda V^* P_{\R(V_0)}} \\
{} & \geq \fnorm{P_{\R(\widehat{U})^\perp} U} \sigma_r(\Lambda) \sigma_r(V^* P_{\R(V_0)}) \\
{} & = \fnorm{P_{\R(\widehat{U})^\perp} U U^*} \sigma_r(\Lambda) \sigma_r(V V^* P_{\R(V_0)}) \\
{} & = \fnorm{P_{\R(\widehat{U})^\perp} P_{\R(U)}} \sigma_r(\Lambda) \sigma_r(P_{\R(V)} P_{\R(V_0)}) \\
{} & = \fnorm{P_{\R(\widehat{U})^\perp} P_{\R(U)}} \sigma_r(\Lambda) \sqrt{1 - \norm{P_{\R(V_0)^\perp} P_{\R(V)}}^2},
\end{aligned}
\end{equation}
where the second inequality holds since $\R(V_0) \subset \R(\widetilde{V})$.

By combining \eqref{eq:pf:lemma:iter_update_rankr:bnd3}, \eqref{eq:pf:lemma:iter_update_rankr:bnd4} and \eqref{eq:pf:lemma:iter_update_rankr:bnd5},
we get
\begin{equation}
\label{eq:pf:lemma:iter_update_rankr:bnd6}
\fnorm{P_{\R(\widehat{U})^\perp} P_{\R(U)}}
\leq 2 C_\delta^{\mathrm{HTP}}
\left(
\frac{\kappa \delta \fnorm{P_{\R(\widehat{V})^\perp} P_{\R(V)}}}{\sqrt{1-\norm{P_{\R(V_0)^\perp} P_{\R(V)}}^2}}
+ \frac{\sqrt{1+\delta} \norm{z}_2}{\sigma_r(\Lambda) \sqrt{1-\norm{P_{\R(V_0)^\perp} P_{\R(V)}}^2}}
\right).
\end{equation}
Finally, note that
\begin{equation}
\label{eq:pf:lemma:iter_update_rankr:bnd7}
\sqrt{r} \sigma_r(\Lambda) = \frac{\sqrt{r} \norm{\Lambda}}{\kappa} \geq \frac{\fnorm{X}}{\kappa}.
\end{equation}
Applying \eqref{eq:pf:lemma:iter_update_rankr:bnd7} to \eqref{eq:pf:lemma:iter_update_rankr:bnd6} completes the proof.
\end{proof}

Define
\[
\rho \triangleq \widetilde{C}_\delta \cdot \delta \quad \text{and} \quad \tau \triangleq \widetilde{C}_\delta \cdot (1+\delta),
\]
where
\begin{align*}
\widetilde{C}_\delta {} & \triangleq \max\left\{
\frac{2 \kappa C_\delta^{\mathrm{HTP}}}{\sqrt{1-\norm{P_{\R(V_0)^\perp} P_{\R(V)}}^2}}
,
\frac{2 \kappa C_\delta^{\mathrm{HTP}}}{\sqrt{1-\norm{P_{\R(U_0)^\perp} P_{\R(U)}}^2}}
\right\}.
\end{align*}
Then, by Lemma~\ref{lemma:iter_update_rankr}, we have
\begin{equation}
\label{eq:pf:thm:iter_update_rankr:recursion1}
\frac{\fnorm{P_{\R(U_t)^\perp P_{\R(U)}}}}{\sqrt{r}}
\leq
\rho \frac{\fnorm{P_{\R(V_{t-1})^\perp} P_{\R(V)}}}{\sqrt{r}}
+ \tau \frac{\norm{z}_2}{\norm{\A(X)}_2}
\end{equation}
and
\begin{equation}
\label{eq:pf:thm:iter_update_rankr:recursion2}
\frac{\fnorm{P_{\R(V_t)^\perp P_{\R(V)}}}}{\sqrt{r}}
\leq
\rho \frac{\fnorm{P_{\R(U_{t-1})^\perp} P_{\R(U)}}}{\sqrt{r}}
+ \tau \frac{\norm{z}_2}{\norm{\A(X)}_2}
\end{equation}
for all $t \geq 1$.
Furthermore, by the choice of $\delta = \frac{0.04}{\kappa}$ and \eqref{eq:conv_scspf_initcond}, we have $0 < \rho < 1$.
Therefore, from the alternating recursion formula in \eqref{eq:pf:thm:iter_update_rankr:recursion1} and \eqref{eq:pf:thm:iter_update_rankr:recursion2}, the left-hand-side of \eqref{eq:pf:thm:iter_update_rankr:recursion1} converges linearly as
\begin{equation}
\label{eq:pf:thm:iter_update_rankr:limtheta}
\limsup_{t \to \infty} \frac{\fnorm{P_{\R(U_t)^\perp P_{\R(U)}}}}{\sqrt{r}} \leq \frac{\tau}{1-\rho} \cdot \frac{\norm{z}_2}{\norm{\A(X)}_2}.
\end{equation}

It remains to bound the error term $\fnorm{U \Lambda V^* - U_t V_t^*}$ in terms of $\fnorm{P_{\R(U_t)^\perp P_{\R(U)}}}$.
Without loss of generality, we may assume that $U_t^* U_t = I_r$.
(Indeed, the update of $V_t$ depends on $U_t$ only in terms of the subspace spanned by it.)
Note that the estimation error is decomposed as
\begin{align*}
U \Lambda V^* - U_t V_t^*
{} & = P_{\R(U_t)} U \Lambda V^* - U_t V_t^* + P_{\R(U_t)^\perp} U \Lambda V^* \\
{} & = U_t [(U_t^* U) \Lambda V^* - V_t^*] + P_{\R(U_t)^\perp} U \Lambda V^*.
\end{align*}
Furthermore, it follows from the RIP-guarantee of B-HTP that $V_t$ satisfies
\begin{align*}
\fnorm{V_t^* - U_t^* U \Lambda V^*}
{} & \leq 2 C_\delta^{\mathrm{HTP}} \left[ \delta \norm{\Lambda} \fnorm{P_{\R(U_t)^\perp} U} + \sqrt{1+\delta} \norm{z}_2 \right] \\
{} & \leq 2 C_\delta^{\mathrm{HTP}} \left[ \delta \norm{\Lambda} \fnorm{P_{\R(U_t)^\perp} U} + \fnorm{X} (1+\delta) \nu \right].
\end{align*}
Therefore,
\begin{align*}
\fnorm{U \Lambda V^* - U_t V_t^*}^2
{} & = \fnorm{P_{\R(U_t)} (U \Lambda V^* - U_t V_t^*)}^2 + \fnorm{P_{\R(U_t)^\perp} (U \Lambda V^* - U_t V_t^*)}^2 \\
{} & = \fnorm{U_t [(U_t^* U) \Lambda V^* - V_t^*]}^2 + \fnorm{P_{\R(U_t)^\perp} U \Lambda V^*}^2 \\
{} & \leq \fnorm{U_t^* U \Lambda V^* - V_t^*}^2 + \norm{\Lambda}^2 \fnorm{P_{\R(U_t)^\perp} U}^2 \\
{} & \leq \norm{\Lambda}^2 \left[1 + 2 (2 \delta C_\delta^{\mathrm{HTP}})^2\right] \fnorm{P_{\R(U_t)^\perp} P_{\R(U)}}^2
+ 2 \left[\fnorm{X} (1+\delta) C_\delta^{\mathrm{HTP}} \nu \right]^2,
\end{align*}
which implies
\begin{equation}
\label{eq:pf:thm:iter_update_rankr:errbnd}
\frac{\fnorm{X - U_t V_t^*}}{\fnorm{X}}
\leq \kappa \sqrt{1 + 8 (\delta C_\delta^{\mathrm{HTP}})^2} ~ \frac{\fnorm{P_{\R(U_t)^\perp} P_{\R(U)}}}{\sqrt{r}}
+ \sqrt{2} (1+\delta) C_\delta^{\mathrm{HTP}} \frac{\norm{z}_2}{\norm{\A(X)}_2}.
\end{equation}

Finally, by applying \eqref{eq:pf:thm:iter_update_rankr:limtheta} to \eqref{eq:pf:thm:iter_update_rankr:errbnd}, we obtain
\[
\limsup_{t \to \infty} \frac{\fnorm{X - U_t V_t^*}}{\fnorm{X}}
\leq
\underbrace{ \left( \frac{\tau \kappa}{1-\rho} \sqrt{1 + 8 (\delta C_\delta^{\mathrm{HTP}})^2} + \sqrt{2} (1+\delta) C_\delta^{\mathrm{HTP}} \right) }_{C'}
\frac{\norm{z}_2}{\norm{\A(X)}_2},
\]
where the factor $C'$ is no greater than $55 \kappa^2 + 3 \kappa + 3$ under the hypotheses of Theorem~\ref{thm:conv_scspf}.
This completes the proof.

\subsection{Proof of Theorem~\ref{thm:pgrip_init_ds_rankr}}
\label{subsec:proof_rankr_init}
It suffices to show that the condition in \eqref{eq:conv_scspf_initcond} is satisfied by the initial estimates
$U_0 = U_0^\text{\rm Th}$ and $V_0 = V_0^\text{\rm Th}$,
where $U_0$ and $V_0$ are generated respectively by Algorithms~\ref{alg:initU0} and \ref{alg:initV0}.
Then, the result follows from Theorem~\ref{thm:conv_scspf}.
Specifically, we need to show that $V_0^\text{\rm Th}$ satisfies
\begin{equation}
\label{eq:pf:thm:pgrip_init_ds_rankr:resV0}
\norm{P_{\R(V)^\perp} P_{\R(V_0^\text{\rm Th})}} < 0.95,
\end{equation}
and $U_0^\text{\rm Th}$ satisfies
\begin{equation}
\label{eq:pf:thm:pgrip_init_ds_rankr:resU0}
\norm{P_{\R(U)^\perp} P_{\R(U_0^\text{\rm Th})}} < 0.95.
\end{equation}
We present the proof only for \eqref{eq:pf:thm:pgrip_init_ds_rankr:resV0}.
The proof for \eqref{eq:pf:thm:pgrip_init_ds_rankr:resU0} follows by symmetry.

The first step of our proof is to show that the support-estimate $\widehat{J}_1$ (resp. $\widehat{J}_2$) in Algorithm~\ref{alg:initV0}
includes $\widetilde{J}_1$ (resp. $\widetilde{J}_2$), which are defined in Theorem~\ref{thm:pgrip_ds_iidG_rankr}.
The following lemmas provide a sufficient condition for the containment.

\begin{lemma}
\label{lemma:includeJtilde1rankr}
Suppose the hypotheses of Theorem~\ref{thm:pgrip_init_ds_rankr} hold.
Let $J_1 \subset [n_1]$ denote the set of the indices of the nonzero rows of $U$.
Let $\widetilde{J}_1 \subset J_1$ be defined in \eqref{eq:def_widetildeJs}.
If
\begin{equation}
\label{eq:lemma:includeJtilde1rankr:cond}
2 \kappa \left[\delta + (1+\delta) \nu \right] < \sigma_r(U^* \Pi_{\widetilde{J}_1}).
\end{equation}
then $\widetilde{J}_1 \subset \widehat{J}_1$.
\end{lemma}

\begin{proof}[Proof of Lemma~\ref{lemma:includeJtilde1rankr}]
See Appendix~\ref{sec:pf:lemma:includeJtilde1rankr}.
\end{proof}

\begin{lemma}
\label{lemma:includeJtilde2rankr}
Suppose the hypotheses of Theorem~\ref{thm:pgrip_init_ds_rankr} hold.
Let $J_2 \subset [n_2]$ denote the set of the indices of the nonzero rows of $V$.
Let $\widetilde{J}_2$ be a subset of $J_2$.
If
\begin{equation}
\label{eq:lemma:includeJtilde2rankr:cond}
2 \kappa \left[\delta + (1+\delta) \nu \right] < \sigma_r(U^* \Pi_{\widetilde{J}_1}) \sigma_r(V^* \Pi_{\widetilde{J}_2}).
\end{equation}
then $\widetilde{J}_2 \subset \widehat{J}_2$.
\end{lemma}

\begin{proof}[Proof of Lemma~\ref{lemma:includeJtilde2rankr}]
See Appendix~\ref{sec:pf:lemma:includeJtilde2rankr}.
\end{proof}

It is straightforward to verify that \eqref{eq:lemma:includeJtilde1rankr:cond} and \eqref{eq:lemma:includeJtilde2rankr:cond}
are satisfied under the hypotheses of Theorem~\ref{thm:pgrip_init_ds_rankr}.

Next, from the results by Lemmas~\ref{lemma:includeJtilde1rankr} and \ref{lemma:includeJtilde2rankr},
we derive an upper bound on $\norm{P_{\R(V_0)^\perp} P_{\R(V)}}$ using the $\sin\theta$ theorem,
which is summarized in the following lemma.

\begin{lemma}
\label{lemma:init_w_good_est_supp2_rankr}
Suppose the hypotheses of Theorem~\ref{thm:pgrip_init_ds_rankr} hold.
Let $\widehat{J}_1 \subset [n_1]$ and $\widehat{J}_2 \subset [n_2]$ satisfy $|\widehat{J}_1| \leq s_1$ and $|\widehat{J}_2| \leq s_2$, respectively.
Let $V_0$ be a matrix whose columns are the $r$ leading singular vectors of $\Pi_{\widehat{J}_1} X \Pi_{\widehat{J}_2}$.
Then,
\begin{equation}
\label{eq:init_ub_rankr}
\norm{P_{\R(V_0)^\perp} P_{\R(V)}}
\leq \frac{\delta + \sqrt{1 - \sigma_r^2(V^* \Pi_{\widehat{J}_2})} + (1+\delta) \nu}{\sigma_r(\Pi_{\widehat{J}_1} U)/\kappa - \delta - (1+\delta) \nu}.
\end{equation}
\end{lemma}

\begin{proof}[Proof of Lemma~\ref{lemma:init_w_good_est_supp2_rankr}]
See Appendix~\ref{sec:pf:lemma:init_w_good_est_supp2_rankr}.
\end{proof}

We verify that \eqref{eq:pf:thm:pgrip_init_ds_rankr:resV0}
is obtained from the upper bound in Lemma~\ref{lemma:init_w_good_est_supp2_rankr}
under the hypotheses of Theorem~\ref{thm:pgrip_init_ds_rankr}.

Finally, in the case when $s_1 = n_1$ and $s_2 = n_2$,
we have $\widehat{J}_1 = [n_1]$ and $\widehat{J}_2 = [n_2]$.
These support estimates are trivially obtained without requiring that
$\sigma_r(V^* \Pi_{\widehat{J}_2})$ and $\sigma_r(\Pi_{\widehat{J}_1} U)$ exceed a constant.
Furthermore, it follows that
$\sigma_r(V^* \Pi_{\widehat{J}_2}) = 1$ and $\sigma_r(\Pi_{\widehat{J}_1} U) = 1$.
Thus \eqref{eq:init_ub_rankr} reduces to
\begin{equation}
\label{eq:init_ub_rankr_nonsparse}
\norm{P_{\R(V_0)^\perp} P_{\R(V)}}
\leq \frac{\delta + (1+\delta) \nu}{1/\kappa - \delta - (1+\delta) \nu}.
\end{equation}
Recall that it was assumed that $\delta = \frac{0.04}{\kappa}$ and $\nu = \frac{0.04}{\kappa}$ in Theorem~\ref{thm:pgrip_init_ds_rankr}.
Then the upper bound in the right-hand-side of \eqref{eq:init_ub_rankr_nonsparse} is less than 0.95 for any $\kappa \geq 1$. Therefore, the condition $\kappa \leq 4$ is not necessary in this case.
This completes the proof.

\subsection{Proof of \prettyref{thm:lb}}
\label{sec:pf-lb}

The proof of the rate-distortion lower bound for rank-$r$ matrix in \prettyref{thm:lb} relies on two auxiliary results, which we present first:
\begin{enumerate}[(a)]
	\item \prettyref{thm:RDv}: a rate-distortion lower bound for a $r$-dimensional uniform random subspace in $\complex^n$, where the distortion measure is the squared distance between subspaces, more precisely, the Frobenius norm between the respective projection matrices;
	\item \prettyref{thm:RDT}: a rate-distortion lower bound for a complex orthogonal matrix uniformly distributed with respect to the quadratic distortion.
\end{enumerate}
Both results are tight asymptotically (see \prettyref{rmk:dof-vv}); however, the main point here is a non-asymptotic bound.

\begin{theorem}	\label{thm:RDv}
	Let $V$ be uniformly distributed on the $r$-dimensional complex Stiefel manifold $\sfV(\complex^n, r)$.
	There exists a universal constant $c>0$ such that
	for any $n \geq 2$, $r \in [n]$ and $D > 0$,
	\begin{align}
	\inf_{P_{\hV\hV^*|VV^*}} \{I(VV;\hV\hV^*) : \Expect \|VV^*-\hV\hV^*\|_{\mathrm{F}}^2 \leq r D\} \geq & ~ 	(n-r)r \log \frac{c}{D},
		\label{eq:RDvv}
\end{align}
where the infimum is taken over all probability transition kernels such that $\hV \in \sfV(\complex^n,r)$.
\end{theorem}

\begin{theorem}	\label{thm:RDT}
	Let $T$ be uniformly distributed on the orthogonal group $O(r)$ in $\complex^r$.	There exists a universal constant $c'>0$ such that
	for any $r \in \naturals$ and $D > 0$,
	\begin{align}
	\inf_{P_{\hT|T}} \{I(T;\hT) : \Expect \fnorm{T-\hT}^2 \leq r D\} \geq & ~ \frac{r^2}{2} \log \frac{c'}{D},
		\label{eq:RDT}
\end{align}
where the infimum is taken over all probability transition kernels from $O(r)$ to itself.
\end{theorem}

\begin{remark}
\label{rmk:dof-vv}
The lower bound in \prettyref{thm:RDv} is in fact sharp within constant factors and, in addition, asymptotically sharp in the low-distortion regime, and the rate-distortion function on the left-hand side of \prettyref{eq:RDvv} is in fact $(n-r)r \log \frac{1}{D}(1+o(1))$ as $D \to 0$, since a matching upper bound can be obtained by quantization and using the covering number bound of the Grassmannian manifold \cite{Szarek82}. The pre-log factor $(n-r)r$ in \prettyref{eq:RDvv} deserves a careful explanation. We first recall that the low-distortion asymptotics of the rate-distortion function obtained in \cite{KD94} for mean-square distortion:
\begin{equation}
R_X(D) = \frac{d(X)}{2} \log \frac{1}{D}(1+o(1)),	
	\label{eq:dk}
\end{equation}
 where $d(X)$ is the \emph{information dimension} of the random vector $X$ \cite{renyi}, which, in the absolute continuous case, coincides with the (real) topological dimension of the support of $X$.
As mentioned in \prettyref{rmk:dof}, the number of free (real) variables in $V$ is $2nr-r^2$.
	Furthermore, the loss function $\|VV^*-\hV\hV^*\|$ corresponds to the subspace distance which is rotationally invariant. This further reduces the number of free variables by the degrees of freedom of the orthogonal group $O(r)$ in $\complex^r$, which is $\sum_{i=1}^r (2r-2i+1) = r^2$. Therefore under the subspace distance distortion metric, the \emph{effective} number of degrees of freedom in $V$ is $2r(n-r)$, 	and the corresponding rate-distortion behavior \prettyref{eq:RDvv} is consistent with the information dimension characterization \prettyref{eq:dk}.

Similarly, \prettyref{thm:RDT} is tight when $D \to 0$, which is met by the covering number bound of the orthogonal group \cite{Szarek82}.
\end{remark}

To prove \prettyref{thm:RDv} first we need an auxiliary result from linear algebra.
\begin{lemma}
\label{lmm:distance}	
Let $O(r)$ be the set of $r\times r$ complex orthogonal matrices. Then
for any $V,\hat V\in \sfV(\complex^n,r)$, there exists $ A \in O(r)$, such that
\begin{equation}
\fnorm{V - \hat{V}A} \leq \fnorm{ VV^* - \hat{V}\hat{V}^* }.
\end{equation}	
\end{lemma}
\begin{proof}
\begin{align*}
\fnorm{ VV^* - \hat{V}\hat{V}^* }^2 & = 2\sum_{i=1}^r \sin^2\theta_i = 2r - 2\sum_{i=1}^r \cos^2\theta_i \\
&= 2r - 2\sum_{i=1}^r (\Real\langle u_i, \hat{u}_i \rangle)^2
\end{align*}
where $\theta_1,\ldots,\theta_r$ are the principal angles between subspaces $\text{span}(V), \text{span}(\hat{V})$, and $U = [u_1 ,\ldots, u_r], \hat{U} = [\hat u_1  ,\ldots, \hat{u}_r]$, where $u_i, \hat{u}_i$ are the corresponding principal vectors. Note that $U = VR, \hat{U} = \hat{V}\hat{R}$, for some $R, \hat{R} \in O(r)$. Then
\begin{equation}
\fnorm{ U - \hat{U} }^2 = 2r - 2 \Real\langle U, \hat{U} \rangle = 2r - 2\sum_{i=1}^r \Real\langle u_i, \hat{u}_i \rangle = 2r - 2\sum_{i=1}^r \cos\theta_i.
\end{equation}
Since $0 \leq \cos\theta_i \leq 1$,
$\fnorm{ VV^* - \hat{V}\hat{V}^* }^2 \geq \fnorm{ U - \hat{U} }^2$. Moreover,
\begin{equation}
\fnorm{ U - \hat{U} }^2 = \fnorm{ VR - \hat{V}\hat{R} }^2 = \fnorm{ V - \hat{V}\hat{R}R^* }^2,
\end{equation}
thus by choosing $A = \hat{R}R^*$, $\fnorm{ V - \hat{V}A } \leq \fnorm{ VV^* - \hat{V}\hat{V}^* }$.
\end{proof}

\begin{proof}[Proof of \prettyref{thm:RDv}]
When $r=n$, the lower bound is automatically true (and tight since the subspace has no randomness). Henceforth we assume that $r \leq n-1$.
Also, it is sufficient to consider $D<1$.
Under the mean-square distortion, the rate-distortion function for $\calCN(0,1)$ is given by
		\begin{equation}
		R(D) = \log^+ \frac{1}{D}, \quad D > 0,
	\label{eq:RDCN}
\end{equation}
where 	$\log^+ \triangleq \max\{\log, 0\}$.
		Let $W$ be an $n\times r$ random matrix with \iid entries drawn from $\calCN(0,1)$.
As a consequence of \prettyref{eq:RDCN}, for any $a > 0$,
		\begin{equation}
	\inf_{P_{\hW|W}} \{I(W;\hat{W}): \Expect \fnorm{\hW-W}^2 \leq n r a\} \geq nr \log^+ \frac{1}{a}.
	\label{eq:IWW}
\end{equation}
Let $W=VR$ be the QR decomposition of $W$, where $V\in \sfV(\complex^n,r)$ and $R$ is a $r\times r$ upper triangular matrix with real-valued diagonals and complex-valued off-diagonals. Since the law of $W$ is left rotationally invariant, $V$ and $R$ are independent. Moreover, $V$ is Haar distributed; the entries of $R$ are independent, where the off-diagonals $\{R_{ij}: i < j\}$ are standard complex Gaussian
 and $2R_{ii}^2$ are independent $\chi^2_{2(n-i+1)}$ for $i \in [r]$ (see, \eg, \cite[Lemma 2.1]{TV04}).

The main idea of obtaining the lower bound \prettyref{eq:RDvv} is to combine a given compressor of the column subspace of $W$ together with
that of the matrix $R$ to yield a lossy compressor of the Gaussian matrix $W$, and the overall performance must obey the rate-distortion function of $W$. To implement this program, fix $0 < D < 1$ and fix any probability transition kernel $P_{\hV\hV^*|VV^*}$ such that $\hV \in \sfV(\complex^n,r)$ and $\Expect{\fnorm{VV^*-\hV\hV^*}^2} \leq Dr$.
By \prettyref{lmm:distance}, there exists $A=A(V,\hV) \in O(r)$, such that $\fnorm{V-\hV A} \leq \fnorm{VV^*-\hV\hV^*}$.
Next we use the metric entropy bound for the orthogonal group to produce a quantized version of $A$.
By \cite{Szarek82} (see also \cite{Szarek98}), there exists a universal constant $c_0>1$, such that
the covering number of $O(r)$ with respect to $\fnorm{\cdot}$ is at most $(\frac{\sqrt{c_0 r}}{\epsilon})^{r^2}$ for any $r\in\naturals$ and any $\epsilon\in(0,\sqrt{r})$, where the $r^2$ is the (real) topological dimension of $O(r)$ and $\sqrt{r}$ is the diameter.
Therefore for $\epsilon = \sqrt{r D}$, there exists $T_1,\ldots,T_m \in O(r)$ with $m \leq (\frac{c_0}{D})^{r^2/2}$, that constitute an $\epsilon$-covering of $O(r)$, namely,
for any $S\in O(r)$, there exists $i=i(S) \in [m]$ such that $\fnorm{S-T_{i(S)}} \leq \epsilon$.
Let $A_m = T_{i(A)}$ denote the closest $T_i$ to $A$. Then $\fnorm{A_m-A} \leq \sqrt{rD}$.
Set $\tV = \hV A_m$. Then
\[
\fnorm{V-\tV}\leq \fnorm{V-\hV A} + \fnorm{\hV(A-A_m)} \leq \fnorm{V-\hV A} + \sqrt{rD}
\]
and hence
\begin{equation}
\Expect\fnorm{V-\tV}^2 \leq 2 \Expect\fnorm{V-\hV A}^2 + 2 rD \leq 2 \Expect\fnorm{VV^*-\hV\hV^*}^2  + 2 rD \leq 4rD.	
	\label{eq:Vcompress}
\end{equation}


Next we use $P_{\tV|V}$ to design a lossy compressor for the standard Gaussian matrix $W$. Fix a transition kernel $P_{\tR|R}$ to be specified later
and set $\tW = \tV\tR $.
The dependence diagram for all random variables is as follows:
\[
\begin{tikzcd}
R	\arrow{r}							& \tR	\arrow{r}						& \tW 							& \\		
V \arrow{rd} \arrow{r} 	& \hV \arrow{r} \arrow{d} & \tV	\arrow{u} 		& \\
												& A \arrow{r}					&  A_m \arrow{u} &
\end{tikzcd}
\]
which implies that
\begin{align}
I(W;\tW)
= & ~ I(V, R;\tV,\tR) = I(R;\tR)+I(V;\tV) \leq I(R;\tR)+I(V;\hV, A_m)	\nonumber \\ 
\leq & ~ I(R;\tR)  + I(V;\hV) + H(A_m) \nonumber \\
\leq & ~ I(R;\tR)  + I(V;\hV) + \frac{r^2}{2} \log  \frac{c_0}{D}.	\label{eq:Ivv}
\end{align}
where the last step follows from the fact that $A_m$ is a random variable which takes no more than $m$ values, hence $H(A_m) \leq \log m \leq \frac{r^2}{2} \log  \frac{c_0}{D}$.

Note that for each $i\in[r]$,
$\Expect R_{ii}^2=2m$ and $\Expect R_{ii} = \frac{\Gamma(m + \frac{1}{2})}{\Gamma(m)}$, where $m=n-i+1 \geq 2$ by assumption that $r < n$.
Since $\frac{\Gamma(m + \frac{1}{2})}{\Gamma(m) \sqrt{m-1}} \geq 1$ for any $m \geq 2$, we have $\var R_{ii} \leq 1$.
By \cite[Theorem 4.3.3]{berger}, the rate-distortion function of $R_{ii}$ is majorized by that of the standard normal distribution, \ie,
\[
\min_{P_{\tR_{ii}|R_{ii}}: \Expect (\tR_{ii}-R_{ii})^2 \leq D} I(R_{ii};\tR_{ii}) \leq \frac{1}{2} \log^+ \frac{1}{D}, \quad D > 0.
\]
For the off-diagonal $R_{ij}$ which is independent standard complex normal, by \prettyref{eq:RDCN}, there exists $P_{\tR_{ij}|R_{ij}}$ such that $\Expect |\tR_{ij}-R_{ij}|^2 \leq D$ and $I(R_{ij};\tR_{ij}) =  \log^+ \frac{1}{D}$.
Let $P_{\tR|R} = \prod_{i\leq j} P_{\tR_{ij}|R_{ij}}$. Then
\begin{equation}
	\Expect \fnorm{\tR-R}^2 \leq \frac{r(r+1)D}{2}
	\label{eq:ERR}
\end{equation}
and
\begin{equation}
\label{eq:IRR}
I(R;\tR) = \sum_{1 \leq i\leq j\leq r}I(R_{ij};\tR_{ij}) \leq \pth{\frac{r(r-1)}{2} + \frac{r}{2}} \log^+ \frac{1}{D} = \frac{r^2}{2}  \log^+ \frac{1}{D}.
\end{equation}
To bound the overall distortion for $W$, note that
\begin{align}
\fnorm{W-\tW}^2
\leq & ~ 2\fnorm{(V-\tV)R}^2 + 2\fnorm{V(R-\tR)}^2 \label{eq:b1}\\
\leq & ~ 2\fnorm{V-\tV}^2 \|R\|^2 + 2\fnorm{R-\tR}^2. \label{eq:b2}
\end{align}
Note that $\|R\|=\|W\|$ is the largest singular value of an $n\times r$ standard complex Gaussian matrix. It is well-known that\footnote{This follows from Gordon's inequality and Davidson-Szarek theorem 
cf.~\eg,\cite[Theorems 2.6 and 2.11]{Davidson01}}
$\Expect \|R\|^2 \leq c_1 n$ for some absolute constant $c_1$.
Since $R$ is independent of $\{V,\tV\}$, in view of \prettyref{eq:Vcompress} and \prettyref{eq:ERR} we have
\begin{equation}
\Expect \|W-\tW\|^2  \leq 2\Expect\fnorm{V-\tV}^2 \Expect\|R\|^2 + 2\Expect\fnorm{R-\tR}^2\leq 8c_1 nrD + r(r+1)D \leq c_2 nrD	
	\label{eq:b3}
\end{equation}
for $c_2=8c_1+1$.
In view of \prettyref{eq:IWW} and \prettyref{eq:Ivv}--\prettyref{eq:IRR},
we have
\[
I(V;\hV) \geq nr \log \frac{1}{c_2 D} - \frac{r^2}{2} \log  \frac{c_0}{D} -\frac{r^2}{2}  \log \frac{1}{D} \geq (n-r)r \log \frac{c}{ D}
\]
for some universal constant $c$.
\end{proof}

\begin{proof}[Proof of \prettyref{thm:RDT}]
	The proof follows the same course of proving \prettyref{thm:RDT} (with $n=r$) and is even simpler, since the loss function is the usual Frobenius norm and hence there is no need to introduce the rotation matrix $A$. Recall that $W$ is a $r\times r$ complex Gaussian matrix and $W=TR$ be its QR  decomposition. Given any feasible $P_{\tT|T}$ in \prettyref{eq:RDT}, let $\tW=\tT\tR$, where $\tR$ fulfills \prettyref{eq:IRR} and \prettyref{eq:ERR}. Then similar to \prettyref{eq:b1}-\prettyref{eq:b3}, we have $\Expect \fnorm{W-\hW}^2 \leq 2 \Expect \fnorm{R-\hR}^2 + 2 \Expect \fnorm{T-\hT}^2 \Expect \|R\|^2 \leq C r^2 D$ for some absolute constant $C$. Similar to \prettyref{eq:Ivv}, we have $I(W;\tW) \leq I(T;\tT) + I(R;\tR)$. Then the desired lower bound \prettyref{eq:RDT} on $I(T;\tT)$ follows from \prettyref{eq:IRR} and \prettyref{eq:IWW} with $c'=1/C$.
\end{proof}

Now we are ready to prove the main lower bound for robust reconstruction of rank-$r$ matrices:
\begin{proof}[Proof of \prettyref{thm:lb}]
As in \prettyref{fig:bayesian}, let $\hX = \hU\hV^*=g(f(UV^*)+Z)$ be the reconstruction, where $U,\hU\in \sfV(\complex^{n_1},r)$ and $V,\hV\in \sfV(\complex^{n_2},r)$.
Note that
\begin{align}
\fnorm{UV^* - \hU\hV^*}^2	
= & ~ 2 r^2 - 2 \mathrm{Re} \Iprod{\hU^*U}{\hV^* V}	\nonumber \\
\geq & ~ 2 r^2 - 2 \fnorm{\hU^*U} \fnorm{\hV^* V} \geq 2r^2 - 2 r \fnorm{\hU^* U}, \label{eq:XhX}
\end{align}
where we used Cauchy-Schwarz inequality and $\fnorm{\hV^* V}^2 = \Iprod{\hV\hV^*}{VV^*} \leq \fnorm{\hV\hV^*}\fnorm{VV^*} =r$.
On the other hand,
 $\fnorm{UU^* - \hU\hU^*}^2 = 2r^2 -2\fnorm{\hat U^* U}^2$. Substituting into \prettyref{eq:XhX} and using the fact that $\sqrt{1-x} \leq 1-x/2$ for all $0\leq x \leq 1$, we obtain
 $\fnorm{UU^* - \hU\hU^*}^2 \leq 2 \fnorm{UV^* - \hU\hV^*}^2	$.
 By the assumption that $\Expect \fnorm{UV^*-\hU\hV^*}^2 \leq r D$, we have
	\begin{equation}
	\Expect \fnorm{UU^*-\hU\hU^*}^2 \leq 2 r D, \quad ,\Expect \fnorm{VV^*-\hV\hV^*}^2 \leq 2 r D,
	\label{eq:uuvv}
\end{equation}
where the second inequality follows analogously.
Note that $UU^*$ is in one-to-one correspondence to the column span of $X$.
Furthermore, given $UU^*$, let $\tilde U \in \sfV(\complex^n,r)$ be an arbitrary basis of the column span, so that $\tU\tU^* = UU^*$. In other words, $\tU$ is a deterministic function of $UU^*$.\footnote{To be definitive, one can obtain $\tU=[\tu_1,\ldots,\tu_r]$ recursively as follows: let $\tu_1$ be the normalized first column of $UU^*$, and for $i=2,\ldots,r$ then let $\tu_i$ be the normalized first nonzero column of $P_{E_i^\perp} UU^* P_{E_i^\perp}$, where $E_i=\text{span}(u_1,\ldots,u_{i-1})$.}
Similarly, define $\tV$ as a function of $VV^*$.

	In view of the Markov chain $X \to Y \to \hY \to \hX$, we obtain the following inequalities in parallel to the joint-source-channel-coding converse in Shannon theory:
	\begin{align}
&~m \log \pth{1+ 1/\sigma^2} \nonumber \\
\geq & ~ I(Y; Y+Z)	\label{eq:ms1}\\
\geq & ~ I(X; \hX)
\label{eq:ms2}	\\
= & ~ I(X,UU^*,VV^*; \hX) = I(UU^*,VV^*; \hX) + I(X;\hX|UU^*,VV^*) \label{eq:ms3}	\\
= & ~ I(UU^*,VV^*; \hX) + I(T;\tT|UU^*,VV^*) \label{eq:ms4}\\
\geq & ~ I(UU^*; \hX) + I(VV^*;\hX) + I(T;\hT) \label{eq:ms5} \\
\geq & ~ I(UU^*; \hU\hU^*) + I(VV^*; \hV\hV^*) + I(T;\hT) \label{eq:ms6}	\\
\geq & ~ \inf_{\Expect \fnorm{UU^* - \hU\hU^*}^2 \leq  2rD} I(UU^*; \hU\hU^*|\supp{U}) + \inf_{\Expect \fnorm{VV^* - \hV\hV^*}^2 \leq 2rD} I(VV^*; \hV\hV^*|\supp{V}) \nonumber \\
 &~+ \inf_{\Expect \fnorm{T - \hT}^2 \leq  rD} I(T; \hT)   \label{eq:ms7}\\
\geq & ~ ((s_1-r)r+(s_2-r)r) \log \frac{c}{2D} + \frac{r^2}{2} \log \frac{c'}{D},  \label{eq:ms8}
\end{align}
	where
	\begin{itemize}
	\item \prettyref{eq:ms1}: by the complex Gaussian channel capacity formula with the average power constraint $\Expect \|Y\|_2^2 \leq m$, which is guaranteed by \prettyref{eq:avg-power};
	\item \prettyref{eq:ms2}: by the data processing inequality for mutual information;
	\item \prettyref{eq:ms3}: by the fact the pair $(UU^*,VV^*)$ is a deterministic function of $X=UV^*$;
	\item \prettyref{eq:ms4}: we defined  $T \triangleq \tU^* X \tV, \hT \triangleq \tU^* \hX \tV$;
	\item \prettyref{eq:ms5}: by the mutual independence of $T$, $UU^*$ and $VV^*$ and the property of mutual information $I(A,B;C) \geq I(A;C)+I(B;C)$ whenever $A$ and $B$ are independent;
	\item \prettyref{eq:ms6}: by the fact the pair $(\hU\hU^*,\hV\hV^*)$ is a deterministic function of $\hX$;
	\item	\prettyref{eq:ms7}: by \prettyref{eq:uuvv} and the fact that $\fnorm{T-\hT} = \fnorm{\tU^*(X-\hX)\tV} \leq \fnorm{X-\hX}$, where the infima are over $P_{\hU\hU^*|UU^*,\supp{U}}$,  $P_{\hV\hV^*|VV^*,\supp{V}}$  and $P_{\hT|T}$ respectively;
	\item	\prettyref{eq:ms8}: conditioned on $\supp{U}$ (resp.~$\supp{V}$), the non-zeros of $U$ (resp.~$V$) is uniform on the Stifled manifold of dimension $r$ in $\complex^{s_1}$ (resp.~$\complex^{s_2}$). Furthermore,
	$T$ is uniform over $O(r)$. Applying \prettyref{thm:RDv} and \prettyref{thm:RDT} yields the desired lower bound. \qedhere
\end{itemize}
\end{proof}

\section*{Acknowledgements}
K. Lee and Y. Bresler are supported in part by the National Science Foundation under Grants CCF 10-18789, CCF 10-18660, and IIS 14-47879. K. Lee would like to thank Marius Junge and Angelia Nedi\'c for discussions related to this paper. The authors thank J.A. Geppert and F. Krahmer for their comments on an earlier version \cite{LeeWB2013}.

\appendices

\section{Notation}

We will use the following notation in the appendix:
For $A_1 \in \bbC^{n_1 \times n_1}$ and $A_2 \in \bbC^{n_2 \times n_2}$,
$A_1 \otimes A_2: \bbC^{n_1 \times n_2} \to \bbC^{n_1 \times n_2}$ denotes the linear operator defined by $(A_1 \otimes A_2) X = A_1 X A_2^*$ for all $X \in \bbC^{n_1 \times n_2}$.

\section{RIP Lemmas}
In this section, we present a set of lemmas as consequences of the rank-2 and doubly $(3s_1,3s_2)$-sparse RIP of $\A$ with isometry constant $\delta$.

\begin{lemma}
\label{lemma:ds_RIP2FG}
Let $F: \bbC^{n_2} \to \bbC^{m \times n_1}$ and $G: \bbC^{n_1} \to \bbC^{m \times n_2}$ be defined from $\A$ by (\ref{eq:defFnG}).
Then, $F$ and $G$ satisfy
\begin{equation}
\norm{\Pi_{\widetilde{J}_1} ([F(y)]^* [F(\zeta)] - \langle \zeta, y \rangle \idvr) \Pi_{\widetilde{J}_1}} \leq \delta \norm{y}_2 \norm{\zeta}_2,
\quad \forall \widetilde{J}_1 \subset [n_1],~ |\widetilde{J}_1| \leq 3s_1
\label{eq:lemma:ds_RIP2FG:resF}
\end{equation}
for all $y, \zeta \in \bbC^{n_2}$ with $\norm{y}_0 + \norm{\zeta}_0 \leq 3s_2$, and
\begin{equation}
\norm{\Pi_{\widetilde{J}_2} ([G(x)]^* [G(\xi)] - \langle \xi, x \rangle \idvc) \Pi_{\widetilde{J}_2}} \leq \delta \norm{x}_2 \norm{\xi}_2,
\quad \forall \widetilde{J}_2 \subset [n_2],~ |\widetilde{J}_2| \leq 3s_2
\label{eq:lemma:ds_RIP2FG:resG}
\end{equation}
for all $x, \xi \in \bbC^n$ with $\norm{x}_0 + \norm{\xi}_0 \leq 3s_1$, respectively.
\end{lemma}

\begin{proof}
Since (\ref{eq:lemma:ds_RIP2FG:resF}) is homogeneous in $y$ and $\zeta$, without loss of generality, we may assume that $\norm{y}_2 = \norm{\zeta}_2 = 1$.

Let $u,\tilde{u} \in \bbC^n$ satisfy $\norm{u}_2 = \norm{\tilde{u}}_2 = 1$. Let $v,\tilde{v} \in \bbC^d$ satisfy $\norm{v}_2 = \norm{\tilde{v}}_2 = 1$. Then,
\begin{align*}
{} & \langle uv^*, (\Pi_{\widetilde{J}_1} \otimes P_{\R(y)}) (\A^* \A) (\Pi_{\widetilde{J}_1} \otimes P_{\R(\zeta)}) (\tilde{u} \tilde{v}^*) \rangle \\
{} & = \langle uv^*, (\Pi_{\widetilde{J}_1} \otimes y y^*) (\A^* \A) (\Pi_{\widetilde{J}_1} \otimes \zeta \zeta^*) (\tilde{u} \tilde{v}^*) \rangle \\
{} & = \langle (\Pi_{\widetilde{J}_1} \otimes y y^*) uv^*, (\A^* \A) (\Pi_{\widetilde{J}_1} \otimes \zeta \zeta^*) (\tilde{u} \tilde{v}^*) \rangle \\
{} & = \langle \A (\Pi_{\widetilde{J}_1} \otimes y y^*) (uv^*), \A (\Pi_{\widetilde{J}_1} \otimes \zeta \zeta^*) (\tilde{u} \tilde{v}^*) \rangle \\
{} & = \langle \A (\Pi_{\widetilde{J}_1} u v^* y y^*), \A (\Pi_{\widetilde{J}_1} \tilde{u} \tilde{v}^* \zeta \zeta^*) \rangle \\
{} & = (y^* v) (\tilde{v}^* \zeta) \langle [F(y)] \Pi_{\widetilde{J}_1} u, [F(\zeta)] \Pi_{\widetilde{J}_1} \tilde{u} \rangle \\
{} & = (y^* v) (\tilde{v}^* \zeta) \langle u, \Pi_{\widetilde{J}_1} [F(y)]^* [F(\zeta)] \Pi_{\widetilde{J}_1} \tilde{u} \rangle.
\end{align*}
Similarly,
\begin{align*}
{} & \langle uv^*, (\Pi_{\widetilde{J}_1} \otimes P_{\R(y)}) (\Pi_{\widetilde{J}_1} \otimes P_{\R(\zeta)}) (\tilde{u} \tilde{v}^*) \rangle \\
{} & = (y^* v) (\tilde{v}^* \zeta) \langle \zeta, y \rangle \langle u, \Pi_{\widetilde{J}_1} \tilde{u} \rangle.
\end{align*}
Therefore,
\begin{equation}
\label{eq:cor:RIP2F:identity1}
\begin{aligned}
{} & \langle uv^*, (\Pi_{\widetilde{J}_1} \otimes P_{\R(y)}) (\A^* \A - \idm) (\Pi_{\widetilde{J}_1} \otimes P_{\R(\zeta)}) (\tilde{u} \tilde{v}^*) \rangle \\
{} & = (y^* v) (\tilde{v}^* \zeta) \langle u, \Pi_{\widetilde{J}_1} ([F(y)]^* [F(\zeta)] - \langle \zeta, y \rangle \idvr) \Pi_{\widetilde{J}_1} \tilde{u} \rangle.
\end{aligned}
\end{equation}

In fact, the operator norm of $(\Pi_{\widetilde{J}_1} \otimes P_{\R(y)}) (\A^* \A - \idm) (\Pi_{\widetilde{J}_1} \otimes P_{\R(\zeta)})$ is achieved by maximizing the left hand side of (\ref{eq:cor:RIP2F:identity1}) over $u$ and $\tilde{u}$ supported on $\widetilde{J}_1$ for $v = y$ and $\tilde{v} = \zeta$.

Therefore,
\begin{align*}
{} & \norm{\Pi_{\widetilde{J}_1} ([F(y)]^* [F(\zeta)] - \langle \zeta, y \rangle \idvr) \Pi_{\widetilde{J}_1}} \\
{} & = \norm{(\Pi_{\widetilde{J}_1} \otimes P_{\R(y)}) (\A^* \A - \idm) (\Pi_{\widetilde{J}_1} \otimes P_{\R(\zeta)})} \\
{} & = \norm{(\Pi_{\widetilde{J}_1} \otimes P_{\R(y)+\R(\zeta)}) (\A^* \A - \idm) (\Pi_{\widetilde{J}_1} \otimes P_{\R(y)+\R(\zeta)})} \\
{} & \leq \delta
\end{align*}
where the last step follows since $(\A^* \A - \idm)$ is restricted to a set of rank-2 and doubly $(3s_1,3s_2)$-sparse matrices.

The inequality in (\ref{eq:lemma:ds_RIP2FG:resF}) is derived similarly using the following identity
\begin{equation}
\label{eq:cor:RIP2F:identity2}
\begin{aligned}
{} & \langle uv^*, (P_{\R(x)} \otimes \Pi_{\widetilde{J}_2}) (\A^* \A - \idm) (P_{\R(\xi)} \otimes \Pi_{\widetilde{J}_2}) (\tilde{u} \tilde{v}^*) \rangle \\
{} & = (u^* x) (\xi^* \tilde{u}) \langle \tilde{v}, \Pi_{\widetilde{J}_2} ([G(\xi)]^* [G(x)] - \langle x, \xi \rangle \idvc) \Pi_{\widetilde{J}_2} v \rangle.
\end{aligned}
\end{equation}
\end{proof}

\begin{lemma}
\label{lemma:ds_bndopnormrip}
Let $\widetilde{J}_1 \subset [n_1]$ and $\widetilde{J}_2 \subset [n_2]$ satisfy $|\widetilde{J}_1| \leq 2s_1$ and $|\widetilde{J}_2| \leq 2s_2$, respectively.
Suppose that $X \in \bbC^{n_1 \times n_2}$ is a doubly $(s_1,s_2)$-sparse rank-one matrix.
Then,
\begin{equation}
\norm{\Pi_{\widetilde{J}_1} [(\A^* \A - \idm) (X)] \Pi_{\widetilde{J}_2}} \leq \delta \fnorm{X}.
\label{eq:lemma:ds_bndopnormrip:res}
\end{equation}
\end{lemma}

\begin{proof}
Suppose that $\xi \in \bbS^{n_1-1}$ and $\zeta \in \bbS^{n_2-1}$ are supported on $\widehat{J}_1$ and $\widehat{J}_2$, respectively.
Since (\ref{eq:lemma:ds_bndopnormrip:res}) is homogeneous, without loss of generality, we may assume that $\fnorm{X} = 1$.
Let $X = \lambda u v^*$ denote the singular value decomposition of $X$.
Then, $X = P_{\R(u)} X P_{\R(v)} = (P_{\R(u)} \otimes P_{\R(v)}) (X)$; hence,
\begin{align*}
|\langle \xi \zeta^*, \Pi_{\widetilde{J}_1} [(\A^* \A - \idm) (X)] \Pi_{\widetilde{J}_2} \rangle|
{} & = |\langle P_{\R(\xi)} \xi \zeta^* P_{\R(\zeta)}, (\A^* \A - \idm) (X) \rangle| \\
{} & = |\langle \xi \zeta^*, (P_{\R(\xi)} \otimes P_{\R(\zeta)}) (\A^* \A - \idm) (P_{\R(u)} \otimes P_{\R(v)}) X \rangle| \\
{} & \leq \norm{(P_{\R(\xi)} \otimes P_{\R(\zeta)}) (\A^* \A - \idm) (P_{\R(u)} \otimes P_{\R(X^*)})} \\
{} & \leq \norm{(P_{\R(\xi) + \R(u)} \otimes P_{\R(\zeta) + \R(v)}) (\A^* \A - \idm) (P_{\R(\xi) + \R(u)} \otimes P_{\R(\zeta) + \R(v)})} \\
{} & \leq \delta.
\end{align*}
Maximizing this inequality over $\xi$ and $\zeta$ yields the desired claim.
\end{proof}

\begin{lemma}
\label{lemma:ds_bndnoise}
Let $\widetilde{J}_1 \subset [n_1]$ and $\widetilde{J}_2 \subset [n_2]$ satisfy $|\widetilde{J}_1| \leq 3s_1$ and $|\widetilde{J}_2| \leq 3s_2$, respectively.
Then,
\begin{equation*}
\norm{\Pi_{\widetilde{J}_1} [\A^* (z)] \Pi_{\widetilde{J}_2}} \leq \sqrt{1+\delta} \norm{z}_2, \quad \forall z \in \bbC^m.
\end{equation*}
\end{lemma}

\begin{proof}
Let $\xi \in \bbS^{n_1-1}$ and $\zeta \in \bbS^{n_2-1}$.
\begin{align*}
|\langle \xi \zeta^*, \Pi_{\widetilde{J}_1} [\A^* (z)] \Pi_{\widetilde{J}_2} \rangle|
{} & = |\langle \A(\Pi_{\widetilde{J}_1} \xi \zeta^* \Pi_{\widetilde{J}_2}), z \rangle| \\
{} & \leq \norm{\A(\Pi_{\widetilde{J}_1} \xi \zeta^* \Pi_{\widetilde{J}_2})}_2 \norm{z}_2 \\
{} & \leq \sqrt{1+\delta} \fnorm{\Pi_{\widetilde{J}_1} \xi \zeta^* \Pi_{\widetilde{J}_2}} \norm{z}_2 \\
{} & \leq \sqrt{1+\delta} \norm{z}_2.
\end{align*}
Maximizing this inequality over $\xi$ and $\zeta$ yields the desired claim.
\end{proof}

The following lemma is a consequence of the rank-$2r$ and doubly $(3s_1,3s_2)$-sparse RIP of $\A$ with isometry constant $\delta$.
\begin{lemma}
\label{lemma:reduced_rip}
Suppose that $\calA: \cz^{n_1 \times n_2} \to \cz^m$ satisfies the rank-$2r$ and doubly $(s_1,s_2)$-sparse RIP with isometry constant $\delta$.
Let $\calF$ and $\calG$ be defined from $\calA$, respectively, by \eqref{eq:def_calF} and \eqref{eq:def_calG}.
Let $\calH_1$ (resp. $\calH_2$) denote the Hilbert space of $n_1$-by-$r$ (resp. $n_2$-by-$r$) matrices with the Frobenius norm.
Fix $r$ be an arbitrary positive integer satisfying $r \leq \min(n_1,n_2)$.
Then, the followings holds:
\begin{enumerate}
  \item For all $V \in \cz^{n_2 \times r}$ such that $V^* V = I_r$ and $V$ is row $s_2$-sparse, the linear operator $[\calF(V)]^* [\calF(V)]$ on $\calH_1$ satisfies
  \[
  \max_{J_1 \subset [n_1] : |J_1| \leq s_1} \norm{(\Pi_{J_1} \otimes I_r) ([\calF(V)]^*[\calF(V)] - \id_{\cz^{n_1 \times r}}) (\Pi_{J_1} \otimes I_r)}_{\calH_1 \to \calH_1} \leq \delta.
  \]
  \item For all $V, \widetilde{V} \in \cz^{n_2 \times r}$ such that $\widetilde{V}^* V = 0$ and $[V, \widetilde{V}]$ is row $s_2$-sparse, the linear operator $[\calF(\widetilde{V})]^* [\calF(V)]$ on $\calH_1$ satisfies
  \[
  \max_{J_1 \subset [n_1] : |J_1| \leq s_1} \norm{(\Pi_{J_1} \otimes I_r) [\calF(\widetilde{V})]^*[\calF(V)] (\Pi_{J_1} \otimes I_r)}_{\calH_1 \to \calH_1} \leq \delta \norm{V} \norm{\widetilde{V}}.
  \]
  \item For all $U \in \cz^{n_1 \times r}$ such that $U^* U = I_r$ and $U$ is row $s_1$-sparse, the linear operator $[\calG(U)]^* [\calG(U)]$ on $\calH_2$ satisfies
  \[
  \max_{J_2 \subset [n_2] : |J_2| \leq s_2} \norm{(\Pi_{J_2} \otimes I_r) ([\calG(U)]^*[\calG(U)] - \id_{\cz^{n_2 \times r}}) (\Pi_{J_1} \otimes I_r)}_{\calH_2 \to \calH_2} \leq \delta.
  \]
  \item For all $U, \widetilde{U} \in \cz^{n_1 \times r}$ such that $\widetilde{U}^* U = 0$ and $[U, \widetilde{U}]$ is row $s_1$-sparse, the linear operator $[\calG(\widetilde{U})]^* [\calG(U)]$ on $\calH_2$ satisfies
  \[
  \max_{J_2 \subset [n_2] : |J_2| \leq s_2} \norm{(\Pi_{J_2} \otimes I_r) [\calG(\widetilde{U})]^*[\calG(U)] (\Pi_{J_2} \otimes I_r)}_{\calH_2 \to \calH_2} \leq \delta \norm{U} \norm{\widetilde{U}}.
  \]
\end{enumerate}
\end{lemma}

\begin{proof}[Proof of Lemma~\ref{lemma:reduced_rip}]
We only prove the first two results since the third and fourth results are derived by symmetry.

Fix an arbitrary $V \in \cz^{n_2 \times r}$ so that $V^* V = I_r$ and $V$ is row $s_2$-sparse.
Fix an arbitrary $J_1 \subset [n_1]$ so that $|J_1| \leq s_1$.
Fix arbitrary $U, \widetilde{U} \in \cz^{n_1 \times r}$. Then, we have
\begin{align*}
{} & \langle \widetilde{U}, (\Pi_{J_1} \otimes I_r) ([\calF(V)]^*[\calF(V)] - \id_{\cz^{n_1 \times r}}) (\Pi_{J_1} \otimes I_r) U \rangle \\
{} & = \langle (\Pi_{J_1} \otimes I_r) \widetilde{U}, ([\calF(V)]^*[\calF(V)] - \id_{\cz^{n_1 \times r}}) (\Pi_{J_1} \otimes I_r) U \rangle \\
{} & = \langle \Pi_{J_1} \widetilde{U}, ([\calF(V)]^*[\calF(V)] - \id_{\cz^{n_1 \times r}}) (\Pi_{J_1} U) \rangle \\
{} & = \langle \Pi_{J_1} \widetilde{U}, [\calF(V)]^*[\calF(V)] (\Pi_{J_1} U) \rangle - \langle \Pi_{J_1} \widetilde{U}, \Pi_{J_1} U \rangle \\
{} & = \langle [\calF(V)] (\Pi_{J_1} \widetilde{U}), [\calF(V)] (\Pi_{J_1} U) \rangle - \langle \Pi_{J_1} \widetilde{U}, \Pi_{J_1} U V^* V \rangle \\
{} & = \langle \calA(\Pi_{J_1} \widetilde{U} V^*), \calA(\Pi_{J_1} U V^*) \rangle - \langle \Pi_{J_1} \widetilde{U} V^*, \Pi_{J_1} U V^* \rangle \\
{} & = \langle \Pi_{J_1} \widetilde{U} V^*, \calA^* \calA(\Pi_{J_1} U V^*) \rangle - \langle \Pi_{J_1} \widetilde{U} V^*, \Pi_{J_1} U V^* \rangle \\
{} & = \langle \Pi_{J_1} \widetilde{U} V^*, (\calA^* \calA - \id) (\Pi_{J_1} U V^*) \rangle,
\end{align*}
where the fourth step holds since $V^* V = I_r$ and the fifth step holds by the definition of $\calF$.
Therefore, since $[\Pi_{J_1} U, \Pi_{J_1} \widetilde{U}] = \Pi_{J_1} [U, \widetilde{U}]$ is row $s_1$-sparse and $V$ is row $s_2$-sparse,
by the rank-$2r$ and doubly $(s_1,s_2)$-sparse RIP of $\calA$, it follows that
\begin{align*}
{} & |\langle \widetilde{U}, (\Pi_{J_1} \otimes I_r) ([\calF(V)]^*[\calF(V)] - \id_{\cz^{n_1 \times r}}) (\Pi_{J_1} \otimes I_r) U \rangle| \\
{} & \leq \delta \fnorm{\Pi_{J_1} U V^*} \fnorm{\Pi_{J_1} \widetilde{U} V^*} \\
{} & = \delta \fnorm{\Pi_{J_1} U} \fnorm{\Pi_{J_1} \widetilde{U}} \\
{} & \leq \delta \fnorm{U} \fnorm{\widetilde{U}}.
\end{align*}
By maximizing over $U$ and $\widetilde{U}$ within the unit ball in $\calH_1$, we have
\[
\norm{(\Pi_{J_1} \otimes I_r) ([\calF(V)]^*[\calF(V)] - \id_{\cz^{n_1 \times r}}) (\Pi_{J_1} \otimes I_r)}_{\calH_1 \to \calH_1} \leq \delta.
\]
By maximizing over $J_1$, we get the first result.

The second result is proved in a similar way.
Fix arbitrary $V, \widetilde{V} \in \cz^{n_2 \times r}$ so that $\langle \widetilde{V}, V \rangle = 0$ and $[V, \widetilde{V}]$ is row $s_2$-sparse.
Fix an arbitrary $J_1 \subset [n_1]$ so that $|J_1| \leq s_1$.
Fix arbitrary $U, \widetilde{U} \in \cz^{n_1 \times r}$.
Then, similar to the previous case, we have
\begin{align*}
{} & \langle \widetilde{U}, (\Pi_{J_1} \otimes I_r) [\calF(\widetilde{V})]^*[\calF(V)] (\Pi_{J_1} \otimes I_r) U \rangle \\
{} & = \langle \Pi_{J_1} \widetilde{U} \widetilde{V}^*, \calA^* \calA (\Pi_{J_1} U V^*) \rangle, \\
{} & = \langle \Pi_{J_1} \widetilde{U} \widetilde{V}^*, (\calA^* \calA - \id) (\Pi_{J_1} U V^*) \rangle,
\end{align*}
where the last step holds since $\langle \Pi_{J_1} \widetilde{U} \widetilde{V}^*, \Pi_{J_1} U V^* \rangle = 0$,
which follows from $\widetilde{V}^* V = 0$.
Therefore, since $[\Pi_{J_1} U, \Pi_{J_1} \widetilde{U}] = \Pi_{J_1} [U, \widetilde{U}]$ is row $s_1$-sparse and $[V, \widetilde{V}]$ is row $s_2$-sparse, by the rank-$2r$ and doubly $(s_1,s_2)$-sparse RIP of $\calA$, it follows that
\begin{align*}
{} & |\langle \widetilde{U}, (\Pi_{J_1} \otimes I_r) [\calF(\widetilde{V})]^*[\calF(V)] (\Pi_{J_1} \otimes I_r) U \rangle| \\
{} & \leq \delta \fnorm{\Pi_{J_1} U V^*} \fnorm{\Pi_{J_1} \widetilde{U} \widetilde{V}^*} \\
{} & = \delta \fnorm{\Pi_{J_1} U} \norm{V} \fnorm{\Pi_{J_1} \widetilde{U}} \norm{\widetilde{V}} \\
{} & \leq \delta \fnorm{U} \fnorm{\widetilde{U}} \norm{V} \norm{\widetilde{V}}.
\end{align*}
By maximizing over $U$ and $\widetilde{U}$ within the unit ball in $\calH_1$, we have
\[
\norm{(\Pi_{J_1} \otimes I_r) [\calF(\widetilde{V})]^*[\calF(V)] (\Pi_{J_1} \otimes I_r)}_{\calH_1 \to \calH_1} \leq \delta \norm{V} \norm{\widetilde{V}}.
\]
By maximizing over $J_1$, we get the second result.
\end{proof}

\section{Proof of Lemma~\ref{lemma:init_w_good_est_supp2}}
\label{sec:proof:lemma:init_w_good_est_supp2}

To prove Lemma~\ref{lemma:init_w_good_est_supp2}, we use the non-Hermitian $\sin\theta$ theorem \cite[pp. 102--103]{Wed1972perturbation}.

\begin{lemma}[{Non-Hermitian $\sin\theta$ theorem \cite{Wed1972perturbation}}]
\label{lemma:sintheta}
Let $M \in \bbC^{n \times d}$ be a rank-$r$ matrix.
Let $\widehat{M} \in \bbC^{n \times d}$ be the best rank-$r$ approximation of $M + \Delta$ in the Frobenius norm.
Then,
\[
\max\left(
\sin \theta(\R(M),\R(\widehat{M})), \sin \theta(\R(M^*),\R(\widehat{M}^*))
\right)
\leq \frac{\max(\norm{P_{\R(M)} \Delta},\norm{\Delta P_{\R(M^*)}})}{\sigma_r(M) - \norm{M + \Delta - \widehat{M}}}.
\]
\end{lemma}

\begin{proof}[Proof of Lemma~\ref{lemma:init_w_good_est_supp2}]
Let $M \triangleq \Pi_{\widehat{J}_1} X$ where $X = \lambda u v^*$ with $s_1$-sparse $u \in \bbS^{n_1-1}$ and $s_2$-sparse $v \in \bbS^{n_2-1}$.
Let $\widehat{M}$ denote the best rank-one approximation of $\Pi_{\widehat{J}_1} [\A^*(b)] \Pi_{\widehat{J}_2}$ in the spectral norm.
Let $\Delta \triangleq \Pi_{\widehat{J}_1} [\A^*(b)] \Pi_{\widehat{J}_2} - M$.
Then,
\begin{align*}
\norm{M + \Delta - \widehat{M}}
{} & = \norm{\Pi_{\widehat{J}_1} [\A^*(b)] \Pi_{\widehat{J}_2} - \widehat{M}} \\
{} & \leq \norm{ \Pi_{\widehat{J}_1} [\A^*(b)] \Pi_{\widehat{J}_2} - \Pi_{\widehat{J}_1} X \Pi_{\widehat{J}_2}} \\
{} & = \norm{\Pi_{\widehat{J}_1} [(\A^* \A - \idm)(X)] \Pi_{\widehat{J}_2} + \Pi_{\widehat{J}_1} [\A^*(z)] \Pi_{\widehat{J}_2}} \\
{} & \leq \norm{\Pi_{\widehat{J}_1} [(\A^* \A - \idm)(X)] \Pi_{\widehat{J}_2}} + \norm{\Pi_{\widehat{J}_1} [\A^*(z)] \Pi_{\widehat{J}_2}} \\
{} & \leq \delta \fnorm{X} + \sqrt{1+\delta} \norm{z}_2
\end{align*}
where the second inequality follows from Lemma~\ref{lemma:ds_bndopnormrip} and Lemma~\ref{lemma:ds_bndnoise}.

Similarly, the difference $\Delta$ is bounded in the spectral norm by
\begin{align*}
\norm{\Delta}
{} & = \norm{\Pi_{\widehat{J}_1} [\A^*(b)] \Pi_{\widehat{J}_2} - \Pi_{\widehat{J}_1} X \Pi_{\widehat{J}_2} - \Pi_{\widehat{J}_1} X \Pi_{\widehat{J}_2}^\perp} \\
{} & = \norm{\Pi_{\widehat{J}_1} [(\A^* \A - \idm)(X)] \Pi_{\widehat{J}_2} - \Pi_{\widehat{J}_1} X \Pi_{\widehat{J}_2}^\perp + \Pi_{\widehat{J}_1} [\A^*(z)] \Pi_{\widehat{J}_2}} \\
{} & \leq \norm{\Pi_{\widehat{J}_1} [(\A^* \A - \idm)(X)] \Pi_{\widehat{J}_2}} + \norm{\Pi_{\widehat{J}_1} X \Pi_{\widehat{J}_2}^\perp} + \norm{\Pi_{\widehat{J}_1} [\A^*(z)] \Pi_{\widehat{J}_2}} \\
{} & \leq \delta \fnorm{X} + \fnorm{X} \norm{\Pi_{\widehat{J}_1} u}_2 \norm{\Pi_{\widehat{J}_2}^\perp v}_2 + \sqrt{1+\delta} \norm{z}_2
\end{align*}
where the last step follows from Lemma~\ref{lemma:ds_bndopnormrip} and Lemma~\ref{lemma:ds_bndnoise}.

Note that $\norm{M}$ is rewritten as
\begin{align*}
\norm{M} = \norm{\Pi_{\widehat{J}_1} X}
= \norm{\Pi_{\widehat{J}_1} (\lambda u v^*)}
= \lambda \norm{\Pi_{\widehat{J}_1} u}_2 \
= \fnorm{X} \norm{\Pi_{\widehat{J}_1} u}_2.
\end{align*}

Since $v_0$ is the right singular vector of $\widehat{M}$ and $\R(M^*) = \R(X^*)$, by Lemma~\ref{lemma:sintheta}, we have
\begin{align*}
\sin (\R(X^*),\R(v_0))
{} & \leq \frac{\norm{\Delta}}{\norm{M} - \norm{M + \Delta - \widehat{M}}} \\
{} & \leq \frac{\fnorm{X} \norm{\Pi_{\widehat{J}_1} u}_2 \norm{\Pi_{\widehat{J}_2}^\perp v}_2 + \delta \fnorm{X} + \sqrt{1+\delta} \norm{z}_2}{\fnorm{X} \norm{\Pi_{\widehat{J}_1} u}_2 - \delta \fnorm{X} - \sqrt{1+\delta} \norm{z}_2}.
\end{align*}
Applying (\ref{eq:nuineq}) to this result, we get the assertion in Lemma~\ref{lemma:init_w_good_est_supp2}.
\end{proof}

\section{Proof of Lemma~\ref{lemma:est_supp_dspca}}
\label{sec:proof:lemma:est_supp_dspca}

Let $J_1$ and $J_2$ denote the support of $u$ and $v$, respectively.
By the definition of $(\widehat{J}_1,\widehat{J}_2)$ in (\ref{eq:initbydspca}), we have
\begin{equation}
\norm{\Pi_{\widehat{J}_1} [\A^*(b)] \Pi_{\widehat{J}_2}} \geq \norm{\Pi_{J_1} [\A^*(b)] \Pi_{J_2}}.
\label{eq:proof:lemma:est_supp_dspca:ineq}
\end{equation}

First, the right-hand-side of (\ref{eq:proof:lemma:est_supp_dspca:ineq}) is bounded from below by
\begin{equation}
\begin{aligned}
\norm{\Pi_{J_1} [\A^*(b)] \Pi_{J_2}}
{} & = \norm{\Pi_{J_1} [X + (\A^* \A - \idm) (X) + \A^*(z)] \Pi_{J_2}} \\
{} & \geq \norm{X} - \norm{\Pi_{J_1} [(\A^* \A - \idm) (X)] \Pi_{J_2}} - \norm{\Pi_{J_1} [\A^*(z)] \Pi_{J_2}} \\
{} & \geq (1-\delta) \lambda - \sqrt{1+\delta} \norm{z}_2
\end{aligned}
\label{eq:proof:lemma:est_supp_dspca:rhs}
\end{equation}
where the first inequality follows from Lemma~\ref{lemma:ds_bndopnormrip} and Lemma~\ref{lemma:ds_bndnoise}.

Next, the left-hand side of (\ref{eq:proof:lemma:est_supp_dspca:ineq}) is bounded from above by
\begin{align*}
\norm{\Pi_{\widehat{J}_1} [\A^*(b)] \Pi_{\widehat{J}_2}}
{} & \leq \underbrace{ \norm{\Pi_{\widehat{J}_1} [(\A^* \A) (\Pi_{\widehat{J}_1} X \Pi_{\widehat{J}_2})] \Pi_{\widehat{J}_2}} }_{ \text{(a)} } \\
{} & + \underbrace{ \norm{\Pi_{\widehat{J}_1} [(\A^* \A) (X - \Pi_{\widehat{J}_1} X \Pi_{\widehat{J}_2})] \Pi_{\widehat{J}_2}} }_{ \text{(b)} }
+ \underbrace{ \norm{\Pi_{\widehat{J}_1} [\A^*(z)] \Pi_{\widehat{J}_2}} }_{ \text{(c)} }.
\end{align*}
where (a) and (b) are further bounded using Lemma~\ref{lemma:ds_bndopnormrip} by
\begin{align*}
{} & \norm{\Pi_{\widehat{J}_1} [(\A^* \A) (\Pi_{\widehat{J}_1} X \Pi_{\widehat{J}_2})] \Pi_{\widehat{J}_2}} \\
{} & = \norm{(\Pi_{\widehat{J}_1} \otimes \Pi_{\widehat{J}_2}) (\A^* \A) (\Pi_{\widehat{J}_1} \otimes \Pi_{\widehat{J}_2}) (X)} \\
{} & \leq (1+\delta) \fnorm{\Pi_{\widehat{J}_1} X \Pi_{\widehat{J}_2}}
\end{align*}
and
\begin{align*}
{} & \norm{\Pi_{\widehat{J}_1} [(\A^* \A) (X - \Pi_{\widehat{J}_1} X \Pi_{\widehat{J}_2})] \Pi_{\widehat{J}_2}} \\
{} & = \norm{(\Pi_{\widehat{J}_1} \otimes \Pi_{\widehat{J}_2}) (\A^* \A) (X - \Pi_{\widehat{J}_1} X \Pi_{\widehat{J}_2})} \\
{} & = \norm{(\Pi_{\widehat{J}_1} \otimes \Pi_{\widehat{J}_2}) (\A^* \A - \idm) (X - \Pi_{\widehat{J}_1} X \Pi_{\widehat{J}_2})} \\
{} & \leq \delta \fnorm{X - \Pi_{\widehat{J}_1} X \Pi_{\widehat{J}_2}} \\
{} & = \delta (\fnorm{X}^2 - \fnorm{\Pi_{\widehat{J}_1} X \Pi_{\widehat{J}_2}}^2)^{1/2}
\end{align*}
respectively, and the noise term (c) is bounded using Lemma~\ref{lemma:ds_bndnoise} by
\[
\norm{\Pi_{\widehat{J}_1} [\A^*(z)] \Pi_{\widehat{J}_2}} \leq \sqrt{1+\delta} \norm{z}_2.
\]
Combining all, the left-hand side of (\ref{eq:proof:lemma:est_supp_dspca:ineq}) is bounded from above by
\begin{equation}
\norm{\Pi_{\widehat{J}_1} [\A^*(b)] \Pi_{\widehat{J}_2}} \leq (1+\delta) \fnorm{\Pi_{\widehat{J}_1} X \Pi_{\widehat{J}_2}} + \delta (\fnorm{X}^2 - \fnorm{\Pi_{\widehat{J}_1} X \Pi_{\widehat{J}_2}}^2)^{1/2} + \sqrt{1+\delta} \norm{z}_2.
\label{eq:proof:lemma:est_supp_dspca:lhs}
\end{equation}

Then, (\ref{eq:proof:lemma:est_supp_dspca:ineq}) is implies
\begin{equation}
(1+\delta) \fnorm{\Pi_{\widehat{J}_1} X \Pi_{\widehat{J}_2}} + \delta (\fnorm{X}^2 - \fnorm{\Pi_{\widehat{J}_1} X \Pi_{\widehat{J}_2}}^2)^{1/2}
\geq (1-\delta) \lambda - 2 \sqrt{1+\delta} \norm{z}_2.
\label{eq:proof:lemma:est_supp_dspca:suff1}
\end{equation}

Applying (\ref{eq:nuineq}) to (\ref{eq:proof:lemma:est_supp_dspca:suff1}), we get another necessary condition for (\ref{eq:proof:lemma:est_supp_dspca:ineq}) given by
\begin{align*}
{} & \frac{1+\delta}{\sqrt{\delta^2 + (1+\delta)^2}} \cdot \frac{\fnorm{\Pi_{\widehat{J}_1} X \Pi_{\widehat{J}_2}}}{\fnorm{X}} \\
{} & + \frac{\delta}{\sqrt{\delta^2 + (1+\delta)^2}} \cdot \frac{(\fnorm{X}^2 - \fnorm{\Pi_{\widehat{J}_1} X \Pi_{\widehat{J}_2}}^2)^{1/2}}{\fnorm{X}} \\
{} & \geq \frac{(1-\delta) - 2 (1+\delta) \nu}{\sqrt{\delta^2 + (1+\delta)^2}}.
\end{align*}

Define $\alpha,\beta \in [0, \pi/2]$ by
\begin{align*}
\alpha \triangleq \sin^{-1} \left( \frac{\fnorm{\Pi_{\widehat{J}_1} X \Pi_{\widehat{J}_2}}}{\fnorm{X}} \right)
\quad \text{and} \quad
\beta \triangleq \sin^{-1} \left( \frac{\delta}{\sqrt{\delta^2 + (1+\delta)^2}} \right).
\end{align*}

Then, using a trigonometric identity, we get
\[
\sin \alpha \cos \beta + \cos \alpha \sin \beta = \sin(\alpha + \beta) \geq \frac{(1-\delta) - 2 (1+\delta) \nu}{\sqrt{\delta^2 + (1+\delta)^2}},
\]
which is rewritten as
\[
\alpha \geq
\sin^{-1} \left( \frac{(1-\delta) - 2 (1+\delta) \nu}{\sqrt{\delta^2 + (1+\delta)^2}} \right)
- \sin^{-1} \left( \frac{\delta}{\sqrt{\delta^2 + (1+\delta)^2}} \right).
\]

Therefore, (\ref{eq:proof:lemma:est_supp_dspca:ineq}) implies
\[
\frac{\fnorm{\Pi_{\widehat{J}_1} X \Pi_{\widehat{J}_2}}}{\fnorm{X}} \geq
\sin \left[ \sin^{-1} \left( \frac{(1-\delta) - 2 (1+\delta) \nu}{\sqrt{\delta^2 + (1+\delta)^2}} \right)
- \sin^{-1} \left( \frac{\delta}{\sqrt{\delta^2 + (1+\delta)^2}} \right) \right].
\]

\section{Proof of Lemma~\ref{lemma:est_supp_dthres}}
\label{sec:proof:lemma:est_supp_dthres}

Let $J_1$ and $J_2$ denote the support of $u$ and of $v$, respectively.
For $\widetilde{J}_1 \subset \widehat{J}_1$, it suffices to show
\begin{equation}
\min_{j \in \widetilde{J}_1} \tnorm{e_j^* [\A^*(b)]}_{s_2} > \max_{j \in [n_1] \setminus J_1} \tnorm{e_j^* [\A^*(b)]}_{s_2}
\label{eq:proof:lemma:est_supp_thres:cond}
\end{equation}
where $e_j \in \bbC^{n_1}$ denotes the $j$th column of the $n_1 \times n_1$ identity matrix.
Let
\[
u_{\min} \triangleq \min_{j \in \widetilde{J}} |[u]_j|.
\]

Then, for $j \in \widetilde{J}_1$, we have
\begin{align*}
\tnorm{e_j^* [\A^*(b)]}_{s_2}
{} & \geq \norm{e_j^* [\A^*(b)] \Pi_{J_2}}_2 \\
{} & = \norm{e_j^* [(\A^*\A - \idm) (X) + X + \A^*(z)] \Pi_{J_2}}_2 \\
{} & \geq \norm{e_j^* X}_2 - \norm{e_j^* [(\A^*\A - \idm) (X)] \Pi_{J_2}}_2 - \norm{e_j^* [\A^*(z)] \Pi_{J_2}}_2 \\
{} & \geq \lambda |u[j]| - \delta \lambda - \sqrt{1+\delta} \norm{z}_2 \\
{} & \geq \lambda \left[ u_{\min} - \delta - (1+\delta) \nu \right]
\end{align*}
where the third inequality follows from Lemma~\ref{lemma:ds_bndopnormrip} and Lemma~\ref{lemma:ds_bndnoise}.

Next, for $j \in [n_1] \setminus J_1$, there exists $J' \subset [n_2]$ with $|J'| \leq s_2$ such that
\begin{align*}
\tnorm{e_j^* [\A^*(b)]}_{s_2}
{} & = \norm{e_j^* [\A^*(b)] \Pi_{J'}}_2 \\
{} & = \norm{e_j^* [(\A^*\A - \idm) (X) + X + \A^*(z)] \Pi_{J'}}_2 \\
{} & \leq \lambda \left[ \delta + (1+\delta) \nu \right]
\end{align*}
where the last step follows by Lemma~\ref{lemma:ds_bndopnormrip}, Lemma~\ref{lemma:ds_bndnoise}, and $e_j^* X = 0$.

Therefore, a sufficient condition for (\ref{eq:proof:lemma:est_supp_thres:cond}) is given by
\[
u_{\min} > 2 \delta + 2 (1+\delta) \nu.
\]

\section{Proof of Lemma~\ref{lemma:includeJtilde1rankr}}
\label{sec:pf:lemma:includeJtilde1rankr}
A sufficient condition for $\widetilde{J}_1 \subset \widehat{J}_1$ is given by
\begin{equation}
\label{eq:lemma:includeJtilde1rankr:res}
\min_{j \in \widetilde{J}_1} \tnorm{e_j^*[\A^*(b)]}_{s_2} > \min_{j \in [n_1] \setminus J_1} \tnorm{e_j^*[\A^*(b)]}_{s_2}.
\end{equation}
Therefore, it suffices to show \eqref{eq:lemma:includeJtilde1rankr:res} holds.

We first derive a lower bound on the left-hand-side of \eqref{eq:lemma:includeJtilde1rankr:res}.
Let $J_2 \subset [n_2]$ denote the set of the indices of the nonzero rows of $V$.
For any $j \in \widetilde{J}_1$, we have
\begin{align*}
\tnorm{e_j^* [\A^*(b)]}_{s_2}
{} & \geq \norm{e_j^* [\A^*(b)] \Pi_{J_2}}_2 \\
{} & = \norm{e_j^* [(\A^*\A - \id)(X) + X + \A^*(z)] \Pi_{J_2}}_2 \\
{} & \geq \norm{e_j^* X \Pi_{J_2}}_2 - \norm{e_j^* [(\A^*\A - \id)(X)] \Pi_{J_2}}_2 - \norm{e_j^* [\A^*(z)] \Pi_{J_2}}_2 \\
{} & \geq \norm{e_j^* X \Pi_{J_2}}_2 - \delta \norm{X} - \sqrt{1+\delta} \norm{z}_2 \\
{} & = \norm{e_j^* X}_2 - \delta \norm{X} - \sqrt{1+\delta} \norm{z}_2,
\end{align*}
where the fourth step follows from Lemmas~\ref{lemma:rip_spectral}, \ref{lemma:ds_bndopnormrip}, and \ref{lemma:ds_bndnoise}.

Next, we drive an upper bound on the right-hand-side of \eqref{eq:lemma:includeJtilde1rankr:res}.
For any $j \in [n_1] \setminus J_1$, there exists $J' \subset [n_2]$ with $|J'| \leq s_2$ such that
\begin{align*}
\tnorm{e_j^* [\A^*(b)]}_{s_2}
{} & = \norm{e_j^* [\A^*(b)] \Pi_{J'}}_2 \\
{} & = \norm{e_j^* [(\A^*\A - \id)(X) + X + \A^*(z)] \Pi_{J'}}_2 \\
{} & \leq \norm{e_j^* X \Pi_{J'}}_2 + \norm{e_j^* [(\A^*\A - \id)(X)] \Pi_{J'}}_2 + \norm{e_j^* [\A^*(z)] \Pi_{J'}}_2 \\
{} & \leq \delta \norm{X} + \sqrt{1+\delta} \norm{z}_2,
\end{align*}
where the last step follows from Lemmas~\ref{lemma:rip_spectral}, \ref{lemma:ds_bndopnormrip}, \ref{lemma:ds_bndnoise}, and $e_j^* X = 0$.

By these bounds and \eqref{eq:lemma:includeJtilde1rankr:res},
we get a sufficient condition for \eqref{eq:lemma:includeJtilde1rankr:res} given by
\begin{equation}
\label{eq:lemma:includeJtilde1rankr:cond2}
2 \left(\delta \norm{X} + \sqrt{1+\delta} \norm{z}_2 \right) < \min_{j \in \widetilde{J}_1} \norm{e_j^* X}_2.
\end{equation}
The right-hand-side of \eqref{eq:lemma:includeJtilde1rankr:cond2} is lower-bounded by
\begin{equation}
\label{eq:lemma:includeJtilde1rankr:cond2:rhs}
\min_{j \in \widetilde{J}_1} \norm{e_j^* X}_2
\geq \sigma_r(X) \min_{j \in \widetilde{J}_1} \norm{e_j^* U}_2
\geq \sigma_r(X) \sigma_r(U^* \Pi_{\widetilde{J}_1}).
\end{equation}
By \eqref{eq:lemma:includeJtilde1rankr:cond2} and \eqref{eq:lemma:includeJtilde1rankr:cond2:rhs} with $\norm{X} \leq \kappa \sigma_r(X)$,
we show that \eqref{eq:lemma:includeJtilde1rankr:cond2} is implied by
\begin{equation}
\label{eq:lemma:includeJtilde1rankr:cond3}
2 \kappa \left[\delta + \sqrt{1+\delta} \frac{\norm{z}_2}{\norm{X}} \right] < \sigma_r(U^* \Pi_{\widetilde{J}_1}).
\end{equation}
Furthermore, by the rank-$2r$ and doubly $(3s_1,3s_2)$-RIP of $\A$,
\begin{equation}
\label{eq:lemma:includeJtilde1rankr:cond3:noise}
\frac{\norm{z}_2}{\norm{X}}
\leq \frac{\fnorm{X}}{\norm{X}} \cdot \frac{\norm{z}_2}{\fnorm{X}}
\leq \frac{\fnorm{X}}{\norm{X}} \cdot \frac{\sqrt{1+\delta} \norm{z}_2}{\norm{A(X)}_2} \leq \sqrt{1+\delta} \nu,
\end{equation}
where the last step follows from \eqref{eq:snrcond_rankr}.
By \eqref{eq:lemma:includeJtilde1rankr:cond3:noise} and \eqref{eq:lemma:includeJtilde1rankr:cond3},
we show that \eqref{eq:lemma:includeJtilde1rankr:cond3} implies \eqref{eq:lemma:includeJtilde1rankr:cond}.
This completes the proof.

\section{Proof of Lemma~\ref{lemma:includeJtilde2rankr}}
\label{sec:pf:lemma:includeJtilde2rankr}
A sufficient condition for $\widetilde{J}_2 \subset \widehat{J}_2$ is given by
\begin{equation}
\label{eq:lemma:includeJtilde2rankr:res}
\min_{j \in \widetilde{J}_2} \norm{\Pi_{\widehat{J}_1} [\A^*(b)] e_j}_2
> \min_{j \in [n_2] \setminus J_2} \norm{\Pi_{\widehat{J}_1} [\A^*(b)] e_j}_2.
\end{equation}
Therefore, it suffices to show \eqref{eq:lemma:includeJtilde2rankr:res} holds.

We first derive a lower bound on the left-hand-side of \eqref{eq:lemma:includeJtilde2rankr:res}.
For any $j \in \widetilde{J}_2$, we have
\begin{align*}
\norm{\Pi_{\widehat{J}_1} [\A^*(b)] e_j}_2
{} & = \norm{\Pi_{\widehat{J}_1} [(\A^*\A - \id)(X) + X + \A^*(z)] e_j}_2 \\
{} & \geq \norm{\Pi_{\widehat{J}_1} X e_j}_2 - \norm{\Pi_{\widehat{J}_1} [(\A^*\A - \id)(X)] e_j}_2 - \norm{\Pi_{\widehat{J}_1} [\A^*(z)] e_j}_2 \\
{} & \geq \norm{\Pi_{\widehat{J}_1} X e_j}_2 - \delta \norm{X} - \sqrt{1+\delta} \norm{z}_2,
\end{align*}
where the last step follows from Lemmas~\ref{lemma:rip_spectral}, \ref{lemma:ds_bndopnormrip}, and \ref{lemma:ds_bndnoise}.

Next, we drive an upper bound on the right-hand-side of \eqref{eq:lemma:includeJtilde2rankr:res}.
For any $j \in [n_2] \setminus J_2$, there exists $J' \subset [n_1]$ with $|J'| \leq s_1$ such that
\begin{align*}
\norm{\Pi_{\widehat{J}_1} [\A^*(b)] e_j}_2
{} & = \norm{\Pi_{\widehat{J}_1} [(\A^*\A - \id)(X) + X + \A^*(z)] e_j}_2 \\
{} & \leq \norm{\Pi_{\widehat{J}_1} X e_j}_2 + \norm{\Pi_{\widehat{J}_1} [(\A^*\A - \id)(X)] e_j}_2 + \norm{\Pi_{\widehat{J}_1} [\A^*(z)] e_j}_2 \\
{} & \leq \delta \norm{X} + \sqrt{1+\delta} \norm{z}_2,
\end{align*}
where the last step follows from Lemmas~\ref{lemma:rip_spectral}, \ref{lemma:ds_bndopnormrip}, \ref{lemma:ds_bndnoise}, and $X e_j = 0$.

By these bounds and \eqref{eq:lemma:includeJtilde2rankr:res},
we get a sufficient condition for \eqref{eq:lemma:includeJtilde2rankr:res} given by
\begin{equation}
\label{eq:lemma:includeJtilde2rankr:cond2}
2 \left(\delta \norm{X} + \sqrt{1+\delta} \norm{z}_2 \right) \leq \min_{j \in \widetilde{J}_2} \norm{\Pi_{\widehat{J}_1} X e_j}_2.
\end{equation}
The right-hand-side of \eqref{eq:lemma:includeJtilde2rankr:cond2} is lower-bounded by
\begin{equation}
\label{eq:lemma:includeJtilde2rankr:cond2:rhs}
\min_{j \in \widetilde{J}_2} \norm{\Pi_{\widehat{J}_1} X e_j}_2
\geq \sigma_r(X) \sigma_r(U^* \Pi_{\widehat{J}_1}) \min_{j \in \widetilde{J}_2} \norm{e_j^* V}_2
\geq \sigma_r(X) \sigma_r(U^* \Pi_{\widehat{J}_1}) \sigma_r(V^* \Pi_{\widetilde{J}_2}).
\end{equation}

Similarly to the proof of Lemma~\ref{lemma:includeJtilde2rankr},
we show that \eqref{eq:lemma:includeJtilde1rankr:cond} implies \eqref{eq:lemma:includeJtilde2rankr:res}.
This completes the proof.

\section{Proof of Lemma~\ref{lemma:init_w_good_est_supp2_rankr}}
\label{sec:pf:lemma:init_w_good_est_supp2_rankr}
To prove Lemma~\ref{lemma:init_w_good_est_supp2_rankr}, we use the non-Hermitian $\sin\theta$ theorem \cite[pp. 102--103]{Wed1972perturbation}.

Let $X = U \Lambda V^*$ denote the singular value decomposition of $X$.
Then, $U$ is row $s_1$-sparse and $V$ is row $s_2$-sparse.
Let $M = \Pi_{\widehat{J}_1} X$.
Let $\widehat{M}$ be the best rank-$r$ approximation of $\Pi_{\widehat{J}_1} [\A^*(b)] \Pi_{\widehat{J}_2}$ in the spectral norm.
Let $\Delta \triangleq \Pi_{\widehat{J}_1} [\A^*(b)] \Pi_{\widehat{J}_2} - M$.
Then,
\begin{align*}
\norm{M + \Delta - \widehat{M}}
{} & = \norm{\Pi_{\widehat{J}_1} [\A^*(b)] \Pi_{\widehat{J}_2} - \widehat{M}} \\
{} & \leq \norm{ \Pi_{\widehat{J}_1} [\A^*(b)] \Pi_{\widehat{J}_2} - \Pi_{\widehat{J}_1} X \Pi_{\widehat{J}_2}} \\
{} & = \norm{\Pi_{\widehat{J}_1} [(\A^* \A - \idm)(X)] \Pi_{\widehat{J}_2} + \Pi_{\widehat{J}_1} [\A^*(z)] \Pi_{\widehat{J}_2}} \\
{} & \leq \norm{\Pi_{\widehat{J}_1} [(\A^* \A - \idm)(X)] \Pi_{\widehat{J}_2}} + \norm{\Pi_{\widehat{J}_1} [\A^*(z)] \Pi_{\widehat{J}_2}} \\
{} & \leq \delta \norm{X} + \sqrt{1+\delta} \norm{z}_2,
\end{align*}
where the last step follows from Lemmas~\ref{lemma:rip_spectral} and \ref{lemma:ds_bndnoise}.

Similarly, the difference $\Delta$ is upper-bounded in the spectral norm by
\begin{align*}
\norm{\Delta}
{} & = \norm{\Pi_{\widehat{J}_1} [\A^*(b)] \Pi_{\widehat{J}_2} - \Pi_{\widehat{J}_1} X \Pi_{\widehat{J}_2} - \Pi_{\widehat{J}_1} X \Pi_{\widehat{J}_2}^\perp} \\
{} & = \norm{\Pi_{\widehat{J}_1} [(\A^* \A - \idm)(X)] \Pi_{\widehat{J}_2} - \Pi_{\widehat{J}_1} X \Pi_{\widehat{J}_2}^\perp + \Pi_{\widehat{J}_1} [\A^*(z)] \Pi_{\widehat{J}_2}} \\
{} & \leq \norm{\Pi_{\widehat{J}_1} [(\A^* \A - \idm)(X)] \Pi_{\widehat{J}_2}} + \norm{\Pi_{\widehat{J}_1} X \Pi_{\widehat{J}_2}^\perp} + \norm{\Pi_{\widehat{J}_1} [\A^*(z)] \Pi_{\widehat{J}_2}} \\
{} & \leq \delta \norm{X} + \norm{\Pi_{\widehat{J}_1} X \Pi_{\widehat{J}_2}^\perp} + \sqrt{1+\delta} \norm{z}_2,
\end{align*}
where the last step follows from Lemma~\ref{lemma:rip_spectral} and Lemma~\ref{lemma:ds_bndnoise}.

The minimum singular value of $M$ is lower-bounded by
\[
\sigma_r(M) \geq \sigma_r(\Lambda) \sigma_r(\Pi_{\widehat{J}_1} U).
\]

Since $V_0$ consists of the $r$ singular vectors of $\widehat{M}$ and $\R(M^*) = \R(X^*)$, by Lemma~\ref{lemma:sintheta}, we have
\begin{equation}
\label{pf:lemma:init_w_good_est_supp2_rankr:res}
\begin{aligned}
\sin (\R(X^*),\R(V_0))
{} & \leq \frac{\norm{\Delta}}{\sigma_r(M) - \norm{M + \Delta - \widehat{M}}} \\
{} & \leq \frac{\delta \norm{X} + \norm{\Pi_{\widehat{J}_1}  X \Pi_{\widehat{J}_2}^\perp} + \sqrt{1+\delta} \norm{z}_2}{\sigma_r(\Lambda) \sigma_r(\Pi_{\widehat{J}_1} U) - \delta \norm{X} - \sqrt{1+\delta} \norm{z}_2} \\
{} & \leq \frac{\delta \norm{X} + \norm{\Pi_{\widehat{J}_1} U} \norm{X} \norm{\Pi_{\widehat{J}_2}^\perp V} + \norm{X} (1+\delta) \nu}{\sigma_r(\Lambda) \sigma_r(\Pi_{\widehat{J}_1} U) - \delta \norm{X} - \norm{X} (1+\delta) \nu} \\
{} & = \frac{\delta + \norm{\Pi_{\widehat{J}_1} U} \norm{\Pi_{\widehat{J}_2}^\perp V} + (1+\delta) \nu}{\sigma_r(\Pi_{\widehat{J}_1} U)/\kappa - \delta - (1+\delta) \nu},
\end{aligned}
\end{equation}
where the third inequality follows from \eqref{eq:snrcond_rankr} and the rank-$2r$ and doubly $(3s_1,3s_2)$-sparse RIP of $\A$.

Finally, we derive an upper bound on $\norm{\Pi_{\widehat{J}_2}^\perp V}$.
Since $\widetilde{J}_2 \subset \widehat{J}_2$, we have
\begin{equation}
\label{pf:lemma:init_w_good_est_supp2_rankr:ub}
\norm{\Pi_{\widehat{J}_2}^\perp V} \leq \norm{\Pi_{\widetilde{J}_2}^\perp V}.
\end{equation}
Thus we will derive an upper bound on $\norm{\Pi_{\widetilde{J}_2}^\perp V}$.
By the relations between the principal angles between the two $r$-dimensional subspaces $\R(V)$ and $\R(\Pi_{\widetilde{J}_2})$, we have
\begin{equation}
\label{pf:lemma:init_w_good_est_supp2_rankr:id}
\sigma_r^2(P_{\R(V)} P_{\R(\Pi_{\widetilde{J}_2})}) + \norm{P_{\R(V)} P_{\R(\Pi_{\widetilde{J}_2})^\perp}}^2 = 1.
\end{equation}
Since the projections are expressed as
\[
P_{\R(V)} = V V^* \quad \text{and} \quad P_{\R(\Pi_{\widetilde{J}_2})} = \Pi_{\widetilde{J}_2},
\]
\eqref{pf:lemma:init_w_good_est_supp2_rankr:id} is equivalently rewritten to
\[
\sigma_r^2(VV^* \Pi_{\widetilde{J}_2}) + \norm{VV^* \Pi_{\widetilde{J}_2}^\perp}^2 = 1.
\]
Then it follows that
\begin{equation}
\label{pf:lemma:init_w_good_est_supp2_rankr:id2}
\norm{\Pi_{\widetilde{J}_2}^\perp V}
= \norm{VV^* \Pi_{\widetilde{J}_2}^\perp}
= \sqrt{1 - \sigma_r^2(VV^* \Pi_{\widetilde{J}_2})}
= \sqrt{1 - \sigma_r^2(V^* \Pi_{\widetilde{J}_2})}.
\end{equation}
Applying \eqref{pf:lemma:init_w_good_est_supp2_rankr:ub} and \eqref{pf:lemma:init_w_good_est_supp2_rankr:id2}
to \eqref{pf:lemma:init_w_good_est_supp2_rankr:res} completes the proof.




\bibliographystyle{IEEEtran}


\end{document}